\documentclass{amsproc}

\usepackage{thmtools,thm-restate}

\usepackage{amsmath,amsthm,amssymb,latexsym}

\newtheorem{theorem}{Theorem}[section]
\newtheorem{lemma}[theorem]{Lemma}
\newtheorem{proposition}[theorem]{Proposition}

\theoremstyle{definition}
\newtheorem{definition}[theorem]{Definition}
\newtheorem{example}[theorem]{Example}
\newtheorem{remark}[theorem]{Remark}

\usepackage{longtable}

\usepackage{comment}

\numberwithin{equation}{section}

\usepackage{array}

\usepackage{color}
\usepackage{xcolor}

\definecolor{medgreen}{rgb}{0.0, 0.75, 0.0}

\definecolor{darkgreen}{rgb}{0.0, 0.4, 0.0}

\definecolor{darkyellow}{rgb}{1, 0.8, 0.0}

\usepackage{colortbl}

\newcommand{\wt}{T}
\newcommand{\wwt}{\mathcal{T}}

\newcommand{\prof}{\mathsf{prof}}
\newcommand{\anonprof}{\mathsf{anon}\mbox{-}\mathsf{prof}}

\newcommand{\gray}{gray!75}

\usepackage{tikz}
\usepackage{pgf} 
\usetikzlibrary{patterns,automata,arrows,shapes,snakes,topaths,trees,backgrounds,positioning,through,calc}

\usetikzlibrary{tikzmark,arrows.meta}
\newcolumntype{M}{>{\centering\mbox{}\vfill\arraybackslash}m{50pt}<{\vfill}}

\usepackage{mathrsfs}

\usepackage{hyperref}

\begin{document}

\title[Weakenings of expansion consistency and resoluteness in voting]{Impossibility theorems involving weakenings of  \\ expansion consistency and resoluteness in voting\footnote{Forthcoming in \textit{Mathematical Analyses of Decisions, Voting, and Games}, eds. M. A. Jones, D. McCune, and J. Wilson, Contemporary Mathematics, American Mathematical Society, 2023.}}

\author{Wesley H. Holliday}

\address{Department of Philosophy and Group in Logic and the Methodology of Science, University of California, Berkeley}
\email{wesholliday@berkeley.edu}

\author{Chase Norman}

\address{Department of Electrical Engineering and Computer Science, University of California, Berkeley}
\email{c\_@berkeley.edu}

\author{Eric Pacuit}
\address{Department of Philosophy, University of Maryland}
\email{epacuit@umd.edu}

\author{Saam Zahedian}
\address{Stanford Institute for Economic Policy Research, Stanford University}
\email{zahedian@stanford.edu}



\begin{abstract} A fundamental principle of individual rational choice is Sen's $\gamma$ axiom, also known as \textit{expansion consistency}, stating that any alternative chosen from each of two menus must be chosen from the union of the menus. Expansion consistency can also be formulated in the setting of social choice. In voting theory, it states that any candidate chosen from two fields of candidates must be chosen from the combined field of candidates. An important special case of the axiom is \textit{binary} expansion consistency, which states that any candidate chosen from an initial field of candidates and chosen in a head-to-head match with a new candidate must also be chosen when the new candidate is added to the field, thereby ruling out \textit{spoiler effects}. In this paper, we study the tension between this weakening of expansion consistency and weakenings of \textit{resoluteness}, an axiom demanding the choice of a single candidate in any election. As is well known, resoluteness is inconsistent with basic fairness conditions on social choice, namely anonymity and neutrality. Here we prove that even significant weakenings of resoluteness, which are consistent with anonymity and neutrality, are inconsistent with binary expansion consistency. The proofs make use of SAT solving, with the correctness of a SAT encoding formally verified in the Lean Theorem Prover, as well as a strategy for generalizing impossibility theorems obtained for special types of voting methods (namely majoritarian and pairwise voting methods) to impossibility theorems for arbitrary voting methods. This proof strategy may  be of independent interest for its potential applicability to other impossibility theorems in social choice.\end{abstract}

\maketitle

\tableofcontents

\section{Introduction}\label{IntroSection}

Since the inception of modern social choice theory, a key question has been to what extent principles of individual rational choice can be imposed on social choice. For his famous Impossibility Theorem, Arrow \cite{Arrow1963} assumed that society's choice function obeys the same rationality constraints as a rational individual's choice function. In combination with his axioms of Independence of Irrelevant Alternatives (IIA) and universal domain, Arrow's social rationality assumptions leave room for only undemocratic social choice functions (see, e.g., \cite[Thm.~5]{Wilson1972}). In fact, in the presence of IIA and universal domain, even much weaker social rationality assumptions still lead to a slew of Arrow-style impossibility theorems (see, e.g., \cite[Ch.~3]{Suzumura1983}). More important for our purposes here is that even without IIA, principles of individual rational choice can cause problems in social choice. 

On the conception of individual rational choice as a matter of \textit{maximization}, a rational agent is representable as choosing from any menu of options the maximal elements according to some underlying relation of preference between options. As shown by Sen \cite{Sen1971}, assuming finitely many options, an individual's choices can be represented in this way if and only if their choices satisfy two axioms, known as $\alpha$ and $\gamma$, or \textit{contraction consistency} and \textit{expansion consistency}.

When applied to social choice---in particular to voting---contraction consistency states that a candidate chosen from some field of candidates must still be chosen in any smaller field in which the candidate remains, holding fixed the preferences of the voters. For example, if $b$ is among the winning candidates in an election when the field is $\{a,b,c\}$, then $b$ must still be among the winning candidates when the field is $\{a,b\}$. Yet this contradicts basic principles of social choice---in particular, anonymity and neutrality (see Section \ref{IndividualToSocial}) plus majority rule in two-candidate elections. Consider an election in which three voters submit the following rankings of the candidates: $abc$, $bca$, $cab$. In this case, anonymity and neutrality require that $a$, $b$, and $c$ are all tied for winning the election. Then by contraction consistency, $b$ must be among the winners when the field is $\{a,b\}$. But this contradicts majority rule in two-candidate elections, since two voters rank $a$ over $b$, while only one ranks $b$ above $a$.\footnote{\label{AlphaNote}In fact, the same style of reasoning shows that contraction consistency is inconsistent with anonymity, neutrality, and the requirement that a single winner be chosen in any two-candidate election with an odd number of voters. For anonymity, neutrality, and the requirement together imply that for the election described in the main text, either majority rule is applied to its two-candidate subelections or minority rule is applied to its two-candidate subelections. In the first case, we are done by the argument in the main text. In the second case, contraction consistency implies that $b$ must be among the winners when the field is $\{b,c\}$, which contradicts minority rule.} Thus, for one of the central problems in voting theory---namely, devising anonymous and neutral voting methods that agree with majority-rule in two candidate elections---contraction consistency is unacceptable. 

By contrast, expansion consistency appears more benign. When applied to voting, expansion consistency states that any candidate chosen from two fields of candidates must be chosen from the combined field of candidates, holding fixed the preferences of the voters. For example, if $a$ is chosen from $\{a,b,c\}$ and from $\{a,d,e\}$, then $a$ must also be chosen from $\{a,b,c,d,e\}$. Now there are anonymous and neutral voting methods that agree with majority-rule in two candidate elections while also satisfying expansion consistency, as studied systematically in \cite{Brandt2022}. Examples include Split Cycle \cite{HP2020b,HP2021}, Top Cycle \cite{Smith1973,Schwartz1986}, Uncovered Set \cite{Miller1980,Duggan2013}, and Weighted Covering \cite{Dutta1999,Fernandez2018}. 

However, we will show that expansion consistency still imposes significant constraints on social choice. A striking feature of the voting methods satisfying expansion consistency listed above is that they are ``less resolute'' than some other well-known voting methods that violate expansion consistency, such as  Beat Path \cite{Schulze2011,Schulze2018}, Minimax \cite{Simpson1969,Kramer1977}, and Ranked Pairs \cite{Tideman1987,ZavistTideman1989}. A voting method is \textit{resolute} if it selects a unique winner in any election. While resoluteness is inconsistent with anonymity and neutrality, as shown by considering the three-voter election above, one can weaken resoluteness to avoid such an inconsistency. For example, one may require a unique winner only in elections without the kind of symmetries that force a tie in the election above (technically, in any election that is \textit{uniquely weighted}, as defined in Section \ref{IndividualToSocial}). Such weaker forms of resoluteness are satisfied by some of the voting methods  that violate expansion consistency but violated by the voting methods  that satisfy expansion consistency.\footnote{The two-stage majoritarian rule proposed by Horan and Spurmont \cite{Horan2022} satisfies expansion consistency and resoluteness but at the expense of violating neutrality. In this paper, we are interested only in anonymous and neutral methods.} In this paper, we will prove that this is no coincidence: we identify weakenings of resoluteness that are inconsistent with expansion consistency.

In response, one may turn to weaker forms of expansion consistency, such as \textit{binary} expansion consistency. Binary expansion consistency states that any candidate chosen from an initial field of candidates and (uniquely) chosen in a head-to-head match with a new candidate must also be chosen when the new candidate is added to the field. Some voting methods that violate expansion consistency still satisfy binary expansion consistency, which is good enough to rule out \textit{spoiler effects} in elections  (see Section \ref{IndividualToSocial}). However, we will prove that the relevant weakenings of resoluteness are even inconsistent with binary expansion consistency. Thus, there appears to be a basic tension between axioms in the spirit of expansion consistency and axioms in the spirit of resoluteness, assuming anonymity and neutrality.

The rest of the paper is organized as follows. In Section \ref{ImpossSection}, we motivate and precisely state our main impossibility theorems, and we prove special cases of the theorems. The rest of the paper is devoted to proving the theorems in full generality. Our strategy to do so is outlined in Section \ref{ProofStrategy} and carried out in Sections \ref{RepresentationSection}--\ref{PreservationSection}. We conclude in Section \ref{Conclusion} with some suggestions for future research. Appendices \ref{SATmajoritarian}-\ref{Verification} contain details of computer-aided techniques used in the proofs of our theorems.

\section{The impossibility theorems}\label{ImpossSection}

\subsection{Origins in Individual Choice}\label{Origins} In this paper, we consider impossibility theorems arising from the attempt to transfer axioms on individual choice to social choice. We briefly recall the necessary background on individual choice. Given a finite set $X$, a \textit{choice function} on $X$ is a map $C:\wp(X)\setminus\{\varnothing\}\to\wp(X) \setminus\{\varnothing\}$ such that for all $S\in \wp(X)\setminus\{\varnothing\}$, we have $\varnothing\neq C(S)\subseteq S$. The intuitive interpretation is that when offered a choice from the menu $S$, the agent considers all the elements in $C(S)$ choiceworthy. To force the choice of a single alternative, one can impose the following axiom on $C$:
\[\mbox{\textbf{resoluteness}: $|C(S)|=1$ for all $S\in \wp(X)\setminus\{\varnothing\}$}.\]

 The classic question concerning choice functions is a representation question: under what conditions can we view the agent as if she chooses from any given menu the greatest elements according to a \textit{preference relation}? Where $R$ is a binary relation on $X$, we  say that $C$ is \textit{represented by $R$} if for all $S\in \wp(X)\setminus\{\varnothing\}$,
\[C(S)=\{x\in S\mid \mbox{for all }y\in S,\, xRy\}.\]

\noindent Sen \cite{Sen1971} identified the following axioms on $C$ that are necessary and sufficient for representability by a binary relation:
\begin{eqnarray*} 
&&\mbox{$\alpha$: if $S\subseteq S'$, then $S\cap C(S')\subseteq C(S)$};\\ 
&&\mbox{$\gamma$: $C(S)\cap C(S')\subseteq C(S\cup S')$}.
\end{eqnarray*}

\begin{theorem}[Sen] Let $C$ be a choice function on a finite set $X$.
\begin{enumerate}
\item $C$ is representable by a binary relation on $X$ iff $C$ satisfies $\alpha$ and $\gamma$.
\item if $C$ is resolute, then the following are equivalent:
\begin{enumerate}
\item $C$ is representable by a binary relation on $X$;
\item $C$ is representable by a linear order on $X$;
\item $C$ satisfies $\alpha$.
\end{enumerate}
\end{enumerate}
\end{theorem}
\noindent One can also characterize choice functions that satisfy only $\alpha$ (resp.~$\gamma$) in terms of different notions of representability by a family of relations instead of a single relation (see \cite[Thm.~6]{Duggan2019} for $\alpha$ and \cite[Thm.~6B]{Tyson2008}, \cite[Thm.~7]{Duggan2019}, and \cite[Thm.~1]{Brandt2022} for $\gamma$).\footnote{One can also consider a notion of representability by a binary relation between subsets of $X$, rather than between elements of $X$, and there are analogues of $\alpha$ and $\gamma$ in this setting of set-based rationalizability \cite{Brandt2011}.}

In this paper, we shall see that there are difficulties transferring even weaker versions of the above axioms from individual choice to social choice. First, consider the following weakening of $\gamma$:
\[\mbox{\textbf{binary $\boldsymbol{\gamma}$}: if $x\in C(A)$ and $C(\{x,y\})=\{x\}$, then $x\in C(A\cup\{y\})$.}\] 

\noindent As Bordes \cite{Bordes1983} puts it (in the context of choice between candidates in voting, to which we turn in Section \ref{IndividualToSocial}), ``if $x$ is a winner in $A$\dots one cannot turn $x$ into a loser by introducing new alternatives to which $x$ does not lose in duels'' (p.~125). In fact, Bordes' formulation suggests replacing $C(\{x,y\})=\{x\}$ by $x\in  C(\{x,y\})$, resulting in a slightly stronger axiom, but we will work with the weaker version, as impossibility theorems with weaker axioms are stronger results. Binary $\gamma$ is an obvious consequence of $\gamma$, setting $S=A$ and $S'=\{x,y\}$.

We will also consider a weakening of $\alpha$ with the same antecedent as binary $\gamma$:
\[\mbox{\textbf{binary $\boldsymbol{\alpha}$}: if $x\in C(A)$ and $C(\{x,y\})=\{x\}$, then $C(A\cup\{y\})\subseteq C(A)$.}\]

\noindent In other words, adding to the menu an alternative that loses to an initially chosen alternative does not lead to any new choices from the menu. That binary $\alpha$ follows from $\alpha$ is not quite as immediate as the analogous implication for $\gamma$---but almost.

\begin{proposition} Any choice function satisfying $\alpha$ also satisfies binary $\alpha$.
\end{proposition}
\begin{proof}Suppose $x\in C(A)$ and $C(\{x,y\})=\{x\}$. Since $\{x,y\}\subseteq A\cup \{y\}$, we have $\{x,y\}\cap C(A\cup \{y\})\subseteq C(\{x,y\})$ by $\alpha$. The since $y\not\in C(\{x,y\})$, we conclude $y\not\in C(A\cup \{y\})$. Hence $C(A\cup\{y\})\subseteq A$. By $\alpha$ again, $A\cap C(A\cup \{y\})\subseteq C(A)$, which with the previous sentence implies $C(A\cup \{y\})\subseteq C(A)$.
\end{proof}

Finally, we will consider a weakening of resoluteness. This weakening says that while some non-singleton choice sets may be allowed, we should never enlarge the choice set when adding an alternative that loses to some initially chosen alternative: 
\[\mbox{\textbf{$\boldsymbol{\alpha}$-resoluteness}: if $x\in C(A)$ and $C(\{x,y\})=\{x\}$, then $|C(A\cup\{y\})|\leq |C(A)|$}.\]

\noindent Clearly $\alpha$-resoluteness is implied by binary $\alpha$ and by resoluteness.

\subsection{From Individual Choice to Social Choice}\label{IndividualToSocial} We now explain how the axioms on individual choice from Section \ref{Origins} can be applied to social choice.

Fix infinite sets $\mathcal{X}$ and $\mathcal{V}$ of \textit{candidates} and \textit{voters}, respectively. A \textit{profile} $\mathbf{P}$ is a function $\mathbf{P}:V\to\mathcal{L}(X)$ where $V$ is a finite subset of $\mathcal{V}$, $X$ is a finite subset of $\mathcal{X}$, and $\mathcal{L}(X)$ is the set of all strict linear orders on $X$.\footnote{This linearity assumption strengthens our impossibility theorems. That is, our impossibility theorems for voting methods defined only on profiles of linear orders immediately imply corresponding impossibility theorems for voting methods defined on profiles of weak orders; for if there were a voting method satisfying the relevant axioms on profiles of weak orders, then restricting the voting method to profiles of linear orders would yield the kind of voting method we will prove does not exist. By contrast, proving that no voting method defined on profiles of weak orders satisfies certain axioms does not immediately imply that no voting method defined only on linear orders satisfies the axioms. Selecting winners in profiles of weak orders may be more difficult.} We write $V(\mathbf{P})$ for $V$ and $X(\mathbf{P})$ for $X$.  Let $\prof$ be the set of all profiles. Given a set $A$ of candidates, the profile $\mathbf{P}|_{A}$ assigns to each voter their ranking of the candidates in $\mathbf{P}$ restricted to $A$, and for a particular candidate $y$, we write $\mathbf{P}_{-y}=\mathbf{P}|_{X(\mathbf{P})\setminus\{y\}}$. 

Given two candidates $x,y$ in a profile $\mathbf{P}$, the \textit{margin of $x$ over $y$ in $\mathbf{P}$} is the number of voters who rank $x$ above $y$ in $\mathbf{P}$ minus the number who rank $y$ above $x$ in $\mathbf{P}$; if this margin is positive, we say that $x$ is \textit{majority preferred to $y$ in $\mathbf{P}$}. The \textit{majority graph of $\mathbf{P}$}, denoted $M(\mathbf{P})$, is the directed graph whose set of vertices is $X(\mathbf{P})$ with an edge from $x$ to $y$ if $x$ is majority preferred to $y$. The \textit{margin graph of $\mathbf{P}$}, denoted $\mathcal{M}(\mathbf{P})$, is the weighted directed graph obtained from $M(\mathbf{P})$ by weighting the edge from $x$ to $y$ by the margin of $x$ over $y$ in $\mathbf{P}$. See Figure~\ref{GraphFig}.

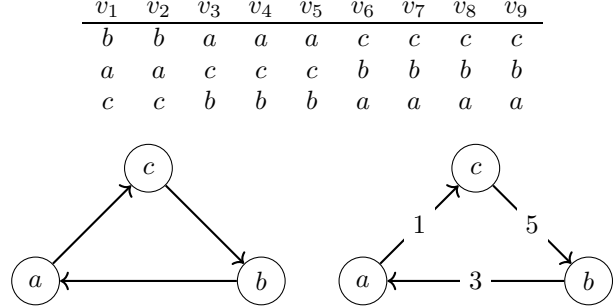
\begin{figure}[ht]
\begin{center}
\begin{tabular}{ccccccccc}
$v_1$ & $v_2$ & $v_3$ & $v_4$ & $v_5$ & $v_6$ & $v_7$ & $v_8$ & $v_9$  \\\hline
$b$ & $b$ & $a$ & $a$ & $a$ &  $c$ &  $c$&  $c$&  $c$  \\
$a$ & $a$ &  $c$ & $c$ &$c$ & $b$& $b$& $b$& $b$ \\
$c$ & $c$ &  $b$ & $b$ &$b$ &  $a$&  $a$&  $a$&  $a$ \\
\end{tabular}
\end{center}
\medskip
\begin{center}
\begin{tikzpicture}

\node[circle,draw, minimum width=0.25in] at (0,0) (a) {$a$}; 
\node[circle,draw,minimum width=0.25in] at (3,0) (b) {$b$}; 
\node[circle,draw,minimum width=0.25in] at (1.5,1.5) (c) {$c$};

\path[->,draw,thick] (c) to node {$$} (b);
\path[->,draw,thick] (b) to node {$$} (a);
\path[->,draw,thick] (a) to node {$$} (c);

\end{tikzpicture}\qquad \begin{tikzpicture}

\node[circle,draw, minimum width=0.25in] at (0,0) (a) {$a$}; 
\node[circle,draw,minimum width=0.25in] at (3,0) (b) {$b$}; 
\node[circle,draw,minimum width=0.25in] at (1.5,1.5) (c) {$c$};

\path[->,draw,thick] (c) to node[fill=white] {$5$} (b);
\path[->,draw,thick] (b) to node[fill=white] {$3$} (a);
\path[->,draw,thick] (a) to node[fill=white] {$1$} (c);

\end{tikzpicture}
\end{center}
\caption{A preference profile $\mathbf{P}$ (top), its majority graph $M(\mathbf{P})$ (bottom left), and its margin graph $\mathcal{M}(\mathbf{P})$ (bottom right).}\label{GraphFig}
\end{figure}

\begin{remark}\label{Parity}Since voters submit linear orders of the candidates, it follows that all weights in $\mathcal{M}(\mathbf{P})$ must have the same parity, which is even if there are any distinct candidates with a margin of zero between them.\end{remark}

\noindent McGarvey \cite{McGarvey1953} proved that any asymmetric directed graph is isomorphic to the majority graph of a profile. Debord \cite{Debord1987} proved that any asymmetric directed graph with weighted edges satisfying the parity constraint in Remark \ref{Parity} is isomorphic to the margin graph of a profile. We will review the proofs of these results in Section~\ref{McGarveyDebord}.

A \textit{voting method} $F$ assigns to each profile $\mathbf{P}$ in some nonempty set $\mathrm{dom}(F)$ of profiles a nonempty set $F(\mathbf{P})$ of candidates who are tied for winning the election. All of the axioms on individual choice functions from the previous section corresponds to axioms on voting methods as follows, where all profiles are assumed to belong to $\mathrm{dom}(F)$:\\

\noindent\textbf{resoluteness}: $|F(\mathbf{P})| = 1$ for all profiles $\mathbf{P}\in \mathrm{dom}(F)$;\\

\noindent\textbf{$\alpha$}: if $Y\subseteq X(\mathbf{P})$, then $Y\cap F(\mathbf{P})\subseteq F(\mathbf{P}|_{Y})$;\\

\noindent\textbf{$\gamma$}: if $X(\mathbf{P}) =Y\cup Z$, then $ F(\mathbf{P}|_{Y})\cap F(\mathbf{P}|_{Z})\subseteq F(\mathbf{P})$;\\

\noindent \textbf{binary $\gamma$}: if $x\in F(\mathbf{P}_{-y})$ and $F(\mathbf{P}|_{ \{x,y\}})=\{x\}$, then $x\in F(\mathbf{P})$;\footnote{Note that binary $\gamma$ for choice functions can be written equivalently as: if $x\in C(B\setminus \{y\})$ and $C(\{x,y\})=\{x\}$, then $x\in C(B)$.}\\

\noindent \textbf{binary $\alpha$}: if $x\in F(\mathbf{P}_{-y})$ and $F(\mathbf{P}|_{ \{x,y\}})=\{x\}$, then $F(\mathbf{P})\subseteq F(\mathbf{P}_{-y})$;\\

\noindent \textbf{$\alpha$-resoluteness}:  if $x\in F(\mathbf{P}_{-y})$ and $F(\mathbf{P}|_{ \{x,y\}})=\{x\}$, then $|F(\mathbf{P})|\leq | F(\mathbf{P}_{-y})|$. \\

Let us consider three examples illustrating how  binary $\gamma$, binary $\alpha$, and $\alpha$-resoluteness compare in strength to $\gamma$ and $\alpha$. We begin by showing that binary $\gamma$ is weaker than $\gamma$.

\begin{example}\label{Banks} The Banks voting method \cite{Banks1985} is defined as follows.\footnote{We thank Felix Brandt for providing this example.} Say that a \textit{chain} in $M(\mathbf{P})$ is a subset of $X(\mathbf{P})$ linearly ordered by the edge relation of $M(\mathbf{P})$. Then $a\in Banks(\mathbf{P})$ if $a$ is the maximum element with respect to the edge relation of some maximal chain in $M(\mathbf{P})$. For convenience, we repeat here a proof that Banks satisfies binary $\gamma$ but not $\gamma$ from \cite{HP2020b}. 

For binary $\gamma$, suppose $a\in Banks(\mathbf{P}_{-b})$ and $a\in Banks(\mathbf{P}|_{ \{a,b\}})$. Hence  $a$ is the maximum element of some maximal chain $R$ in $M(\mathbf{P}_{-b})$. Let $R'$ be a maximal chain in $M(\mathbf{P})$ with $R\subseteq R'$. Since $a\in Banks(\mathbf{P}|_{ \{a,b\}})$, $b$ is not the maximum  of $R'$, so $a$ is the maximum  of $R'$, so $a\in Banks(\mathbf{P})$. Note this establishes the strengthening of binary $\gamma$ with $x\in F(\mathbf{P}|_{ \{x,y\}})$ in place of $F(\mathbf{P}|_{ \{x,y\}})=\{x\}$.

For a violation of $\gamma$, consider a profile $\mathbf{P}$ with the majority graph in Figure~\ref{BanksGamma}, which exists by McGarvey's \cite{McGarvey1953} theorem (Theorem \ref{McGarveyThm}.\ref{McGarveyThm1} below). The maximal chains of $M(\mathbf{P})$, ordered by the edge relation, are as follows:
\begin{eqnarray*}
&&\mbox{$(c, a, g, b)$\quad $(c, g, d, b)$\quad $(c, e, a, g)$\quad $(c, d, b, e)$\quad $(c, d, e, a)$}\\
&&\mbox{$(f,c,g)$\quad$(e, a, f, g)$\quad $(d, b, e, f)$\quad $(d, a, b, f)$\quad$(d, e, a, f)$.}
\end{eqnarray*}
Then  $Banks(\mathbf{P}|_{ \{a,b,c,e,f\}}) \cap Banks(\mathbf{P}|_{ \{a, d,f,g\}}) \ni a \not\in Banks(\mathbf{P})=\{c, d, e, f\}$.

\begin{figure}[h]
\begin{center}
\begin{minipage}{2.5in}
\begin{tikzpicture}
\node[circle,draw,minimum width=0.25in] at (-1,2)  (a) {$a$}; 
\node[circle,draw,minimum width=0.25in] at (2,2)  (b) {$b$}; 
\node[circle,draw,minimum width=0.25in] at (4,0)  (c) {$c$};
\node[circle,draw,minimum width=0.25in] at (2,-2) (d) {$d$}; 
\node[circle,draw,minimum width=0.25in] at (-1,-2) (e) {$e$}; 
\node[circle,draw,minimum width=0.25in] at (-3,-1)  (f) {$f$};
\node[circle,draw,minimum width=0.25in] at (-3,1)  (g) {$g$};

\path[->,draw,thick] (a) to node {} (b);
\path[<-,draw,thick] (a) to node {} (c);
\path[<-,draw,thick] (a) to node {} (d);
\path[<-,draw,thick] (a) to node {} (e);
\path[->,draw,thick] (a) to node {} (f);
\path[->,draw,thick] (a) to node {} (g);

\path[<-,draw,thick] (b) to node {} (c);
\path[<-,draw,thick] (b) to node {} (d);
\path[->,draw,thick] (b) to node {} (e);
\path[->,draw,thick] (b) to node {} (f);
\path[<-,draw,thick] (b) to node {} (g);

\path[->,draw,thick] (c) to node {} (d);
\path[->,draw,thick] (c) to node {} (e);
\path[<-,draw,thick] (c) to node {} (f);
\path[->,draw,thick] (c) to node {} (g);

\path[->,draw,thick] (d) to node {} (e);
\path[->,draw,thick] (d) to node {} (f);
\path[<-,draw,thick] (d) to node {} (g);

\path[->,draw,thick] (e) to node {} (f);
\path[->,draw,thick] (e) to node {} (g);

\path[->,draw,thick] (f) to node {} (g);

\end{tikzpicture}
\end{minipage}
\end{center}
\caption{}\label{BanksGamma}
\end{figure}

Yet we cannot use Banks to show that binary $\alpha$ is weaker than $\alpha$, as it violates even $\alpha$-resoluteness, as shown by a profile with the following majority graph in Figure \ref{ViolatesAlphaResoluteness}. Note that $Banks(\mathbf{P}_{-y})=\{x,z,w\}$, $Banks(\mathbf{P}|_{ \{x,y\}})=\{x\}$, and $Banks(\mathbf{P})=\{x,y,z,w\}$, violating $\alpha$-resoluteness.\end{example}

\begin{figure}[h]
\begin{center}
\begin{tikzpicture}

\node[circle,draw, minimum width=0.25in] at (0,0) (a) {$x$}; 
\node[circle,draw,minimum width=0.25in] at (1.5,1.5) (b) {$y$}; 
\node[circle,draw,minimum width=0.25in] at (3,0) (c) {$z$}; 
\node[circle,draw,minimum width=0.25in] at (1.5,-1.5) (d) {$w$}; 

\path[->,draw,thick,black] (b) to   (c);
\path[->,draw,thick,black] (c) to (a);
\path[->,draw,thick,black] (a) to    (b);
\path[<-,draw, thick,black] (b) to     (d);
\path[<-,draw, thick,black] (d) to    (a);
\path[->,draw,thick,black] (d) to    (c);

\end{tikzpicture}
\end{center}
\caption{}\label{ViolatesAlphaResoluteness}\end{figure}

The following example shows both that binary $\gamma$ is weaker than $\gamma$ and that binary $\alpha$ is weaker than $\alpha$ for voting methods.

\begin{example}\label{ParetoScoring} A candidate $x$ \textit{Pareto dominates} a candidate $y$ in profile $\mathbf{P}$ if all voters rank $x$ above $y$. Say that the \textit{Pareto score} of a candidate $x$ is the number of candidates whom $x$ Pareto dominates. Then define the Pareto scoring method as follows: $x\in PS(\mathbf{P})$ if $x$ is Pareto undominated and $x$ has maximal Pareto score among candidates who are Pareto undominated. To see that Pareto scoring satisfies binary $\gamma$ and binary $\alpha$, suppose $x\in PS(\mathbf{P}_{-y})$ and $PS(\mathbf{P}|_{ \{x,y\}})=\{x\}$. The latter implies that $x$ Pareto dominates $y$. Hence $y\not\in PS(\mathbf{P})$, and $x$'s Pareto score in $\mathbf{P}$ is one greater than $x$'s Pareto score in $\mathbf{P}_{-y}$. Since the most a candidate's Pareto score can increase from $\mathbf{P}_{-y}$ to $\mathbf{P}$ is one, it follows that $x\in PS(\mathbf{P})$ and $PS(\mathbf{P})\subseteq PS(\mathbf{P}_{-y})$. However, Pareto scoring violates $\gamma$ and $\alpha$; in fact, it violates even the weakenings of $\gamma$ and $\alpha$ obtained from binary $\gamma$ and binary $\alpha$ by replacing $F(\mathbf{P}|_{ \{x,y\}})=\{x\}$ with $x\in F(\mathbf{P}|_{ \{x,y\}})$, as shown by the profile in Figure \ref{ParetoProf}. Note that $x$ and $y$ are Pareto undominated, while $y$ (but not $x$) Pareto dominates $z$. Then $x\in PS(\mathbf{P}_{-y})$ and $x\in PS(\mathbf{P}|_{ \{x,y\}})$, but $PS(\mathbf{P})=\{y\}$.\end{example}
\begin{figure}[h]
\begin{center}
\begin{tabular}{ccc}
$1$ & $1$    &  $1$  \\\hline
$y$ & $y$  &  $x$  \\
$x$ & $z$  &  $y$  \\
$z$ & $x$  &  $z$  
\end{tabular}
\end{center}
\caption{}\label{ParetoProf}
\end{figure}

Finally, we show that $\alpha$-resoluteness is weaker than binary $\alpha$ over a restricted domain of profiles.

\begin{example}\label{CopelandEx} Consider profiles with an odd number of voters and at most four candidates. With respect to this restricted domain, the Copeland voting method \cite{Copeland1951} satisfies $\alpha$-resoluteness, as well as binary $\gamma$, but violates binary $\alpha$. The \textit{Copeland score} of a candidate $x$ in a profile $\mathbf{P}$ is defined using the majority graph $M(\mathbf{P})$ as the outdegree of $x$ minus the indegree of $x$, i.e., the number of head-to-head wins minus the number of head-to-head losses. The Copeland voting method selects as the winner in $\mathbf{P}$ the candidates with maximal Copeland score. Since this method is anonymous and neutral (see below for definitions), it violates $\alpha$ by the argument given in Section \ref{IntroSection}. To see that it satisfies $\alpha$-resoluteness, as well as binary $\gamma$,  over the restricted domain, suppose $x\in Copeland(\mathbf{P}_{-y})$ and $Copeland(\mathbf{P}|_{\{x,y\}})=\{x\}$. We consider two cases. First, suppose $x$ has an arrow to every other candidate in $M(\mathbf{P}_{-y})$. Then since $Copeland(\mathbf{P}|_{\{x,y\}})=\{x\}$, $x$ has an arrow to every other candidate in $M(\mathbf{P})$, so $Copeland(\mathbf{P})=\{x\}$. Now suppose that $x$ has an incoming arrow from some other candidate in $M(\mathbf{P}_{-y})$. Then since $x\in Copeland(\mathbf{P}_{-y})$, it follows that $M(\mathbf{P}_{-y})$ is a cycle $x\to a \to b \to x$, so each candidate has a Copeland score of 0 and hence $Copeland(\mathbf{P}_{-y})=\{x,a,b\}$.  Then $M(\mathbf{P})$ is of the form in Figure \ref{CopelandFig}, where the dashed lines represent arrows that could point in either direction.
\begin{figure}[h]
\begin{center}
\begin{tikzpicture}

\node[circle,fill = medgreen!50,draw, minimum width=0.25in] at (0,0) (a) {$x$}; 
\node[circle,draw,minimum width=0.25in] at (1.5,1.5) (b) {$a$}; 
\node[circle,draw,minimum width=0.25in] at (3,0) (c) {$b$}; 
\node[circle,draw,minimum width=0.25in] at (1.5,-1.5) (d) {$y$}; 

\path[->,draw,thick,black] (b) to   (c);
\path[->,draw,thick,black] (c) to   (a);
\path[->,draw,thick,black] (a) to (b);
\path[-,draw, dashed, thick,black] (b) to (d);
\path[<-,draw, thick,black] (d) to  (a);
\path[-,draw,dashed, thick,black] (d) to  (c);

\end{tikzpicture}
\end{center}
\caption{}\label{CopelandFig}
\end{figure}

\noindent No matter the choice of direction for the arrows replacing the dashed lines in Figure~\ref{CopelandFig}, $x\in Copeland(\mathbf{P})$ and $|Copeland(\mathbf{P})|\leq 3$. This shows that Copeland satisfies binary $\gamma$ and $\alpha$-resoluteness for profiles in the restricted domain. However, it is easy to see that Copeland violates these axioms when moving from 4 to 5 candidates.
\end{example}

We will be interested in the interaction of the above axioms borrowed from individual choice with some standard axioms on voting methods that are distinctively social-choice theoretic. In particular, we are interested in basic conditions of fairness with respect to voters and candidates. For the fairness condition with respect to voters, given $i,j\in\mathcal{V}$, let $\tau^{i\leftrightarrows j}$ be the permutation of $\mathcal{V}$ that swaps only $i$ and $j$. Given a profile  $\mathbf{P}$, let $\mathbf{P}^{i\leftrightarrows j}$ be the profile with $V(\mathbf{P}^{i\leftrightarrows j})=\tau^{i\leftrightarrows j} [V(\mathbf{P})]$ and $X(\mathbf{P}^{i\leftrightarrows j})=X(\mathbf{P})$ such that for $k\in V(\mathbf{P}^{i\leftrightarrows j})$, we have $\mathbf{P}^{i\leftrightarrows j}(k)= \mathbf{P}(\tau^{i\leftrightarrows j}(k))$.\\

\noindent \textbf{anonymity}: for any profile $\mathbf{P}\in\mathrm{dom}(F)$ and $i,j\in\mathcal{V}$, if $\mathbf{P}^{i\leftrightarrows j}\in\mathrm{dom}(F)$, then $F(\mathbf{P}^{i\leftrightarrows j})=F(\mathbf{P})$.\footnote{In our variable-election setting, this axiom could be called \textit{global anonymity} to distinguish it from an axiom of \textit{local anonymity} that only quantifies over $i,j\in V(\mathbf{P})$. Local anonymity permits swapping the ballots assigned to two voters in $\mathbf{P}$ but does not permit changing the set of voters.}\\

\noindent One can equivalently state anonymity with an arbitrary permutation $\tau$ of $\mathcal{V}$, since the transpositions $\tau^{i\leftrightarrows j}$ generate the full permutation group. Another equivalent statement is that if two profiles $\mathbf{P}$ and $\mathbf{P}'$ are such that for each linear order $L$, the same number of voters submit $L$ as their ballot in $\mathbf{P}$ as in $\mathbf{P}'$, then $F(\mathbf{P})=F(\mathbf{P}')$.

For the fairness condition with respect to candidates, given $a,b\in\mathcal{X}$, let $\pi_{a\leftrightarrows b}$ be the permutation of $\mathcal{X}$ that swaps only $a$ and $b$. Given a linear order $L$ of some finite subset of $\mathcal{X}$, define $L_{a\leftrightarrows b}$ such that for all $x,y\in\mathcal{X}$, $(x,y)\in L_{a\leftrightarrows b}$ iff  $(\pi_{a\leftrightarrows b}(x),\pi_{a\leftrightarrows b}(y))\in L$. Given $X\subseteq\mathcal{X}$, let $X_{a\leftrightarrows b}=\pi_{a\leftrightarrows b}[X]$. Given a profile  $\mathbf{P}$, let $\mathbf{P}_{a\leftrightarrows b}$ be the profile with $V(\mathbf{P}_{a\leftrightarrows b})=V(\mathbf{P})$ and $X(\mathbf{P}_{a\leftrightarrows b})= X(\mathbf{P})_{a\leftrightarrows b}$ such that for $i\in V(\mathbf{P}_{a\leftrightarrows b})$, we have $\mathbf{P}_{a\leftrightarrows b}(i)= \mathbf{P}(i)_{a\leftrightarrows b}$. \\

\noindent \textbf{neutrality}: for any profile $\mathbf{P}\in\mathrm{dom}(F)$ and $a,b\in\mathcal{X}$, if $\mathbf{P}_{a\leftrightarrows b}\in\mathrm{dom}(F)$, then $F(\mathbf{P}_{a\leftrightarrows b})=F(\mathbf{P})_{a\leftrightarrows b}$.\footnote{In our variable-election setting, this axiom could be called \textit{global neutrality} to distinguish it from an axiom of \textit{local neutrality} that only quantifies over $a,b\in X(\mathbf{P})$. Local neutrality permits swapping the positions of two candidates in $\mathbf{P}$ but does not permit changing the set of candidates.}\\

\noindent One can equivalently state neutrality with an arbitrary permutation $\pi$ of $\mathcal{X}$. 

One of the most basic facts in voting theory is that there is no voting method satisfying  anonymity, neutrality, and resoluteness on all profiles. This can be seen with the three-candidate profile in Section \ref{IntroSection}, for which anonymity and neutrality force the selection of all three candidates. In fact, the problem already arises with two candidates:\\

\noindent \textbf{binary resoluteness}: $|F(\mathbf{P})| = 1$ for all profiles $\mathbf{P}$ such that $|X(\mathbf{P})|=2$.\\

\noindent Consider a profile $\mathbf{P}$ with only two candidates $x,y$ and an even number of voters, half of whom rank $x$ above $y$ and half of whom rank $y$ above $x$. Then anonymity and neutrality imply $F(\mathbf{P})=\{x,y\}$, violating resoluteness.

 \begin{proposition}\label{NoResolute} There is no voting method defined on all two-candidate profiles satisfying anonymity, neutrality, and binary resoluteness.
 \end{proposition}
 
A natural response is to weaken binary resoluteness to the following:\\

\noindent \textbf{binary quasi-resoluteness}: $|F(\mathbf{P})| = 1$ for all profiles $\mathbf{P}$ such that $|X(\mathbf{P})|=2$ and for $x,y\in X(\mathbf{P})$ with $x\neq y$, the margin of $x$ over $y$ is nonzero.\\

\noindent Unlike binary resoluteness, binary quasi-resoluteness is consistent with anonymity and neutrality. Indeed, almost all standard voting methods satisfy binary quasi-resoluteness, including, e.g., the Plurality voting method that selects in a given profile the candidates with the most first-place votes. 

In this paper, we are interested in impossibility theorems that arise by adding axioms to the initial package of anonymity, neutrality, and binary quasi-resoluteness. First, adding binary $\alpha$ to this package results in an inconsistency by the argument in Section \ref{IntroSection} (see Footnote \ref{AlphaNote}). Thus, we move on to  binary $\gamma$. As argued in \cite{HP2020b}, binary $\gamma$ can be viewed as an axiom that mitigates \textit{spoiler effects} in elections.\footnote{The axiom called \textit{stability for winners} in \cite{HP2020b} is almost the same as binary $\gamma$, but with `$F(\mathbf{P}|_{ \{x,y\}})=\{x\}$' in the statement of binary $\gamma$ replaced by `the margin of $x$ over $y$ in $\mathbf{P}$ is positive'. Assuming $F$ reduces to majority voting in two-candidate profiles, $F$ satisfying binary~$\gamma$ is equivalent to $F$ satisfying stability for winners.} Consider one of the most famous U.S. elections involving a spoiler effect: the 2000 Presidential Election in Florida.  Though the election did not collect ranked ballots from voters, polling suggests \cite{Magee2003,Herron2007} that Gore would have beaten Bush in the two-candidate election $\mathbf{P}|_{ \{\mbox{\footnotesize{Bush, Gore}}\}}$, and Gore would have beaten Nader in the two-candidate election $\mathbf{P}|_{ \{\mbox{\footnotesize{Gore, Nader}}\}}$, yet Gore was not a winner in the three-candidate election $\mathbf{P}|_{ \{\mbox{\footnotesize{Bush, Gore, Nader}}\}}$ according to Plurality voting.\footnote{Of course there were more than three candidates in this election, but this simplified example contains the essential points.} Thus, Plurality voting violates binary $\gamma$. 

In the example above, Gore is the \textit{Condorcet winner}, a candidate who is majority preferred to every other candidate. Thus, the example also shows that Plurality voting violates the property of \textit{Condorcet consistency}, which requires that a method selects the Condorcet winner whenever one exists. However, binary $\gamma$ and Condorcet consistency are independent properties. For example, the Condorcet consistent methods Beat Path \cite{Schulze2011,Schulze2018}, Minimax \cite{Simpson1969,Kramer1977}, and Ranked Pairs \cite{Tideman1987,ZavistTideman1989} all violate binary $\gamma$. Conversely, the trivial voting method that always declares all candidates tied for the win satisfies binary $\gamma$ but not Condorcet consistency. However, there is a connection between binary $\gamma$ and the concept of the Condorcet winner: if a voting method satisfies binary $\gamma$ and reduces to majority voting in two-candidate elections, then the Condorcet winner, if one exists, is always \textit{among} the winners of the election. In fact, the examples of voting methods in the literature satisfying binary $\gamma$, including Banks \cite{Banks1985},  Split Cycle \cite{HP2020b,HP2021}, Top Cycle \cite{Smith1973,Schwartz1986}, Uncovered Set \cite{Miller1980,Duggan2013}, and Weighted Covering \cite{Dutta1999,Fernandez2018}, are all Condorcet consistent. Thus, Condorcet consistency and binary $\gamma$ are cross-cutting but related properties, as shown by the list of voting methods in Table~\ref{binaryGammaFig}. 

\begin{table}[h]
\begin{center}
\begin{tabular}{c|c}
methods violating binary $\gamma$ & methods satisfying binary $\gamma$ \\
\hline
Beat Path$^\star$ & Banks$^\star$ \\
Borda &  Split Cycle$^\star$ \\
Instant Runoff & Top Cycle$^\star$  \\
Minimax$^\star$ & Uncovered Set$^\star$\\
Plurality & Weighted Covering$^\star$ \\
Ranked Pairs$^\star$ & 
\end{tabular}
\end{center}
\caption{examples of methods violating or satisfying binary $\gamma$. Condorcet consistent voting methods are marked with a $\star$.}\label{binaryGammaFig}
\end{table}

Though we can consistently add binary $\gamma$ to our initial package of axioms, adding binary $\gamma$ together with even the weakest consequence of $\alpha$ that we have discussed---$\alpha$-resoluteness---leads to an impossibility theorem. Because axioms like binary $\gamma$ and $\alpha$-resoluteness are variable-candidate axioms, involving choosing winners from elections involving different sets of candidates, we are especially interested in exactly how many candidates are needed to trigger impossibility results.\footnote{Optimizing these impossibility results for the number of voters is an interesting subject for future research.} If there is no voting method satisfying the relevant axioms on all profiles with up to $n$ candidates, then of course there is no voting method satisfying the relevant axioms on all profiles with up to $m$ candidates for $m\geq n$; thus, we are interested in the least number of candidates for which an impossibility theorem holds. 

\begin{theorem}\label{Impossibility1} $\,$
\begin{enumerate}
\item\label{Impossibility1a} There is no voting method satisfying anonymity, neutrality, binary quasi-resoluteness, binary $\gamma$, and $\alpha$-resoluteness for profiles with up to 6 candidates. 
\item\label{Impossibility1b} There is a voting method satisfying the above axioms for profiles with up to 5 candidates.\footnote{Note this improves on the example of Copeland from Example \ref{CopelandEx}.}
\end{enumerate}
\end{theorem}

\noindent We will prove Theorem \ref{Impossibility1} with the help of a SAT solver, as explained in Section~\ref{ProofStrategy}. First, however, we will give a human-readable proof of a weaker version of Theorem~\ref{Impossibility1}.\ref{Impossibility1a} that assumes  that the voting method satisfies the relevant axioms for profiles having up to 7 candidates. 

\begin{proposition}\label{Impossibility1seven} There is no voting method satisfying anonymity, neutrality, binary quasi-resoluteness, binary $\gamma$, and $\alpha$-resoluteness for profiles with up to 7 candidates. \end{proposition}

\begin{proof} Suppose for contradiction that there is a voting method $F$ satisfying the axioms for profiles with up to 7 candidates. By anonymity, neutrality, and binary quasi-resoluteness, for every possible margin $m>0$ and number $k$ of voters, either (i) for every two-candidate profile with $k$ voters in which one candidate has a margin of $m$ over the other, the unique winner is the majority winner or (ii) for every two-candidate profile with $k$ voters in which one candidate has a margin of $m$ over the other, the unique winner is the minority winner. For $m=6$ and $k=18$, suppose that (i) holds; the proof in case (ii) is analogous. Now consider the profile $\mathbf{P}$ in Figure~\ref{ImpossProf}, whose margin graph is shown in Figure~\ref{ImpossGraph} (all arrows in/out of a rectangle go in/out of every node in the rectangle). By (i), if there is an arrow from $x$ to $y$ in Figure~\ref{ImpossGraph}, then $F(\mathbf{P}|_{ \{x,y\}})=\{x\}$.

\begin{figure}[h]

\setlength{\tabcolsep}{2pt}
\begin{center}
\begin{tabular}{c|c|c|c|c|c|c|c|c|c|c|c|c|c|c|c|c|c}
$1$ & $1$ &$1$ &$1$ &$1$ &$1$ &$1$ &$1$ &$1$ &$1$ &$1$ &$1$ &$1$ &$1$ &$1$ &$1$ &$1$ &$1$ \\
\hline
$e$ & $e$  &  $e$ & $e$ & $e$  &  $e$ & $d$ & $d$  &  $d$ & $d$ & $d$  &  $d$ & \cellcolor{green!25}$c$ & \cellcolor{green!25}$c'$  &  \cellcolor{green!25}$c$ & \cellcolor{green!25}$c'$ & \cellcolor{green!25}$c$  &  \cellcolor{green!25}$c'$  \\
 \cellcolor{blue!25}$a$ &  \cellcolor{blue!25}$a$ &  \cellcolor{blue!25}$a'$ &  \cellcolor{blue!25}$a'$ &  \cellcolor{blue!25}$a''$ &  \cellcolor{blue!25}$a''$ & $b$ & $b$& $b$& $b$& $b$& $b$ & \cellcolor{green!25}$c'$ & \cellcolor{green!25}$c$  &  \cellcolor{green!25}$c'$ & \cellcolor{green!25}$c$ & \cellcolor{green!25}$c'$  &  \cellcolor{green!25}$c$ \\
 \cellcolor{blue!25}$a'$ &  \cellcolor{blue!25}$a'$ &  \cellcolor{blue!25}$a'' $ &  \cellcolor{blue!25}$a''$&  \cellcolor{blue!25}$a$ &  \cellcolor{blue!25}$a$ & $e$ & $e$& $e$& $e$& $e$& $e$& $d$& $d$& $d$& $d$& $d$& $d$\\
 \cellcolor{blue!25}$a''$ &  \cellcolor{blue!25}$a''$ &  \cellcolor{blue!25}$a$ &  \cellcolor{blue!25}$a$ &  \cellcolor{blue!25}$a'$ &  \cellcolor{blue!25}$a'$ & \cellcolor{green!25}$c$ & \cellcolor{green!25}$c'$  &  \cellcolor{green!25}$c$ & \cellcolor{green!25}$c'$ & \cellcolor{green!25}$c$  &  \cellcolor{green!25}$c'$ & \cellcolor{blue!25}$a$ & \cellcolor{blue!25}$a$ & \cellcolor{blue!25}$a'$ & \cellcolor{blue!25}$a'$ & \cellcolor{blue!25}$a''$ & \cellcolor{blue!25}$a''$\\
$b$ & $b$& $b$& $b$& $b$& $b$ & \cellcolor{green!25}$c'$ & \cellcolor{green!25}$c$  &  \cellcolor{green!25}$c'$ & \cellcolor{green!25}$c$ & \cellcolor{green!25}$c'$  &  \cellcolor{green!25}$c$ & \cellcolor{blue!25}$a'$ & \cellcolor{blue!25}$a'$ & \cellcolor{blue!25}$a'' $ & \cellcolor{blue!25}$a''$& \cellcolor{blue!25}$a$ & \cellcolor{blue!25}$a$\\
 \cellcolor{green!25}$c$ &  \cellcolor{green!25}$c'$  &   \cellcolor{green!25}$c$ &  \cellcolor{green!25}$c'$ &  \cellcolor{green!25}$c$  &   \cellcolor{green!25}$c'$ &  \cellcolor{blue!25}$a$ & \cellcolor{blue!25}$a$ & \cellcolor{blue!25}$a'$ & \cellcolor{blue!25}$a'$ & \cellcolor{blue!25}$a''$ & \cellcolor{blue!25}$a''$ & \cellcolor{blue!25}$a''$ & \cellcolor{blue!25}$a''$ & \cellcolor{blue!25}$a$ & \cellcolor{blue!25}$a$ & \cellcolor{blue!25}$a'$ & \cellcolor{blue!25}$a'$\\
  \cellcolor{green!25}$c'$ &  \cellcolor{green!25}$c$  &   \cellcolor{green!25}$c'$ &  \cellcolor{green!25}$c$ &  \cellcolor{green!25}$c'$  &   \cellcolor{green!25}$c$ & \cellcolor{blue!25}$a'$ & \cellcolor{blue!25}$a'$ & \cellcolor{blue!25}$a'' $ & \cellcolor{blue!25}$a''$& \cellcolor{blue!25}$a$ & \cellcolor{blue!25}$a$ & $b$ & $b$& $b$& $b$& $b$& $b$\\
 $d$& $d$& $d$& $d$& $d$& $d$ & \cellcolor{blue!25}$a''$ & \cellcolor{blue!25}$a''$ & \cellcolor{blue!25}$a$ & \cellcolor{blue!25}$a$ & \cellcolor{blue!25}$a'$ & \cellcolor{blue!25}$a'$ & $e$ & $e$& $e$& $e$& $e$& $e$
\end{tabular}
\end{center}
\caption{A profile whose subprofiles (with 7 or fewer candidates) are used in the proof of Proposition \ref{Impossibility1seven}.}\label{ImpossProf}
\end{figure}

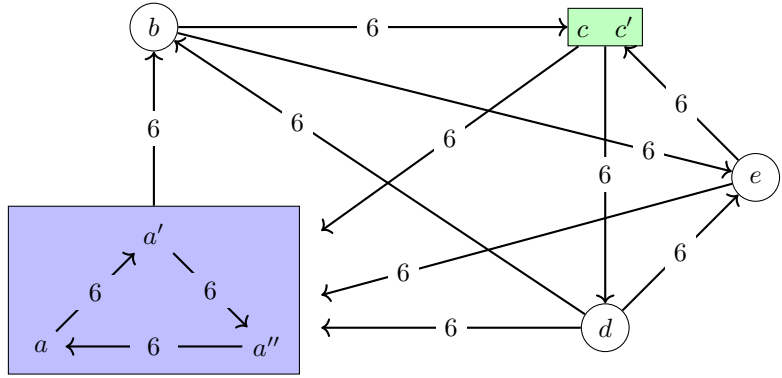
\begin{figure}[h]
\begin{center}
\begin{tikzpicture}

\node[draw, minimum width=0.25in,fill=blue!25] at (0,.5) (a) {\begin{tikzpicture}[scale=0.5]

\node  at (0,0) (a) {$a$}; 
\node at (3,3) (a') {$a'$};  
\node at (6,0) (a'') {$a''$}; 

\path[->,draw,thick] (a) to node[fill=blue!25] {$6$}  (a');
\path[->,draw,thick] (a') to node[fill=blue!25] {$6$} (a'');
\path[->,draw,thick] (a'') to node[fill=blue!25] {$6$} (a);
  \end{tikzpicture}}; 
  
\node[circle,draw,minimum width=0.25in] at (0,4) (b) {$b$}; 
\node[draw,minimum width=0.25in,fill=green!25] at (6,4) (c) {$c$\quad $c'$}; 
\node[circle,draw,minimum width=0.25in] at (6,0) (d) {$d$}; 
\node[circle,draw,minimum width=0.25in] at (8,2) (e) {$e$}; 

\node at (2.1,1.2) (a*) {};
\node at (2.1,.4) (a**) {};
\node at (2.1,0) (a***) {};

\path[->,draw,thick] (a) to node[fill=white] {$6$} (b);
\path[->,draw,thick] (b) to node[fill=white] {$6$}  (c);
\path[->,draw,thick] (c) to node[fill=white] {$6$} (a*);
\path[->,draw,thick] (c) to node[fill=white] {$6$}  (d);
\path[->,draw,thick] (d) to node[fill=white,pos=.7] {$6$}  (b);
\path[->,draw,thick] (d) to node[fill=white] {$6$}  (a***);
\path[->,draw,thick] (e) to node[fill=white,pos=.8] {$6$}  (a**);
\path[->,draw,thick] (b) to node[fill=white,pos=.85] {$6$} (e);
\path[->,draw,thick] (e) to node[fill=white] {$6$} (c);
\path[->,draw,thick] (d) to node[fill=white] {$6$} (e);

  \end{tikzpicture}
  \end{center}
  \caption{Margin graph of the profile in Figure \ref{ImpossProf}.} \label{ImpossGraph}
  \end{figure}

Although the profile in Figure \ref{ImpossProf} contains 8 candidates, in what follows we will only use the assumption that $F$ satisfies the relevant axioms on subprofiles of $\mathbf{P}$ with up to 7 candidates. We will repeatedly use the fact that in a perfectly symmetrical three-candidate profile realizing a majority cycle, such as $\mathbf{P}|_{ \{a,b,c\}}$,  $\mathbf{P}|_{ \{b,c,d\}}$, and $\mathbf{P}|_{ \{a,b,e\}}$, anonymity and neutrality imply that every candidate is chosen. Thus, ${c\in F(\mathbf{P}|_{ \{a,b,c\}})}$, so $c\in F(\mathbf{P}|_{ \{a,a',a'',b,c\}})$ by two applications of binary $\gamma$. By anonymity and neutrality we also have that \[\mbox{${\{a,a',a''\}\cap F(\mathbf{P}|_{ \{a,a',a'',b,c\}})\neq\varnothing}$ if and only if ${\{a,a',a''\}\subseteq F(\mathbf{P}|_{ \{a,a',a'',b,c\}})}$.}\] Thus, one of the following three cases holds.

Case 1: $F(\mathbf{P}|_{ \{a,a',a'',b,c\}})=\{c\}$. Hence $c\in F(\mathbf{P}|_{ \{a,a',a'',b,c,d\}})$ by binary $\gamma$. Moreover,  by anonymity and neutrality, ${d\in F(\mathbf{P}|_{ \{b,c,d\}})}$, so ${d\in F(\mathbf{P}|_{ \{a,a',a'',b,c,d\}})}$ by three applications of binary $\gamma$. But together $F(\mathbf{P}|_{ \{a,a',a'',b,c\}})=\{c\}$, $F(\mathbf{P}|_{ \{c,d\}})=\{c\}$, and ${c,d\in F(\mathbf{P}|_{ \{a,a',a'',b,c,d\}})}$ contradict $\alpha$-resoluteness.

Case 2: $F(\mathbf{P}|_{ \{a,a',a'',b,c\}})=\{b,c\}$. Hence $b\in F(\mathbf{P}|_{ \{a,a',a'',b,c,c'\}})$ by binary $\gamma$.  By anonymity and neutrality, \[\mbox{$\{c,c'\}\cap F(\mathbf{P}|_{ \{a,a',a'',b,c,c'\}})\neq\varnothing$ if and only if $\{c,c'\}\subseteq F(\mathbf{P}|_{ \{a,a',a'',b,c,c'\}})$,}\] and \[\mbox{$\{a,a',a''\}\cap F(\mathbf{P}|_{ \{a,a',a'',b,c,c'\}})\neq\varnothing$ if and only if $\{a,a',a''\}\subseteq F(\mathbf{P}|_{ \{a,a',a'',b,c,c'\}})$.}\] Then given $F(\mathbf{P}|_{ \{a,a',a'',b,c\}})=\{b,c\}$ and $F(\mathbf{P}|_{ \{b,c'\}})=\{b\}$, $\alpha$-resoluteness implies  $F(\mathbf{P}|_{ \{a,a',a'',b,c,c'\}})=\{b\}$. Next, by anonymity and neutrality, ${e\in F(\mathbf{P}|_{ \{a,b,e\}})}$ and hence $e\in F(\mathbf{P}|_{ \{a,a',a'',b,c,c',e\}})$ by four applications of binary $\gamma$. Moreover, given $F(\mathbf{P}|_{ \{a,a',a'',b,c,c'\}})=\{b\}$, we have $b\in  F(\mathbf{P}|_{ \{a,a',a'',b,c,c',e\}})$ by binary $\gamma$. But together $F(\mathbf{P}|_{  \{a,a',a'',b,c,c'\}})=\{b\}$, $F(\mathbf{P}|_{ \{b,e\}})=\{b\}$, and ${b,e\in F(\mathbf{P}|_{ \{a,a',a'',b,c,c',e\}})}$ contradict $\alpha$-resoluteness.

Case 3: $\{c,a,a',a''\}\subseteq F(\mathbf{P}|_{ \{a,a',a'',b,c\}})$. By anonymity and neutrality, we have $F(\mathbf{P}|_{ \{a,b,c\}})=\{a,b,c\}$. Then since $F(\mathbf{P}|_{ \{a',c\}}) = \{c\}$, it follows by binary~$\gamma$ and $\alpha$-resoluteness that $c\in F(\mathbf{P}|_{ \{a,a',b,c\}})$ and $|F(\mathbf{P}|_{ \{a,a',b,c\}})|\leq 3$, which with $F(\mathbf{P}|_{ \{a'',c\}}) = \{c\}$ and $\alpha$-resoluteness implies $|F(\mathbf{P}|_{ \{a,a',a'',b,c\}})|\leq 3$, which contradicts the assumption of the case.\end{proof}

\begin{remark} $\,$
\begin{enumerate}
\item If we drop either anonymity or neutrality, then there are voting methods satisfying the remaining axioms from Theorem \ref{Impossibility1}. For an anonymous (resp.~neutral) voting method $F$ satisfying the other axioms, fix a linear order $L$ of the set of all possible candidates (resp.~for each possible set $V$ of voters, fix a voter $i_V\in V$), and for any profile $\mathbf{P}$, let $F(\mathbf{P})=\{x\}$ where $x$ is the maximal element in the linear order $L$ restricted to $X(\mathbf{P})$ (resp.~in $i_{V(\mathbf{P})}$'s linear order of $X(\mathbf{P})$) such that no candidate ranked below $x$ in this linear order is majority preferred to $x$. Observe that $F$ satisfies binary $\gamma$. Then since $F$ always selects a unique winner, it also satisfies $\alpha$-resoluteness.
 
\item If we drop binary quasi-resoluteness, then there are voting methods satisfying the remaining axioms from Theorem \ref{Impossibility1}, e.g., the trivial voting method that always selects all candidates as winners (note that this method trivially satisfies binary $\gamma$ and $\alpha$-resoluteness, because the condition $F(\mathbf{P}|_{ \{x,y\}})=\{x\}$ is always false).
\item It is open whether if we drop binary $\gamma$, there are voting methods satisfying the remaining axioms from Theorem \ref{Impossibility1}. However, we do know from SAT solving (see Appendix \ref{SATmajoritarian}) that there are voting methods satisfying the remaining axioms at least for profiles having up to 7 candidates.
\item If we drop $\alpha$-resoluteness, then all the voting methods in the right column of Table \ref{binaryGammaFig} satisfy the remaining axioms from Theorem \ref{Impossibility1}.
\end{enumerate}
\end{remark}

While binary $\gamma$ and $\alpha$-resoluteness for \textit{individual choice functions} are easily satisfied---in particular by any choice function representable by a binary relation---Theorem \ref{Impossibility1}.\ref{Impossibility1a} shows that binary $\gamma$ and $\alpha$-resoluteness for \textit{voting methods} are together inconsistent with some of the most uncontroversial axioms for voting, namely anonymity, neutrality, and binary quasi-resoluteness.

A natural response to Proposition \ref{NoResolute} and Theorem \ref{Impossibility1}.\ref{Impossibility1a} is that asking for any kind of resoluteness property with respect to \textit{all} profiles, including highly symmetric profiles, is asking for too much. Thus, $\alpha$-resoluteness is asking for too much. This motivates a move to more qualified resoluteness properties that only apply to profiles with sufficiently broken symmetries. For example, call a profile $\mathbf{P}$ \textit{uniquely weighted} if for any two distinct pairs $a,b$ and $x,y$ of distinct candidates, the margin of $a$ over $b$ in $\mathbf{P}$ is not equal to the margin of $x$ over $y$ in $\mathbf{P}$. Thus, the profile in Figure~\ref{ImpossGraph} is \textit{not} uniquely weighted. By contrast, in elections with many voters, one would expect with high probability to obtain a uniquely weighted profile. Now consider restricting resoluteness only to uniquely weighted profiles:\\

\noindent\textbf{quasi-resoluteness}: $|F(\mathbf{P})| = 1$ for all \textit{uniquely weighted} profiles $\mathbf{P}$.\\

\noindent There are many voting methods satisfying the trio of anonymity, neutrality, and quasi-resoluteness, including Beat Path, Minimax,  and Ranked Pairs, as well as any voting method that uses a quasi-resolute voting method as a tiebreaker, e.g., Plurality with Minimax tiebreaking (in the event of a tie for the most first place votes). However, we will show that even this weakening of resoluteness is incompatible with basic principles of individual rational choice applied to voting---in particular, with the binary $\gamma$ axiom alone---assuming anonymity and neutrality.

\begin{theorem}\label{Impossibility4} There is no voting method satisfying anonymity, neutrality, binary $\gamma$, and quasi-resoluteness for profiles with up to 4 candidates.
\end{theorem}
\noindent Theorem \ref{Impossibility4} explains the observation from the literature that voting methods satisfying binary $\gamma$, listed in Table \ref{binaryGammaFig}, all violate quasi-resoluteness, while the quasi-resolute methods mentioned above all violate binary $\gamma$.  Note that 4 candidates is optimal for the impossibility, since Beat Path, Minimax, Ranked Pairs, and Split Cycle satisfy the axioms in Theorem \ref{Impossibility4} for profiles up to 3 candidates.

We will prove Theorem \ref{Impossibility4} with the help of a SAT solver, as explained in Section \ref{ProofStrategy}. First, however, we will give a human-readable proof of a variation on Theorem  \ref{Impossibility4}. For the variation, we add the following assumptions. Given a profile $\mathbf{P}$, the profile $n\mathbf{P}$ replaces each voter from $\mathbf{P}$ with $n$ copies of that voter with the same ballot.\footnote{I.e., $V(n\mathbf{P})= \{(i,m)\mid i\in V(\mathbf{P}),m\in \mathbb{Z}^+,m\leq n\}$ and $(a,b)\in n\mathbf{P}((i,m))$ if $(a,b)\in \mathbf{P}(i)$.}

\begin{definition}\label{Invariance} Let $F$ be a voting method.
\begin{enumerate}
\item $F$ is  \textit{homogeneous} if  for all $\mathbf{P}\in\mathrm{dom}(F)$ and $n\in\mathbb{Z}^+$, if $n\mathbf{P}\in \mathrm{dom}(F)$, then $F(\mathbf{P})=F(n\mathbf{P})$ \cite{Smith1973}.
\item $F$ is \textit{cancellative} if for all $\mathbf{P},\mathbf{P}'\in\mathrm{dom}(F)$, if $\mathbf{P}'$ is obtained from $\mathbf{P}$ by deleting  one copy of each linear order of $X(\mathbf{P})$, then $F(\mathbf{P})=F(\mathbf{P'})$.
\end{enumerate}
\end{definition}

Adding these assumptions alongside those of Theorem \ref{Impossibility4} would make for a weaker impossibility theorem. However, because we will be able to construct profiles by hand in our proof, we can weaken quasi-resoluteness to the following property.

\begin{definition} A profile $\mathbf{P}$ is \textit{uniquely ranked} if for every $Y\subseteq X(\mathbf{P})$ with $|Y|\geq 2$:
\begin{enumerate}
\item for each $n \in \{1,\dots, |Y|\}$, no two candidates have the same number of $n$-th place votes in $\mathbf{P}|_{ Y}$; 
\item no two candidates have the same \textit{Borda score} in $\mathbf{P}\mathord{\mid}_{Y}$, where the Borda score of $x$ in a profile $\mathbf{Q}$ is $\underset{i\in V(\mathbf{Q})}{\sum} |\{y\in X(\mathbf{Q})\mid (x,y)\in \mathbf{Q}(i) \}|$.
\end{enumerate}
A voting method $F$ satisfies \textit{weak quasi-resoluteness} if $|F(\mathbf{P})|=1$ for every profile $\mathbf{P}$ that is both uniquely weighted and uniquely ranked.
\end{definition}

Examples of voting methods that are weakly quasi-resolute but not quasi-resolute include Plurality, Borda, Instant Runoff, Coombs \cite{Coombs1964,Grofman2004}, Baldwin \cite{Baldwin1926}, Nanson \cite{Nanson1882,Niou1987}, and Bucklin \cite{Hoag1926}. In the following proof, it is tedious to check by hand that all the profiles are uniquely ranked (for code to check this, see \href{https://github.com/szahedian/binary-gamma}{https://github.com/szahedian/binary-gamma}). However, it is easy to check that the profiles are uniquely weighted, which suffices for the version of Proposition \ref{Impossibility4b} with quasi-resoluteness in place of weak quasi-resoluteness.

\begin{proposition}\label{Impossibility4b} There is no anonymous, neutral, homogeneous, and cancellative voting method on the domain of profiles with up to 4 candidates  satisfying binary $\gamma$ and weak quasi-resoluteness.
\end{proposition}

\begin{proof} Suppose for contradiction there is such a voting method $F$.

\textbf{Case A}: in a two-candidate profile $\mathbf{P}_{a,b}$ with 4 voters all of whom rank $a$ above $b$, $a$ is a winner. Then by anonymity, neutrality, homogeneity, and cancellativity, we have the following  fact $(\star)$: for any even $m\geq 4$ and $n\in\mathbb{Z}^+$, in any two-candidate $\mathbf{P}'$ profile with $m$ voters and a margin of $4n$ for some candidate $x$ over another, $y$, $x$ is a winner. For by neutrality, $x$ wins in the profile $\mathbf{P}_{x,y}$ obtained from $\mathbf{P}_{a,b}$ by swapping $x$ for $a$ and $y$ for $b$. Then by homogeneity, $x$ wins in $n\mathbf{P}_{x,y}$, wherein the margin of $x$ over $y$ is $4n$. Let $\mathbf{P}'_0$ be obtained from $n\mathbf{P}_{x,y}$ by adding $(m-4n)/2$ voters with the linear order $xy$ and  $(m-4n)/2$ voters with the linear order $yx$, so the number of voters is $m$ and the margin of $x$ over $y$ is still $4n$. It follows by cancellativity that $x$ wins in $\mathbf{P}'_0$. Then by anonymity, $x$ wins in $\mathbf{P}'$.

Now consider the following $96$-voter profile $\mathbf{P}_0$:
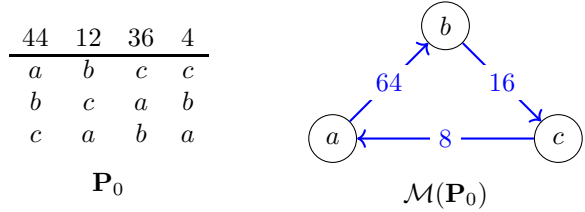
\begin{center}
\begin{minipage}{2in}\begin{center}
\begin{tabular}{cccc}
$44$ & $12$    &  $36$ & $4$  \\\hline
$a$ & $b$  &  $c$ & $c$  \\
$b$ & $c$  &  $a$ & $b$ \\
$c$ & $a$  &  $b$ &  $a$
\end{tabular}\vspace{.1in}

$\mathbf{P}_0$\end{center}\end{minipage}
\begin{minipage}{2in} \begin{tikzpicture}

\node[circle,draw,minimum width=0.25in] at (0,0) (a) {$a$}; 
\node[circle,draw,minimum width=0.25in] at (1.5,1.5) (b) {$b$}; 
\node[circle,draw,minimum width=0.25in] at (3,0) (c) {$c$}; 

\path[->,draw,thick, blue] (a) to node[fill=white] {$64$} (b);
\path[->,draw,thick, blue] (b) to node[fill=white] {$16$} (c);
\path[->,draw,thick, blue] (c) to node[fill=white] {$8$} (a);

\node  at (1.5,-.75) () {$\mathcal{M}(\mathbf{P}_0)$}; 

\end{tikzpicture}
\end{minipage}
\end{center}

\noindent \textbf{Case A.1}: $a\in F(\mathbf{P}_0)$. Then by neutrality and homogeneity, we have $d\in F(\mathbf{Q}_1)$ for the following $48$-voter profile $\mathbf{Q}_1$:

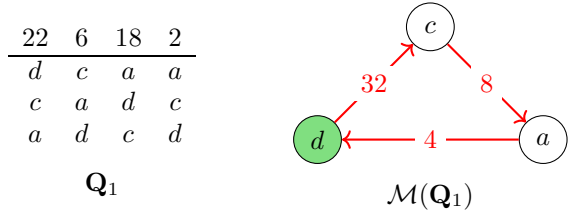
\begin{center}\begin{minipage}{2in}\begin{center}
\begin{tabular}{cccc}
$22$ & $6$    &  $18$ & $2$  \\\hline
$d$ & $c$  &  $a$ & $a$  \\
$c$ & $a$  &  $d$ & $c$ \\
$a$ & $d$  &  $c$ &  $d$
\end{tabular}\vspace{.1in}

$\mathbf{Q}_1$\end{center}\end{minipage}\begin{minipage}{2in}\begin{tikzpicture}

\node[circle,fill = medgreen!50, draw,minimum width=0.25in] at (0,0) (a) {$d$}; 
\node[circle,draw,minimum width=0.25in] at (1.5,1.5) (b) {$c$}; 
\node[circle,draw,minimum width=0.25in] at (3,0) (c) {$a$}; 

\path[->,draw,thick, red] (a) to node[fill=white] {$32$} (b);
\path[->,draw,thick, red] (b) to node[fill=white] {$8$} (c);
\path[->,draw,thick, red] (c) to node[fill=white] {$4$} (a);

\node  at (1.5,-.75) () {$\mathcal{M}(\mathbf{Q}_1)$}; 

\end{tikzpicture}\end{minipage}
\end{center}
Then by cancellativity, $d\in F(\mathbf{Q}_1^+)$ where $\mathbf{Q}_1^+$ is the 96-voter profile obtained from $\mathbf{Q}_1$ by adding 8 copies of each linear order of $\{a,c,d\}$:

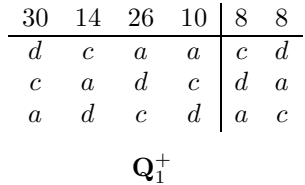
\begin{center}\begin{minipage}{1.5in}\begin{center}
\begin{tabular}{cccc | cc}
$30$ & $14$    &  $26$ & $10$ & $8$ & $8$  \\\hline
$d$ & $c$  &  $a$ & $a$ & $c$  & $d$  \\
$c$ & $a$  &  $d$ & $c$ & $d$ & $a$ \\
$a$ & $d$  &  $c$ &  $d$& $a$ & $c$ 
\end{tabular}\vspace{.1in}

$\mathbf{Q}_1^+$\end{center}\end{minipage}\end{center}
Then where $\mathbf{P}_1$ is the following uniquely weighted and uniquely ranked 96-voter profile such that $\mathbf{P}_1\mathord{\mid_{\{a,b,c\}}}=\mathbf{P}_0$ and $\mathbf{P}_1\mathord{\mid_{\{a,c,d\}}} =\mathbf{Q}_1^+$, $(\star)$ and binary $\gamma$ imply $\{a,d\}\subseteq F(\mathbf{P}_1)$, violating weak quasi-resoluteness:

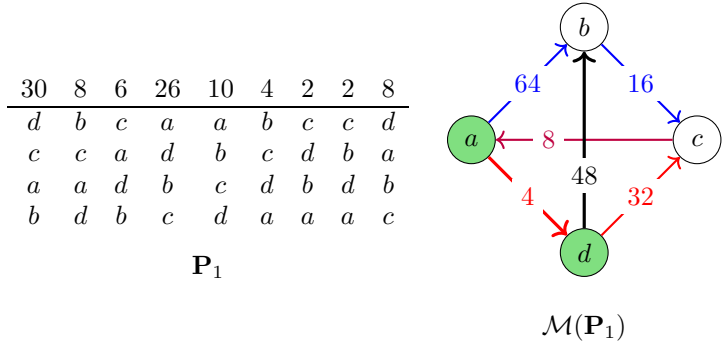
\begin{center}

\begin{minipage}{2.5in}\begin{center}\begin{tabular}{ccccccccc}
$30$ & $8$ & $6$    &  $26$ & $10$ & $4$ & $2$ & $2$ &  $8$  \\\hline
$d$ & $b$ & $c$ &  $a$ & $a$ & $b$  & $c$ & $c$ & $d$  \\
$c$ & $c$ & $a$ &  $d$ & $b$ & $c$ & $d$ & $b$ & $a$ \\
$a$ & $a$ & $d$ &  $b$ &  $c$& $d$ & $b$ & $d$ &$b$ \\
$b$ & $d$ & $b$ & $c$ & $d$ & $a$ & $a$ &$a$ & $c$
\end{tabular}

\vspace{.1in}

$\mathbf{P}_1$
\end{center}
\end{minipage}\begin{minipage}{2.25in}
\begin{tikzpicture}

\node[circle,fill = medgreen!50,draw, minimum width=0.25in] at (0,0) (a) {$a$}; 
\node[circle,draw,minimum width=0.25in] at (1.5,1.5) (b) {$b$}; 
\node[circle,draw,minimum width=0.25in] at (3,0) (c) {$c$}; 
\node[circle,fill = medgreen!50,draw,minimum width=0.25in] at (1.5,-1.5) (d) {$d$}; 

\path[->,draw,thick,blue] (b) to node[fill=white] {$16$} (c);
\path[->,draw,thick,purple] (c) to node[fill=white,pos=.7] {$8$} (a);
\path[->,draw,thick,blue] (a) to node[fill=white] {$64$} (b);
\path[<-,draw,very thick,black] (b) to node[fill=white,pos=.7] {$48$} (d);
\path[<-,draw,very thick,red] (d) to node[fill=white] {$4$} (a);
\path[->,draw,thick,red] (d) to node[fill=white] {$32$} (c);

\node at (1.5,-2.5)  {$\mathcal{M}(\mathbf{P}_1)$}; 

\end{tikzpicture}\end{minipage}
\end{center}

\noindent \textbf{Case A.2}: $b\in F(\mathbf{P}_0)$. Then by neutrality and homogeneity, we have $d\in F(\mathbf{Q}_2)$ for the following 768-voter profile $\mathbf{Q}_2$:

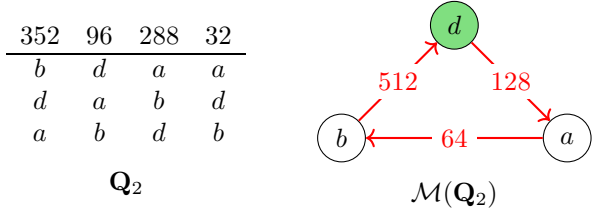
\begin{center}
\begin{minipage}{2in}\begin{center}
\begin{tabular}{cccc}
$352$ & $96$    &  $288$ & $32$  \\\hline
$b$ & $d$  &  $a$ & $a$  \\
$d$ & $a$  &  $b$ & $d$ \\
$a$ & $b$  &  $d$ &  $b$
\end{tabular}\vspace{.1in}

$\mathbf{Q}_2$\end{center}
\end{minipage}\begin{minipage}{2in}\begin{tikzpicture}

\node[circle, draw,minimum width=0.25in] at (0,0) (a) {$b$}; 
\node[circle,fill = medgreen!50,draw,minimum width=0.25in] at (1.5,1.5) (b) {$d$}; 
\node[circle,draw,minimum width=0.25in] at (3,0) (c) {$a$}; 

\path[->,draw,thick, red] (a) to node[fill=white] {$512$} (b);
\path[->,draw,thick, red] (b) to node[fill=white] {$128$} (c);
\path[->,draw,thick, red] (c) to node[fill=white] {$64$} (a);

\node  at (1.5,-.75) () {$\mathcal{M}(\mathbf{Q}_2)$}; 

\end{tikzpicture}\end{minipage}
\end{center}
Moreover, by cancellativity, $b\in F(\mathbf{P}_0)$ implies $b\in F(\mathbf{P}_0^+)$ for the 768-voter profile $\mathbf{P}_0^+$ obtained from $\mathbf{P}_0$ by adding 112 copies of each linear order over $\{a,b,c\}$:

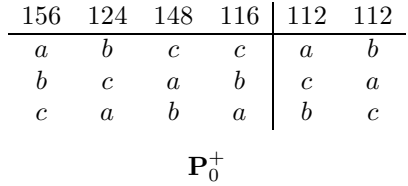
\begin{center}
\begin{tabular}{cccc | cc}
$156$ & $124$    &  $148$ & $116$ & $112$ & $112$  \\\hline
$a$ & $b$  &  $c$ & $c$ & $a$ & $b$  \\
$b$ & $c$  &  $a$ & $b$ & $c$ & $a$ \\
$c$ & $a$  &  $b$ &  $a$ & $b$ & $c$
\end{tabular}\vspace{.1in}

$\mathbf{P}_0^+$
\end{center}
Then where $\mathbf{P}_2$ is the following uniquely weighted and uniquely ranked 768-voter profile such that $\mathbf{P}_2\mathord{\mid_{\{a,b,c\}}}=\mathbf{P}_0^+$ and $\mathbf{P}_2\mathord{\mid_{\{a,b,d\}}} =\mathbf{Q}_2$, $(\star)$ and binary $\gamma$ imply $\{b,d\}\subseteq F(\mathbf{P}_2)$, violating weak quasi-resoluteness:

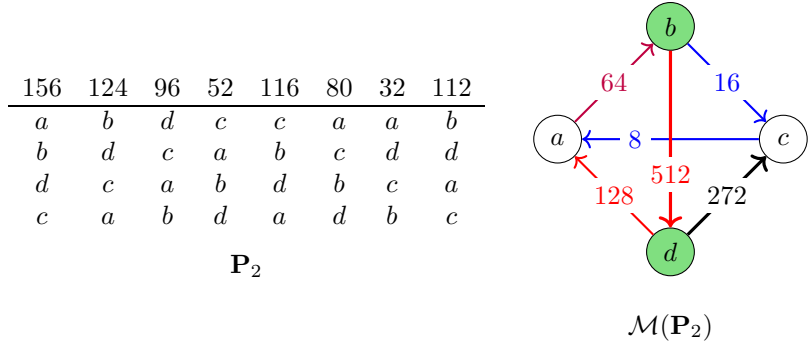
\begin{center}
\begin{minipage}{3in}
\begin{center}\begin{tabular}{cccccccc}
$156$ & $124$    &  $96$ & $52$ & $116$ & $80$ & $32$ & $112$  \\\hline
$a$ & $b$  &  $d$ & $c$ & $c$ & $a$ & $a$ & $b$  \\
$b$ & $d$  &  $c$ & $a$ & $b$ & $c$ & $d$ & $d$ \\
$d$ & $c$  &  $a$ & $b$&  $d$ & $b$ & $c$& $a$ \\
$c$ & $a$ &  $b$ & $d$ & $a$ & $d$ &$b$  & $c$
\end{tabular}

\vspace{.1in}

$\mathbf{P}_2$
\end{center}
\end{minipage}\begin{minipage}{3in}
\begin{tikzpicture}

\node[circle,draw, minimum width=0.25in] at (0,0) (a) {$a$}; 
\node[circle,fill = medgreen!50,draw,minimum width=0.25in] at (1.5,1.5) (b) {$b$}; 
\node[circle,draw,minimum width=0.25in] at (3,0) (c) {$c$}; 
\node[circle,fill = medgreen!50,draw,minimum width=0.25in] at (1.5,-1.5) (d) {$d$}; 
\node at (1.5,-2.5)  {$\mathcal{M}(\mathbf{P}_2)$}; 

\path[->,draw,thick,blue] (b) to node[fill=white] {$16$} (c);
\path[->,draw,thick,blue] (c) to node[fill=white,pos=.7] {$8$} (a);
\path[->,draw,thick,purple] (a) to node[fill=white] {$64$} (b);
\path[->,draw,very thick,red] (b) to node[fill=white,pos=.7] {$512$} (d);
\path[->,draw,thick,red] (d) to node[fill=white] {$128$} (a);
\path[->,draw,very thick,black] (d) to node[fill=white] {$272$} (c);

\end{tikzpicture}\end{minipage}
\end{center}

\noindent \textbf{Case A.3}: $c\in F(\mathbf{P}_0)$.  Then by neutrality and homogeneity, we have $d\in F(\mathbf{Q}_3)$ for the following 24-voter profile $\mathbf{Q}_3$:

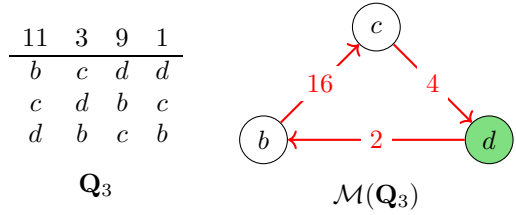
\begin{center}
\begin{minipage}{1.5in}\begin{center}
\begin{tabular}{cccc}
$11$ & $3$    &  $9$ & $1$  \\\hline
$b$ & $c$  &  $d$ & $d$  \\
$c$ & $d$  &  $b$ & $c$ \\
$d$ & $b$  &  $c$ &  $b$
\end{tabular}\vspace{.1in}

$\mathbf{Q}_3$\end{center}\end{minipage}\begin{minipage}{1.5in} \begin{tikzpicture}

\node[circle, draw,minimum width=0.25in] at (0,0) (a) {$b$}; 
\node[circle,draw,minimum width=0.25in] at (1.5,1.5) (b) {$c$}; 
\node[circle,fill = medgreen!50,draw,minimum width=0.25in] at (3,0) (c) {$d$}; 

\path[->,draw,thick, red] (a) to node[fill=white] {$16$} (b);
\path[->,draw,thick, red] (b) to node[fill=white] {$4$} (c);
\path[->,draw,thick, red] (c) to node[fill=white] {$2$} (a);

\node  at (1.5,-.75) () {$\mathcal{M}(\mathbf{Q}_3)$}; 

\end{tikzpicture}\end{minipage}
\end{center}
Then by cancellativity, $d\in F(\mathbf{Q}_3^+)$ where $\mathbf{Q}_3^+$ is the 96-voter profile obtained from $\mathbf{Q}_3$ by adding 12 copies of each linear order of $\{b,c,d\}$:
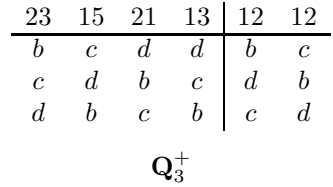
\begin{center}
\begin{minipage}{2.5in}
\begin{center}
\begin{tabular}{cccc | cc}
$23$ & $15$    &  $21$ & $13$ & $12$ & $12$  \\\hline
$b$ & $c$  &  $d$ & $d$ & $b$ & $c$  \\
$c$ & $d$  &  $b$ & $c$ & $d$  & $b$ \\
$d$ & $b$  &  $c$ &  $b$ & $c$ & $d$
\end{tabular}\vspace{.1in}

$\mathbf{Q}_3^+$
\end{center}
\end{minipage}
\end{center}
Then where $\mathbf{P}_3$ is the following uniquely weighted and uniquely ranked 96-voter profile such that $\mathbf{P}_3\mathord{\mid_{\{a,b,c\}}}=\mathbf{P}_0$ and $\mathbf{P}_3\mathord{\mid_{\{b,c,d\}}} =\mathbf{Q}_3^+$, $(\star)$ and binary $\gamma$ imply $\{c,d\}\subseteq F(\mathbf{P}_3)$, violating weak quasi-resoluteness:

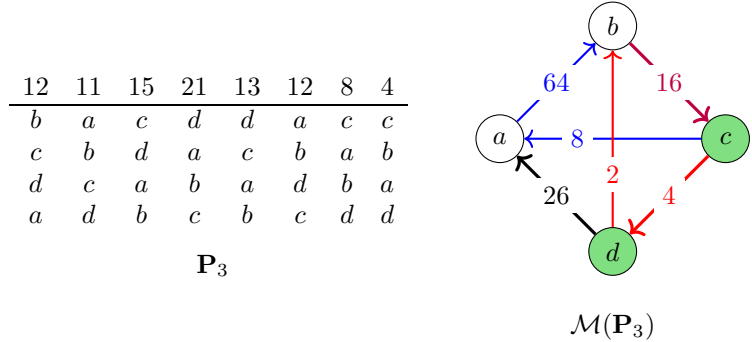
\begin{center}\begin{minipage}{2.75in}
\begin{center}
\begin{tabular}{cccccccc}
$12$ & $11$ & $15$    &  $21$ & $13$ & $12$ & $8$ & $4$  \\\hline
$b$ &$a$ & $c$  &  $d$ & $d$ & $a$ & $c$ & $c$  \\
$c$ &$b$& $d$  &  $a$ & $c$ & $b$  & $a$ & $b$ \\
$d$ &$c$& $a$  &  $b$ &  $a$ & $d$ & $b$ & $a$ \\
$a$ &$d$ & $b$ & $c$ & $b$ & $c$ & $d$ & $d$
\end{tabular}
\vspace{.1in}

$\mathbf{P}_3$
\end{center}
\end{minipage}\begin{minipage}{2.75in}
\begin{tikzpicture}

\node[circle,draw, minimum width=0.25in] at (0,0) (a) {$a$}; 
\node[circle,draw,minimum width=0.25in] at (1.5,1.5) (b) {$b$}; 
\node[circle,fill = medgreen!50,draw,minimum width=0.25in] at (3,0) (c) {$c$}; 
\node[circle,fill = medgreen!50,draw,minimum width=0.25in] at (1.5,-1.5) (d) {$d$}; 
\node at (1.5,-2.5)  {$\mathcal{M}(\mathbf{P}_3)$}; 

\path[->,draw,very thick,purple] (b) to node[fill=white] {$16$} (c);
\path[->,draw,thick,blue] (c) to node[fill=white,pos=.7] {$8$} (a);
\path[->,draw,thick,blue] (a) to node[fill=white] {$64$} (b);
\path[<-,draw,thick,red] (b) to node[fill=white,pos=.7] {$2$} (d);
\path[->,draw,very thick,black] (d) to node[fill=white] {$26$} (a);
\path[<-,draw,very thick,red] (d) to node[fill=white] {$4$} (c);

\end{tikzpicture}\end{minipage}
\end{center}

\textbf{Case B}: in a two-candidate profile with 4 voters all of whom rank $x$ above $y$, $x$ is not a winner, so $y$ is a winner. Then the proof can be obtained from that of Case A by reversing all ballots and all arrows.\end{proof}

Reflection on the proof of Proposition \ref{Impossibility4b} suggests that without homogeneity and cancellativity, there is an explosion in the number of cases one must consider, so we cannot prove Theorem \ref{Impossibility4} in the same manner.

In the rest of the paper, we will build up to proofs of Theorems \ref{Impossibility1}.\ref{Impossibility1a} and Theorem~\ref{Impossibility4}. As we have seen, these impossibility theorems show that a core principle of individual rational choice, namely (binary) $\gamma$, is inconsistent with weakenings of resoluteness in the setting of anonymous and neutral social choice. Though these impossibility results raise important conceptual questions about which axioms ought to be weakened in response, in this paper we focus only on the proofs of the results. We hope that these proofs may also suggest useful strategies for other impossibility results in social choice.

\section{Proof strategy}\label{ProofStrategy}

In this section, we outline our proofs of Theorems \ref{Impossibility1}.\ref{Impossibility1a} and \ref{Impossibility4}. In principle, there is a straightforward proof strategy coming from mathematical logic. Consider a language $L_{prof}$ of propositional logic \cite[\S~1.1]{Enderton2001} that contains for each profile $\mathbf{P}$ and nonempty $Y\subseteq X(\mathbf{P})$ a propositional variable $A_{\mathbf{P}, Y}$, which we interpret to mean that the set of winning candidates for profile $\mathbf{P}$ is $Y$. Then it is easy to show (cf.~the Appendix) that there is a set $\Sigma$ of formulas of $L_{prof}$ that is satisfiable if and only if there exists a voting method satisfying anonymity, neutrality, binary quasi-resoluteness, binary $\gamma$, and $\alpha$-resoluteness (resp.~anonymity, neutrality, binary $\gamma$, and quasi-resoluteness). Now if $\Sigma$ is unsatisfiable, then by the compactness theorem for propositional logic \cite[\S~1.7]{Enderton2001}, there is a finite subset $\Sigma_0$ of $\Sigma$ that is unsatisfiable. Since $\Sigma_0$ is finite, there are some maximum numbers $n$ and $m$ of voters and candidates, respectively, appearing in any profile in the set $\{\mathbf{P}\mid A_{\mathbf{P},Y}\mbox{ occurs in some formula in }\Sigma_0\}$.  Now there is a formula $\varphi_{n,m}$ of $L_{prof}$ that it satisfiable if and only if there is a voting method satisfying the relevant axioms for profiles with up to $n$ voters and $m$ candidates; and the satisfiability of $\varphi_{n,m}$ can be checked using a SAT solver on a computer. Thus, if it is possible to prove a result like Theorem \ref{Impossibility1}.\ref{Impossibility1a} (resp.~\ref{Impossibility4}) at all, then it is in principle possible to do so by using a SAT solver to show the unsatisfiability of $\varphi_{n,m}$ for sufficiently large $n$,~$m$. 

For processing by the SAT solver, we write $\varphi_{n,m}$ in \textit{conjunctive normal form} (CNF), i.e., a conjunction of disjunctions of propositional variables or their negations: $(\ell_{1,1}\vee \dots \vee\ell_{1,k_1})\wedge\dots \wedge(\ell_{j,1}\vee \dots \vee \ell_{j,k_j})$, where each $\ell$ is of the form $A_{\mathbf{P}, Y}$ or $\neg A_{\mathbf{P}, Y}$.  If the SAT solver finds that $\varphi_{n,m}$ is unsatisfiable, then we may apply algorithms that extract a minimal unsatisfiable set (MUS) of conjuncts from  $\varphi_{n,m}$, i.e., an unsatisfiable set of conjuncts such that removing any conjunct results in a satisfiable formula. If the MUS is sufficiently small, we may be able to reconstruct a human-readable proof of its unsatisfiability (as in, e.g., \cite{Brandt2016,Brandl2018}). Even if a human-readable proof is infeasible, we may be able to formally verify the SAT result using an interactive theorem prover (as in, e.g., \cite{Brandl2018,Brandt2018}). SAT solving has already been used to prove a number of impossibility theorems and other results in computational social choice (see, e.g., \cite{Geist2011,Brandt2014,Brandt2016,Brandt2017,Geist2017,Brandl2018,Brandt2018,Brandl2019,Kluiving2020}).

The difficulty with the strategy above for our purposes here is that for more than a few candidates and voters, the number of profiles---and hence propositional variables $A_{\mathbf{P},Y}$---is too large to feasibly carry out the strategy to prove Theorems \ref{Impossibility1}.\ref{Impossibility1a} and \ref{Impossibility4}. Of course, since we are interested in voting methods satisfying anonymity, it would suffice to restrict attention to equivalence classes of profiles related by swapping voters as in the anonymity axiom, or equivalently, with anonymous profiles (see Definition \ref{GenAnonProf}). In fact, since we are interested in voting methods satisfying anonymity and neutrality, it would suffice to work with equivalence classes of profiles related by swapping voters and/or candidates as in the anonymity and neutrality axioms---so-called anonymous and neutral equivalence classes (ANECs) \cite{Egecioglu2009,Egecioglu2013}. Yet there are still too many ANECs if we have more than a few candidates and voters. Table \ref{NumObjects} illustrates this for a mere 6 voters and up to 5 candidates (the relevant number for the possibility result in Theorem \ref{Impossibility1}.\ref{Impossibility1b}).

In fact, we know that profiles with at least 11 voters are needed to prove Theorem \ref{Impossibility4}, because the Split Cycle voting method satisfies anonymity, neutrality, binary $\gamma$, and quasi-resoluteness on all profiles with fewer than 11 voters. For as shown in \cite{HP2020b},  Split Cycle satisfies anonymity, neutrality, and binary $\gamma$ on all profiles, and it satisfies quasi-resoluteness on profiles with fewer than 4 candidates. Then since no 4-candidate profile with fewer than 11 voters is uniquely weighted, Split Cycle vacuously satisfies quasi-resoluteness on all 4-candidate profiles with fewer than 11 voters. Thus, profiles with at least 11 voters are indeed needed to prove Theorem \ref{Impossibility4}, rendering a purely profile-based SAT strategy infeasible.

\begin{table}[h]
\begin{center}
\begin{tabular} {c|c|c|c|c}
 & &&6-voter& \\
 & unlabeled & 6-voter  & anonymous & 6-voter \\
candidates & weak tournaments & ANECs & profiles & profiles  \\
\hline 

3  & 7 & 83  & 462  & 46,656 \\
4  & 42 & 19,941 & 475,020 & 191,102,976\\
5  & 582 & 39,096,565 & 4,690,625,500 & 2,985,984,000,000\\

\end{tabular}
\end{center}
\caption{Comparison of the numbers of \textit{unlabeled weak tournaments}, \textit{preference profiles}, \textit{anonymous profiles}, and \textit{anonymous and neutral equivalence classes of profiles} (ANECs). We chose 6 voters because for each $n\in \{3,4,5\}$, $6$ is the least number of voters such that all $n$-candidate unlabeled weak tournaments are isomorphic to the majority graph of a profile with at most that number of voters \cite{Stearns1959}.}\label{NumObjects}
\end{table}

Fortunately, there is a way to circumvent the problem posed by too many profiles. This involves switching from profiles to abstract majority graphs and abstract margin graphs, not associated with any particular profiles.

An abstract majority graph is simply an asymmetric directed graph, which we will call a \textit{weak tournament} $T=(X,\to)$.\footnote{Here we use the term `weak tournament' as in \cite{Brandt2007}.} A \textit{tournament} is a weak tournament such that for all $a,b\in X$ with $a\neq b$, we have $a\to b$ or $b\to a$. 

An abstract margin graph is simply a \textit{weighted weak tournament} $\wwt=(T,w)$, where $\wt$ is a weak tournament and $w$ is a function assigning to each edge in $T$ a positive integer weight, satisfying the parity constraint of Remark \ref{Parity}. A \textit{weighted tournament} is a weighted weak tournament $\wwt=(T,w)$ in which $T$ is a tournament. We say that $\wwt$ is \textit{uniquely weighted} if for all $(a,b),(c,d)\in X(\wt)^2$ with $a\neq b$, $c\neq d$, and $(a,b)\neq (c,d)$, we have $w(a,b)\neq w(c,d)$.

We now explain our proof strategy in two steps.\\

\textbf{Step 1 for Theorem \ref{Impossibility1}.\ref{Impossibility1a}}. We prove an a priori weaker version of Theorem~\ref{Impossibility1}.\ref{Impossibility1a} using the additional assumption that the voting method's choice of winners depends only on the majority graph of the profile.

\begin{definition}\label{Invariance} A voting method $F$  is \textit{majoritarian} if for all $\mathbf{P},\mathbf{P}'\in\mathrm{dom}(F)$, 
\[M(\mathbf{P})=M(\mathbf{P}')\mbox{ implies }F(\mathbf{P})=F(\mathbf{P}').\]
\end{definition}

\noindent If $F$ is majoritarian, then we can think of $F$ as a function that assigns to each weak tournament a subset of the nodes. If $F$ is also neutral, then we can be even more economical, using a strategy from \cite{Brandt2016}. For each isomorphism class of weak tournaments, we fix a canonical representative of the class. Then we regard $F$ as a function that assigns to each canonical weak tournament $T$ a subset of the nodes. In fact, it turns out that our impossibility theorem will hold even if we restrict attention to just canonical \textit{tournaments}. To prove there is no majoritarian voting method satisfying the axioms of Theorem \ref{Impossibility1}.\ref{Impossibility1a}, we can follow the SAT solving strategy sketched above but using a propositional language $L_{maj}$ that contains for each canonical tournament $T$ and subset $Y$ of nodes a propositional variable $A_{T,Y}$, interpreted to mean that in $T$, the set of winners is $Y$. Then we  write a CNF formula $\psi_{6}$ that is satisfiable if and only if there is majoritarian voting method on profiles (whose majority graphs are tournaments) up to $6$ candidates satisfying the axioms listed in Theorem \ref{Impossibility1}.\ref{Impossibility1a}. Since there are far fewer canonical tournaments with up to 6 nodes than there are profiles with up to 6 candidates and even  a few voters, the formula $\psi_6$ will have far fewer variables than the formulas $\varphi_{n,6}$ of $L_{prof}$ mentioned above. Indeed, we successfully carry out this strategy in Appendix \ref{SATmajoritarian} to prove the following.

\begin{restatable}{proposition}{fromsatthmone}\label{FromSatThm1} For any finite $Y\subseteq\mathcal{X}$ with $|Y|\geq 6$, there is no \emph{majoritarian} voting method with domain $\{\mathbf{P}\mid X(\mathbf{P})\subseteq Y\}$  satisfying anonymity, neutrality, binary quasi-resoluteness, binary $\gamma$, and $\alpha$-resoluteness.
\end{restatable}
\noindent By contrast, as stated in Theorem \ref{Impossibility1}.\ref{Impossibility1b}, there is such a voting method on the domain $\{\mathbf{P}\mid |X(\mathbf{P})|\leq 5\}$. After obtaining Proposition \ref{FromSatThm1} by SAT solving, we extracted an MUS from the unsatisfiable CNF formula, but we were unable to extract from it a proof with few enough cases to make a human-readable proof practical. For details, see Appendix~\ref{HumanReadable1}.\\

\textbf{Step 1 for Theorem \ref{Impossibility4}}. Toward proving Theorem \ref{Impossibility4}, we will first prove an a priori weaker version using the additional assumption that the voting method's choice of winners depends only on the margin graph of the profile.

\begin{definition} A voting method $F$ is \textit{pairwise} if for all $\mathbf{P},\mathbf{P}'\in\mathrm{dom}(F)$, \[\mathcal{M}(\mathbf{P})=\mathcal{M}(\mathbf{P}')\mbox{ implies }F(\mathbf{P})=F(\mathbf{P}').\]
\end{definition}

\noindent If $F$ is pairwise, then we can think of $F$ as a function that assigns a set of nodes to each weighted weak tournament satisfying the parity constraint in Remark \ref{Parity}. If $F$ is also neutral, then we can fix for each isomorphism class of such weighted weak tournaments a canonical representative and then regard $F$ as a function that assigns to each canonical weighted weak tournament a subset of the nodes. Then following a SAT strategy like the one sketched above for majoritarian methods, we prove the following in Appendix \ref{SATQM}.

\begin{restatable}{proposition}{fromqmfinal}\label{FromQMFinal}For any $Y\subseteq\mathcal{X}$ with $|Y|\geq 4$, there is no \emph{pairwise} voting method  with domain
\begin{eqnarray*}&&\{\mathbf{P}\in\prof\mid X(\mathbf{P})\subseteq Y,\mathcal{M}(\mathbf{P})\mbox{ is uniquely weighted, }\\
&&\qquad\qquad\;\;\;\mbox{and all positive weights belong to }\{2,4,6,8,10,12\}\}\end{eqnarray*}
satisfying anonymity, neutrality, binary $\gamma$, and quasi-resoluteness.
\end{restatable}
\noindent By contrast, as previously noted, there are such voting methods on the domain $\{\mathbf{P}\mid |X(\mathbf{P})|\leq 3\}$. After obtaining Proposition \ref{FromQMFinal}  by SAT solving, we extracted an MUS from the unsatisfiable CNF formula, but in this case the distance from the MUS to a practical human-readable proof was even greater than for Proposition \ref{FromSatThm1}. For details, see Appendix \ref{HumanReadable2}. To increase confidence in our SAT-based proof of Proposition \ref{FromQMFinal}, we also formally verified a SAT encoding in the Lean Theorem Prover, as described in Appendix \ref{Verification}.\\

\textbf{Step 2 for Theorems \ref{Impossibility1}.\ref{Impossibility1a} and \ref{Impossibility4}}. We extend Propositions \ref{FromSatThm1} and \ref{FromQMFinal} for majoritarian and pairwise voting methods to impossibility theorems for \textit{all} voting methods using the following \textit{transfer lemmas}.

 \begin{restatable}{lemma}{TransferLemOne}\label{transferlemone} For any finite $Y\subseteq\mathcal{X}$, if there is a voting method with domain  $\{\mathbf{P}\in\prof\mid X(\mathbf{P})\subseteq Y\}$ satisfying anonymity, neutrality, binary quasi-resoluteness, binary $\gamma$, and $\alpha$-resoluteness, then there is a \emph{majoritarian} voting method with domain $\{\mathbf{P}\in\prof\mid X(\mathbf{P})\subseteq Y\}$ satisfying anonymity, neutrality, binary quasi-resoluteness, binary $\gamma$, and $\alpha$-resoluteness.
\end{restatable}

 \begin{restatable}{lemma}{TransferLemTwo}\label{transferlemtwo} For any finite $Y\subseteq\mathcal{X}$, if there is a voting method with domain $\{\mathbf{P}\in\prof\mid X(\mathbf{P})\subseteq Y\}$ satisfying anonymity, neutrality, binary $\gamma$, and quasi-resoluteness, then for any $m\in\mathbb{Z}^+$, there is a \emph{pairwise} voting method with domain \[\{\mathbf{P}\in\prof\mid X(\mathbf{P})\subseteq Y,\,m\mbox{ is the maximum weight in }\mathcal{M}(\mathbf{P})\}\] satisfying anonymity, neutrality, binary $\gamma$, and quasi-resoluteness.
\end{restatable}

 Combining Propositions \ref{FromSatThm1} and \ref{FromQMFinal} with Lemmas \ref{transferlemone} and \ref{transferlemtwo} establishes Theorems \ref{Impossibility1}.\ref{Impossibility1a} and \ref{Impossibility4}, respectively. Note that to relate Proposition \ref{FromQMFinal} to Lemma \ref{transferlemtwo}, we use the fact that if there is a pairwise voting method $F$ on some domain $D$ satisfying the listed axioms, then for any $D'\subseteq D$, the restriction of $F$ to $D'$ also satisfies the listed axioms.\\

\textbf{Proof strategy for the transfer lemmas}.  Our high-level strategy to prove the transfers lemmas is the following, which we will sketch for Lemma \ref{transferlemone}. Where $\mathsf{VM}$ is the class of all voting methods and $\mathsf{maj}$ is the class of majoritarian voting methods, we want a map
\[\pi_{_\mathsf{maj}}: \mathsf{VM}\to \mathsf{maj}\]
that preserves the axioms in Theorem \ref{Impossibility1}.\ref{Impossibility1a}, i.e., if $F\in \mathsf{VM}$ satisfies the axioms, then so does $\pi_{_\mathsf{maj}}(F)$. Our strategy is to obtain $\pi_{_\mathsf{maj}}$ as a composition of maps. First, we have the map $M$ that sends each profile to its majority graph. Second, we will choose an appropriate map $\phi$ that sends each majority graph to a profile $\mathbf{P}$. Then we will define $\pi_{_\mathsf{maj}}$ such that for each $F\in \mathsf{VM}$, the voting method $\pi_{_\mathsf{maj}}(F)$ is defined by
\[\pi_{_\mathsf{maj}}(F)(\mathbf{P})=F(\phi(M(\mathbf{P}))).\]
It is immediate that if $M(\mathbf{P})=M(\mathbf{P}')$, then $\pi_{_\mathsf{maj}}(F)(\mathbf{P})=\pi_{_\mathsf{maj}}(F)(\mathbf{P}')$, which shows $\pi_{_\mathsf{maj}}(F)\in \mathsf{maj}$. The crucial fact to prove, however, is that $\pi_{_\mathsf{maj}}$ preserves the axioms of the relevant impossibility theorem. This requires a careful choice of $\phi$.

The properties required of $\phi$ depend on the specific axioms in the impossibility theorem. As an example, let us consider the axiom of neutrality. Recall from Section \ref{IndividualToSocial} that a voting method $F$ satisfies neutrality if for any profile $\mathbf{P}$, 
 \[F(\mathbf{P}_{a\leftrightarrows b})=F(\mathbf{P})_{a\leftrightarrows b}.\]
Given a majority graph $G$ and vertices $a$ and $b$, let $G_{a\leftrightarrows b}$ be the isomorphic copy of $G$ with $a$ and $b$ swapped. Then in order to show that $\phi$ preserves neutrality, we would like to show that for any majority graph $G$, we have
\[\phi(G_{a\leftrightarrows b})=\phi(G)_{a\leftrightarrows b}.\]
As we shall see, this commutativity requirement rules out a standard construction of profiles from majority graphs \cite{McGarvey1953} as a candidate for $\phi$.

Though we have sketched the strategy above for proving the majoritarian transfer lemma, Lemma \ref{transferlemone}, an analogous strategy applies to proving the pairwise transfer lemma, Lemma \ref{transferlemtwo}, using an appropriate map $\pi_{\mathrm{pairwise}}$. In the next three sections, we carry out this strategy for proving the transfer lemmas.

\section{Profile representation of weighted weak tournaments}\label{RepresentationSection}

In this section, we introduce a method for representing an abstract majority/margin graph as the majority/margin graph of a profile, which will play the role of the $\phi$ map in our proof strategy outlined at the end of Section \ref{ProofStrategy}.

To prove properties of our representation, it will be convenient to use \textit{anonymous profiles}, which simply record how many voters submit each ranking of the candidates, as reviewed in Section \ref{AnonProf}. In Section \ref{McGarveyDebord}, we then review the standard constructions of (anonymous) profiles from majority and margin graphs due to McGarvey and Debord. We introduce our alternative construction in Section~\ref{AlternativeRep}.

\subsection{Anonymous profiles}\label{AnonProf}

Recall from Section \ref{IndividualToSocial} that $\mathcal{X}$ is an infinite set of candidates and that for $X\subseteq\mathcal{X}$,  $\mathcal{L}(X)$ is the set of all strict linear orders on $X$. 

It will be convenient for certain intermediate calculations to be able to work with a generalization of anonymous profiles allowing the number of voters who submit a certain ballot to be \textit{negative} (cf.~\cite[p.~135]{Saari2008}).

\begin{definition}\label{GenAnonProf} A \textit{generalized anonymous profile} is a function $\mathbf{P}$ such that for some finite $X(\mathbf{P})\subseteq\mathcal{X}$, we have ${\mathbf{P}:\mathcal{L}(X(\mathbf{P}))\to\mathbb{Z}}$. 

 An \textit{anonymous profile} is a generalized anonymous profile ${\mathbf{P}:\mathcal{L}(X(\mathbf{P}))\to\mathbb{N}}$. \end{definition}

Given $\mathbf{P}$ and $\mathbf{P}'$ with $X(\mathbf{P})=X(\mathbf{P}')$, we define $\mathbf{P}+\mathbf{P}'$ with $X(\mathbf{P}+\mathbf{P}')=X(\mathbf{P})$ (resp.~$\mathbf{P}-\mathbf{P}'$ with $X(\mathbf{P}-\mathbf{P}')=X(\mathbf{P})$) such that for all $L\in \mathcal{L}(X(\mathbf{P}))$, \[ (\mathbf{P}+\mathbf{P}')(L)=\mathbf{P}(L)+\mathbf{P}(L')\mbox{ and }(\mathbf{P}-\mathbf{P}')(L)=\mathbf{P}(L)-\mathbf{P}(L').\]
Given $n\in\mathbb{Z}$, we define $n\mathbf{P}$ with $X(n\mathbf{P})=X(\mathbf{P})$ such that for all $L\in \mathcal{L}(X(\mathbf{P}))$, \[(n\mathbf{P})(L)=n \times \mathbf{P}(L).\]

For $a,b\in X(\mathbf{P})$, we define the \textit{margin of $a$ over $b$ in $\mathbf{P}$}, denoted $\mu_{a,b}(\mathbf{P})$, by:
\begin{eqnarray*}
\mu_{a,b}(\mathbf{P})&=& \bigg(\sum_{L\in\mathcal{L}(X(\mathbf{P})):\, aLb}\mathbf{P}(L)\bigg)- \bigg(\sum_{L\in\mathcal{L}(X(\mathbf{P})):\, bLa}\mathbf{P}(L)\bigg).
\end{eqnarray*}
The \textit{majority graph} $M(\mathbf{P})$ and \textit{margin graph} $\mathcal{M}(\mathbf{P})$ of an anonymous profile $\mathbf{P}$ are defined just as for ordinary profiles in Section \ref{IndividualToSocial}. The following is obvious.

\begin{lemma} For generalized anonymous profiles $\mathbf{P},\mathbf{P}'$ with $X(\mathbf{P})=X(\mathbf{P}')$ and ${a,b\in X(\mathbf{P})}$, $\mu_{a,b}(\mathbf{P}+\mathbf{P}')=\mu_{a,b}(\mathbf{P})+\mu_{a,b}(\mathbf{P}')$ and $\mu_{a,b}(\mathbf{P}-\mathbf{P}')=\mu_{a,b}(\mathbf{P})-\mu_{a,b}(\mathbf{P}')$.
\end{lemma}

\subsection{McGarvey and Debord's constructions}\label{McGarveyDebord} The classic constructions of profiles from majority graphs and margin graphs are due to McGarvey \cite{McGarvey1953}  and Debord \cite{Debord1987}. We review these constructions for comparison with our alternative construction in Section \ref{AlternativeRep}.

Henceforth we fix a distinguished strict linear order $L_\star$ on the set $\mathcal{X}$. 

\begin{definition}\label{McGarveyDef} For any weighted weak tournament $\wwt$ such that all weights are even numbers, we define the profile $\mathsf{Deb}(\wwt)$ as follows. For each $a,b\in X(\wt)$ such that $a\to b$ in $\wwt$ with weight $n$, where $x_1,\dots,x_k$ enumerates ${X(\wwt)\setminus\{a,b\}}$ according to the order $L_\star$, we add $\frac{n}{2}$ voters with the linear order $abx_1\dots x_k$ and $\frac{n}{2}$ voters with the linear order $x_k\dots x_1ab$. 

For a weak tournament $\wt$, let $\wt_2$ be the weighted weak tournament based on $\wt$ in which each edge has weight $2$, and define $\mathsf{Mc}(\wt)=\mathsf{Deb}(T_2)$.
\end{definition}

Though it will not play a role in our proofs, we mention how to handle weighted tournaments with weights of odd parity. There is an obvious correspondence between weighted weak tournaments and skew-symmetric matrices. Given weighted weak tournaments $\mathcal{T}_1$ and $\mathcal{T}_2$, let  $\wwt_1-\wwt_2$ be the weighted weak tournament corresponding to the matrix obtained by subtracting the matrix for $\mathcal{T}_2$ from that for~$\mathcal{T}_1$. In addition, for each finite set $X$ of candidates, let profile $\mathbf{P}_{\star,X}$ be the profile with only one voter who submits the restriction to $X$  of $L_\star$.

\begin{definition} For a weighted tournament $\wwt$ with odd weights, we define the profile $\mathsf{Deb}(\wwt)$ as follows. As $\wwt-\mathcal{M}(\mathbf{P}_{\star, X(\wwt)})$ has even number weights, we set \[\mathsf{Deb}(\wwt)=\mathsf{Deb}(\wwt-\mathcal{M}(\mathbf{P}_{\star, X(\wwt)}))+\mathbf{P}_{\star, X(\wwt)}.\] 
\end{definition}

Then it is easy to see that $\mathcal{M}(\mathsf{Deb}(\wwt))=\wwt$, which yields the following.

\begin{theorem}\label{McGarveyThm}  $\,$
\begin{enumerate}
\item\label{McGarveyThm1} \cite{McGarvey1953} Every weak tournament is the majority graph of a profile.
\item \cite{Debord1987} If $\wwt=(T,w)$ is a weighted weak tournament in which all weights have the same parity, which is even if $T$ is not a tournament, then $\wwt$ is the margin graph of a profile.
\end{enumerate}
\end{theorem}

\subsection{An alternative representation}\label{AlternativeRep}

The key to proving the first transfer lemma, Lemma \ref{transferlemone}, is an alternative to McGarvey's representation of weak tournaments by profiles. The basic idea of the representation is simple. To realize a given weak tournament $T$, start with a profile consisting of one copy of each linear order over the candidates; then if there is an edge in $T$ from candidate $a$ to candidate $b$, modify the profile such that every ballot with $ba$ on top becomes a ballot with $ab$ on top, thereby creating an edge from $a$ to $b$ in the majority graph of the profile.

In fact, we need to be more careful: in order to ensure that our profile representation commutes with the operation of restricting to a set of candidates, we will realize a given majority graph $T$ \textit{relative to a set $Y$} of candidates, containing all the candidates in $T$. Thus, we implement the idea above starting with a profile consisting of one copy of each linear order over the candidates in $Y$. The following notation will facilitate our proofs about this construction.

\begin{definition} Given a finite $Y\subseteq\mathcal{X}$, let  $\mathbf{L}_Y$ be the profile such that $\mathbf{L}_Y(L)=1$ for each linear order $L$ on $Y$. For $a,b\in Y$, let $\mathbf{L}_{ab\,\mbox{-}Y}$ be the profile such that $\mathbf{L}_{ab\,\mbox{-}Y}(L)=1$ if $L$ is a linear order on $Y$ of the form $aby_1\dots y_n$ for $y_1,\dots,y_n\in Y\setminus\{a,b\}$ and $\mathbf{L}_{ab\,\mbox{-}Y}(L)=0$ otherwise. We call $\mathbf{L}_{ab\,\mbox{-}Y}$ an \textit{$ab$-$\mbox{Y-}$block}.
\end{definition}

Now we officially define our representation of weak tournaments.

\begin{definition}\label{RepDef} Given a finite $Y\subseteq\mathcal{X}$ and weak tournament $\wt$ with $X(\wt)\subseteq Y$, we define the \textit{profile representation of $\wt$ relative to $Y$}, $\mathfrak{P}_Y(\wt)$, as follows. Start with the profile $\mathbf{L}_{Y}$. For each pair $a,b\in X(\wt)$ with $a\to b$ in $\wt$, flip the first two candidates of the \textit{$ba$-$\mbox{Y-}$block} to obtain an additional \textit{$ab$-$\mbox{Y-}$block}, resulting in a profile $\mathfrak{L}_Y(\wt)$. More formally, let

\begin{equation}\mathfrak{L}_{Y}(\wt)= \mathbf{L}_Y + \underset{\mbox{ {\footnotesize $a\to b$ in $T$}}}{\sum} (\mathbf{L}_{ab\,\mbox{-}Y} - \mathbf{L}_{ba\,\mbox{-}Y}).\label{LYEq}\end{equation}

\noindent Then let 
\begin{equation}\mathfrak{P}_Y(\wt)=\mathfrak{L}_Y(\wt)|_{ X(\wt)}.\label{PfromL}\end{equation}
\end{definition}

It is easy to see that like McGarvey's construction, Definition \ref{RepDef} allows us to represent any weak tournament as the majority graph of a profile. Let $\mathbb{W}\mathbb{T}_Y$ be the set of all weak tournaments $\wt$ with $X(\wt)\subseteq Y$.

\begin{proposition}\label{RepProp0}For any finite $Y\subseteq\mathcal{X}$, $\wt\in \mathbb{W}\mathbb{T}_Y$, we have $\wt =M(\mathfrak{P}_Y(\wt))$. 
\end{proposition}
\noindent This proposition follows from Propositions \ref{RepSpecial} and \ref{RepProp} proved below.

For the purposes of the second transfer lemma, Lemma \ref{transferlemtwo}, we now define a similar construction for certain weighted weak tournaments. The idea is simply to start with a profile that has sufficiently many copies of each linear order, and then to create an edge from $a$ to $b$ with a certain margin, we turn an appropriate number of $ba$-$Y$ blocks into $ab$-$Y$ blocks.

Note that given a finite $Y\subseteq\mathcal{X}$ and $a,b\in Y$, the number of ballots in an \textit{$ab$-$\mbox{Y-}$block} is $(|Y|-2)!$. Later it will be useful to express margins as multiples of 
\[\psi_Y = 2\times (|Y|-2)!,\]
in which case the number of ballots in an \textit{$ab$-$\mbox{Y-}$block} is $\psi_Y/2$. When $Y$ is clear from context, we will write $\psi$ for $\psi_Y$.

\begin{definition}\label{RepDefWeighted} Given a finite $Y\subseteq\mathcal{X}$, $m\in\mathbb{Z}^+$, and weighted weak tournament $\wwt$ with $X(\wwt)\subseteq Y$ each of whose weights is $k\psi_Y$ for $k\in\mathbb{Z}$ and less than or equal to $m\psi$, we define the \textit{profile representation of $\wt$ relative to $Y$ and $m$}, $\mathfrak{P}_{Y,m}(\wwt)$, as follows. Start with  $m\mathbf{L}_{Y}$, so for each $a,b\in X(\wwt)$, there are $m$ \textit{$ab$-$\mbox{Y-}$}blocks in $m\mathbf{L}_{Y}$. For each pair $a,b$ with $a\to b$ in $\wwt$ with weight $k\psi_Y$, in $k$ \textit{$ba$-$\mbox{Y-}$}blocks in $m\mathbf{L}_{Y}$, flip the first two candidates to obtain an additional $k$ \textit{$ab$-$\mbox{Y-}$}blocks, resulting in a profile $\mathfrak{L}_{Y,m}(\wwt)$.  More formally,
\begin{equation}\mathfrak{L}_{Y,m}(\wwt)= m\mathbf{L}_Y + \underset{\mbox{ {\footnotesize $a\overset{k\psi}{\to} b$ in $\wwt$}}}{\sum} (k\mathbf{L}_{ab\,\mbox{-}Y} - k\mathbf{L}_{ba\,\mbox{-}Y}).\label{LYmEq}\end{equation}

\noindent Then let \begin{equation}\mathfrak{P}_{Y,m}(\wwt)=\mathfrak{L}_{Y,m}(\wwt)|_{ X(\wwt)}.\label{PYmEq}\end{equation}
\end{definition}

\begin{figure}[ht]

\begin{center}\begin{minipage}{.65in}$\wt$\end{minipage}\begin{minipage}{4.5in}\begin{tikzpicture}

\node[circle,draw, minimum width=0.25in] at (0,0) (a) {$a$}; 
\node[circle,draw,minimum width=0.25in] at (3,0) (c) {$c$}; 
\node[circle,draw,minimum width=0.25in] at (1.5,1.5) (b) {$b$}; 

\path[->,draw,thick,medgreen] (b) to node[fill=white] {$8$} (c);
\path[->,draw,thick,blue] (c) to node[fill=white] {$12$} (a);
\path[->,draw,thick,red] (a) to node[fill=white] {$4$} (b);

\end{tikzpicture}
\end{minipage} 

\bigskip

\begin{center}\begin{minipage}{.65in}$6\mathbf{L}_Y$\end{minipage}\begin{minipage}{4.5in} {\setlength{\tabcolsep}{2.5pt}\begin{tabular}{cc!{\color{\gray}\vrule}cc!{\color{\gray}\vrule}cc!{\color{\gray}\vrule}cc!{\color{\gray}\vrule}cc!{\color{\gray}\vrule}cc!{\color{\gray}\vrule}cc!{\color{\gray}\vrule}cc!{\color{\gray}\vrule}cc!{\color{\gray}\vrule}cc!{\color{\gray}\vrule}cc!{\color{\gray}\vrule}cc}
$6$ & $6$ & $6$ & $6$ & $6$ & $6$ & $6$ & $6$ & $6$ & $6$ & $6$ & $6$ & $6$ & $6$ & $6$ & $6$ & $6$ & $6$ & $6$ & $6$ & $6$ & $6$ & $6$ & $6$    \\\hline
$a$ & $a$ &  $a$ & $a$ & $a$ & $a$ & $b$ & $b$ &  $b$ & $b$ & $b$ & $b$ & $c$ & $c$ &  $c$ & $c$ & $c$ & $c$ & $d$ & $d$ &  $d$ & $d$ & $d$ & $d$ \\
$b$ &  $b$ & $c$ & $c$ & $d$ & $d$		&$a$ & $a$ & $c$ & $c$ & $d$ & $d$	&$a$ & $a$ & $b$ & $b$ & $d$ & $d$	& $a$ & $a$ &  $b$ & $b$ & $c$ & $c$ \\
$c$ &  $d$ &  $b$ & $d$ & $b$ & $c$	 & $c$ & $d$ & $a$ & $d$ & $a$ & $c$	&$b$ & $d$ & $a$ & $d$ & $a$ & $b$	& $b$ &  $c$ & $a$ & $c$ & $a$ & $b$ \\
$d$ &  $c$ &  $d$ & $b$ & $c$ & $b$	& $d$ & $c$ & $d$ & $a$ & $c$ & $a$	&$d$ & $b$ & $d$ & $a$ & $b$ & $a$	& $c$ &  $b$ &  $c$ & $a$ & $b$ & $a$  \\
\end{tabular}} \end{minipage} 
\end{center}

\bigskip
\begin{center}\begin{minipage}{.65in}$\mathfrak{L}_{Y,6}(\wt)$\end{minipage}\begin{minipage}{4.5in} {\setlength{\tabcolsep}{2.5pt}\begin{tabular}{cc!{\color{\gray}\vrule}cccc!{\color{\gray}\vrule}cc!{\color{\gray}\vrule}cccc!{\color{\gray}\vrule}cc!{\color{\gray}\vrule}cc!{\color{\gray}\vrule}cc!{\color{\gray}\vrule}cccc!{\color{\gray}\vrule}cc!{\color{\gray}\vrule}cc!{\color{\gray}\vrule}cc!{\color{\gray}\vrule}cc}
$6$ & $6$ & $3$ & $3$ & $3$ & $3$ 					& $6$ & $6$ & $1$ & $1$ & $5$ & $5$ 					& $6$ & $6$ & $6$ & $6$ & $6$ & $6$ 	& $2$ & $2$ & $4$ & $4$ & $6$ & $6$ & $6$&$6$&$6$&$6$&$6$&$6$  \\\hline
$a$ & $a$ &  $\color{blue}{c}$ & $\color{blue}{c}$ &$a$ &$a$ & $a$ & $a$ 	& $\color{red}{a}$ & $\color{red}{a}$ & $b$ & $b$ & $b$ & $b$ & $b$ & $b$ 	& $c$ & $c$ &  $\color{medgreen}{b}$ & $\color{medgreen}{b}$ & $c$ & $c$ & $c$ & $c$ 	& $d$ & $d$ &  $d$ & $d$ & $d$ & $d$ \\
$b$ &  $b$ & $\color{blue}{a}$ & $\color{blue}{a}$ & $c$ & $c$ & $d$ & $d$	&$\color{red}{b}$ & $\color{red}{b}$ & $a$ & $a$ & $c$ & $c$ & $d$ & $d$	&$a$ & $a$ & $\color{medgreen}{c}$ & $\color{medgreen}{c}$ & $b$ & $b$ & $d$ & $d$	& $a$ & $a$ &  $b$ & $b$ & $c$ & $c$ \\
$c$ &  $d$ &  $b$ & $d$ &  $b$ & $d$ & $b$ & $c$		 			& $c$ & $d$ & $c$ & $d$ & $a$ & $d$ & $a$ & $c$					&$b$ & $d$ & $a$ & $d$ & $a$ & $d$ & $a$ & $b$	& $b$ &  $c$ & $a$ & $c$ & $a$ & $b$ \\
$d$ &  $c$ &  $d$ & $b$  &  $d$ & $b$ & $c$ & $b$					& $d$ & $c$  & $d$ & $c$ & $d$ & $a$ & $c$ & $a$					&$d$ & $b$ & $d$ & $a$ & $d$ & $a$ & $b$ & $a$	& $c$ &  $b$ &  $c$ & $a$ & $b$ & $a$  \\
\end{tabular}} \end{minipage} 
\end{center}

\bigskip
\begin{center}\begin{minipage}{.65in}$\mathfrak{P}_{Y,6}(\wt)$\end{minipage}\begin{minipage}{4.5in} {\setlength{\tabcolsep}{2.5pt}\begin{tabular}{cc!{\color{\gray}\vrule}cccc!{\color{\gray}\vrule}cc!{\color{\gray}\vrule}cccc!{\color{\gray}\vrule}cc!{\color{\gray}\vrule}cc!{\color{\gray}\vrule}cc!{\color{\gray}\vrule}cccc!{\color{\gray}\vrule}cc!{\color{\gray}\vrule}cc!{\color{\gray}\vrule}cc!{\color{\gray}\vrule}cc}

$6$ & $6$ & $3$ & $3$ & $3$ & $3$ 					& $6$ & $6$ & $1$ & $1$ & $5$ & $5$ 					& $6$ & $6$ & $6$ & $6$ & $6$ & $6$ 	& $2$ & $2$ & $4$ & $4$ & $6$ & $6$ & $6$&$6$&$6$&$6$&$6$&$6$  \\\hline
$a$ & $a$ &  $\color{blue}{c}$ & $\color{blue}{c}$ & $a$ & $a$ & $a$ & $a$ 	& $\color{red}{a}$ & $\color{red}{a}$ & $b$ & $b$ & $b$ & $b$ & $b$ & $b$ 	& $c$ & $c$ &  $\color{medgreen}{b}$ & $\color{medgreen}{b}$ & $c$ & $c$ & $c$ & $c$ 	& $a$ & $a$ &  $b$ & $b$ & $c$ & $c$ \\
$b$ &  $b$ & $\color{blue}{a}$ & $\color{blue}{a}$ & $c$ & $c$ & $b$ & $c$	&$\color{red}{b}$ & $\color{red}{b}$ & $a$ & $a$ & $c$ & $c$ & $a$ & $c$	&$a$ & $a$ & $\color{medgreen}{c}$ & $\color{medgreen}{c}$ & $b$ & $b$ & $a$ & $b$	& $b$ & $c$ &  $a$ & $c$ & $a$ & $b$ \\
$c$ &  $c$ &  $b$ & $b$ &  $b$ & $b$  & $c$ & $b$		 			& $c$ & $c$ & $c$ & $c$ & $a$ & $a$ & $c$ & $a$					&$b$ & $b$ & $a$ & $a$ & $a$ & $a$ & $b$ & $a$								& $c$ &  $b$ & $c$ & $a$ & $b$ & $a$ \\
\end{tabular}} \end{minipage}
\end{center}

\end{center}
\bigskip

\begin{center}\begin{minipage}{.65in}$\mathsf{Deb}(\wt)$\end{minipage}\begin{minipage}{4.5in} \begin{tabular}{cccccc}
$2$ & $2$ & $4$ & $4$ & $6$ & $6$   \\\hline
$\color{red}{a}$ & $c$ &  $\color{medgreen}{b}$ & $a$ & $\color{blue}{c}$ & $b$ \\
$\color{red}{b}$ &  $\color{red}{a}$ & $\color{medgreen}{c}$ & $\color{medgreen}{b}$ & $\color{blue}{a}$ & $\color{blue}{c}$ \\
$c$ &  $\color{red}{b}$ &  $a$ & $\color{medgreen}{c}$ & $b$ & $\color{blue}{a}$ \\
\end{tabular}\end{minipage} 
\end{center}

\caption{Illustration of Definitions \ref{RepDef} and \ref{McGarveyDef} with $Y=\{a,b,c,d\}$}\label{RepFigure}
\end{figure}

\begin{figure}[ht]

\begin{center}\begin{minipage}{.65in}$\wt$\end{minipage} \begin{minipage}{4.25in}\begin{tikzpicture}

\node[circle,draw, minimum width=0.25in] at (0,0) (a) {$a$}; 
\node[circle,draw,minimum width=0.25in] at (3,0) (c) {$c$}; 
\node[circle,draw,minimum width=0.25in] at (1.5,1.5) (b) {$b$}; 
\node[circle,draw,minimum width=0.25in] at (1.5,-1.5) (d) {$d$}; 

\path[->,draw,thick,medgreen] (b) to node[fill=white] {$8$} (c);
\path[->,draw,thick,blue] (c) to node[fill=white,pos=.7] {$12$} (a);
\path[->,draw,thick,red] (a) to node[fill=white] {$4$} (b);
\path[->,draw,thick,darkyellow] (b) to node[fill=white,pos=.7] {$24$} (d);
\path[->,draw,thick,purple] (d) to node[fill=white] {$16$} (a);
\path[->,draw,thick,orange] (d) to node[fill=white] {$20$} (c);

\end{tikzpicture}
\end{minipage}

\bigskip

\begin{center}\begin{minipage}{.65in}$6\mathbf{L}_Y$\end{minipage} \begin{minipage}{4.25in}
{\setlength{\tabcolsep}{2.75pt}\begin{tabular}{cc!{\color{\gray}\vrule}cc!{\color{\gray}\vrule}cc!{\color{\gray}\vrule}cc!{\color{\gray}\vrule}cc!{\color{\gray}\vrule}cc!{\color{\gray}\vrule}cc!{\color{\gray}\vrule}cc!{\color{\gray}\vrule}cc!{\color{\gray}\vrule}cc!{\color{\gray}\vrule}cc!{\color{\gray}\vrule}cc}
$6$ & $6$ & $6$ & $6$ & $6$ & $6$ & $6$ & $6$ & $6$ & $6$ & $6$ & $6$ & $6$ & $6$ & $6$ & $6$ & $6$ & $6$ & $6$ & $6$ & $6$ & $6$ & $6$ & $6$    \\\hline
$a$ & $a$ &  $a$ & $a$ & $a$ & $a$ & $b$ & $b$ &  $b$ & $b$ & $b$ & $b$ & $c$ & $c$ &  $c$ & $c$ & $c$ & $c$ & $d$ & $d$ &  $d$ & $d$ & $d$ & $d$ \\
$b$ &  $b$ & $c$ & $c$ & $d$ & $d$		&$a$ & $a$ & $c$ & $c$ & $d$ & $d$	&$a$ & $a$ & $b$ & $b$ & $d$ & $d$	& $a$ & $a$ &  $b$ & $b$ & $c$ & $c$ \\
$c$ &  $d$ &  $b$ & $d$ & $b$ & $c$	 & $c$ & $d$ & $a$ & $d$ & $a$ & $c$	&$b$ & $d$ & $a$ & $d$ & $a$ & $b$	& $b$ &  $c$ & $a$ & $c$ & $a$ & $b$ \\
$d$ &  $c$ &  $d$ & $b$ & $c$ & $b$	& $d$ & $c$ & $d$ & $a$ & $c$ & $a$	&$d$ & $b$ & $d$ & $a$ & $b$ & $a$	& $c$ &  $b$ &  $c$ & $a$ & $b$ & $a$  \\
\end{tabular}} 
\end{minipage}
\end{center}

\bigskip

\begin{center}\begin{minipage}{.65in}$\mathfrak{L}_{Y,6}(\wt)$ \\ $\mbox{ }\mbox{ }\mbox{ }\mbox{ }\rotatebox{90}{$=$}$ \\ $\mathfrak{P}_{Y,6}(\wt)$\end{minipage} \begin{minipage}{4.25in}
{\setlength{\tabcolsep}{1.6pt}\begin{tabular}{cc!{\color{\gray}\vrule}cccc!{\color{\gray}\vrule}cccc!{\color{\gray}\vrule}cccc!{\color{\gray}\vrule}cc!{\color{\gray}\vrule}cc!{\color{\gray}\vrule}cc!{\color{\gray}\vrule}cccc!{\color{\gray}\vrule}cccc!{\color{\gray}\vrule}cc!{\color{\gray}\vrule}cccc!{\color{\gray}\vrule}cc}
$6$ & $6$ & $3$ & $3$ & $3$ & $3$ 					& $4$ & $4$ & $2$ & $2$ & $1$ & $1$ 					& $5$ & $5$ & $6$ & $6$ & $6$ & $6$ 	& $6$ & $6$ & $2$ & $2$ & $4$ & $4$ & $5$ & $5$ & $1$ & $1$ & $6$ & $6$ & $6$ & $6$ & $0$ & $0$ & $6$ & $6$   \\\hline
$a$ & $a$ &  $\color{blue}{c}$ & $\color{blue}{c}$ & $a$ & $a$ & $\color{purple}{d}$ & $\color{purple}{d}$ 	& $a$ & $a$ & $\color{red}{a}$ & $\color{red}{a}$ & $b$ & $b$ & $b$ & $b$ & $b$ & $b$ 	& $c$ & $c$ &  $\color{medgreen}{b}$ & $\color{medgreen}{b}$ & $c$ & $c$  & $\color{orange}{d}$ & $\color{orange}{d}$ & $c$ & $c$	& $d$ & $d$ &  $\color{darkyellow}{b}$ & $\color{darkyellow}{b}$ & $d$ & $d$ & $d$ & $d$ \\
$b$ &  $b$ & $\color{blue}{a}$ & $\color{blue}{a}$ & $c$ & $c$ & $\color{purple}{a}$ & $\color{purple}{a}$	& $d$ & $d$ & $\color{red}{b}$ & $\color{red}{b}$ & $a$ &$a$ & $c$ & $c$ & $d$ & $d$	&$a$ & $a$ & $\color{medgreen}{c}$ & $\color{medgreen}{c}$ & $b$& $b$ & $\color{orange}{c}$ & $\color{orange}{c}$ & $d$ & $d$ 	& $a$ & $a$ &  $\color{darkyellow}{d}$ & $\color{darkyellow}{d}$ & $b$ & $b$ & $c$ & $c$ \\
$c$ &  $d$ &  $b$ & $d$ &  $b$ & $d$ & $b$ & $c$	& $b$ & $c$	 			& $c$ & $d$ & $c$ & $d$ & $a$ & $d$ & $a$ & $c$					&$b$ & $d$ & $a$ & $d$ & $a$ & $d$ & $a$ & $b$ & $a$ & $b$ 	& $b$ &  $c$ & $a$ & $c$  & $a$ & $c$& $a$ & $b$ \\
$d$ &  $c$ &  $d$ & $b$ & $d$ & $b$ & $c$ & $b$	& $c$ & $b$			& $d$ & $c$ & $d$ & $c$ & $d$ & $a$ & $c$ & $a$					&$d$ & $b$ & $d$ & $a$ & $d$ & $a$ & $b$ & $a$ & $b$ & $a$	& $c$ &  $b$ &  $c$ & $a$  &  $c$ & $a$ & $b$ & $a$  \\
\end{tabular}} 
\end{minipage}
\end{center}

\end{center}

\bigskip

\begin{center}\begin{minipage}{.65in}$\mathsf{Deb}(\wt)$\end{minipage}\begin{minipage}{4.25in}
{\setlength{\tabcolsep}{2.75pt}\begin{tabular}{cccccccccccc}
$2$ & $2$ & $4$ & $4$ & $6$ & $6$ 					& $8$ & $8$ & $10$ & $10$ & $12$ & $12$ \\\hline 
 $\color{red}{a}$ & $d$ & $\color{medgreen}{b}$ &$d$ & $\color{blue}{c}$ & $d$ & $\color{purple}{d}$ & $c$  & \color{orange}{$d$} & $b$& \color{darkyellow}{$b$}& $c$\\
 $\color{red}{b}$& $c$ & $\color{medgreen}{c}$ &$a$& $\color{blue}{a}$ & $b$ & $\color{purple}{a}$& $b$ & \color{orange}{$c$} & $a$& \color{darkyellow}{$d$}& $a$\\
 $c$& $\color{red}{a}$ & $a$ & $\color{medgreen}{b}$& $b$ & $\color{blue}{c}$& $b$& $\color{purple}{d}$ & $a$ & \color{orange}{$d$} & $a$& \color{darkyellow}{$b$}\\
 $d$& $\color{red}{b}$ & $d$ & $\color{medgreen}{c}$& $d$ & $\color{blue}{a}$& $c$& $\color{purple}{a}$ & $b$ & \color{orange}{$c$}& $c$& \color{darkyellow}{$d$}
\end{tabular}}
\end{minipage}
\end{center}

\caption{Illustration of Definitions \ref{RepDef} and \ref{McGarveyDef} with $Y=\{a,b,c,d\}$}\label{RepFigure2}
\end{figure}

We can regard the construction in Definition \ref{RepDef} as a special case of that in Definition \ref{RepDefWeighted}.

\begin{proposition}\label{RepSpecial} Given a finite $Y\subseteq\mathcal{X}$ and weak tournament $\wt$ with $X(\wt)\subseteq Y$, let $T_\psi$ be the weighted tournament based on $T$ with all weights being $\psi_Y$. Then $\mathfrak{P}_Y(\wt)=\mathfrak{P}_{Y,1}(\wt_\psi)$.
\end{proposition}

\begin{proof} Since all weights in $T_\psi$ are $\psi_Y$, the $k$ in Definition \ref{RepDefWeighted} is 1. Then since $m=1$, (\ref{LYmEq}) reduces to (\ref{LYEq}).
\end{proof}

Figures \ref{RepFigure} and \ref{RepFigure2} illustrate our construction from Definition \ref{RepDefWeighted}, contrasting it with Debord's from Definition \ref{McGarveyDef}. This construction allows us to represent any weighted weak tournament in the following class as the margin graph of a profile.

\begin{definition} Given a finite $Y\subseteq\mathcal{X}$ and $m\in\mathbb{Z}^+$, let $\mathbb{W}\mathbb{T}^w_{Y,m}$ be the set of all weighted weak tournaments $\wwt$ with $X(\wwt)\subseteq Y$ such that  each weight in $\wwt$ is of the form  $k\psi_Y$ for some $k\in\mathbb{Z}$ and less than or equal to $m\psi_Y$.\end{definition}

\newpage

\begin{proposition}\label{RepProp} For any finite $Y\subseteq\mathcal{X}$, $m\in\mathbb{Z}^+$, and $\wwt\in \mathbb{W}\mathbb{T}^w_{Y,m}$, we have $\wwt =\mathcal{M}(\mathfrak{P}_{Y,m}(\wwt))$.
\end{proposition}

\begin{proof} Recall that $\mathfrak{P}_{Y,m}$ is defined by (\ref{LYmEq})--(\ref{PYmEq}). Concerning (\ref{LYmEq}), observe that for any $a,b\in X(\wwt)$:
\begin{enumerate}
\item[(i)] $\mu_{a,b}(m\mathbf{L}_Y)=0$;
\item[(ii)] if $a\overset{k\psi}{\to}b$ in $\wwt$, then $\mu_{a,b}(k\mathbf{L}_{ab\,\mbox{-}Y} - k\mathbf{L}_{ba\,\mbox{-}Y})= k\psi_Y$;
\item[(iii)] for any $(c,d)\neq (a,b)$ and $(d,c)\neq (a,b)$, if $c\overset{k\psi}{\to}d$, then  \\ $\mu_{a,b}(k\mathbf{L}_{cd\,\mbox{-}Y} - k\mathbf{L}_{dc\,\mbox{-}Y})= 0$.
\end{enumerate}
Part (i) is immediate from the definition of $\mathbf{L}_Y$. For (ii), we have:
\begin{eqnarray*}
\mu_{a,b}(k\mathbf{L}_{ab\,\mbox{-}Y} - k\mathbf{L}_{ba\,\mbox{-}Y}) & =&   \mu_{a,b}(k\mathbf{L}_{ab\,\mbox{-}Y}) - \mu_{a,b}(k\mathbf{L}_{ba\,\mbox{-}Y})\\
& =&   k\mu_{a,b}(\mathbf{L}_{ab\,\mbox{-}Y}) - k\mu_{a,b}(\mathbf{L}_{ba\,\mbox{-}Y})\\
&=&  k \frac{\psi_Y}{2 } - ( - k \frac{\psi_Y}{2 })\\ 
& = &  k\psi_Y.
\end{eqnarray*}
For (iii), we have:
\begin{eqnarray*}
\mu_{a,b}(k\mathbf{L}_{cd\,\mbox{-}Y} - k\mathbf{L}_{dc\,\mbox{-}Y}) & = &  k\mu_{a,b}(\mathbf{L}_{cd\,\mbox{-}Y}) - k\mu_{a,b}(\mathbf{L}_{dc\,\mbox{-}Y}).
\end{eqnarray*}
Recall that $(c,d)\neq (a,b)$ and $(d,c)\neq (a,b)$. If $\{c,d\}\cap \{a,b\}=\varnothing$, so $a,b\in Y$, then $\mu_{a,b}(\mathbf{L}_{cd\,\mbox{-}Y})=\mu_{a,b}(\mathbf{L}_{dc\,\mbox{-}Y})=0$, so $\mu_{a,b}(k\mathbf{L}_{cd\,\mbox{-}Y} - k\mathbf{L}_{dc\,\mbox{-}Y})= 0$. Now suppose $|\{c,d\}\cap \{a,b\}|=1$. Then without loss of generality suppose $a\in \{c,d\}$ and $b\in Y$. Then $\mu_{a,b}(\mathbf{L}_{cd\,\mbox{-}Y})=\mu_{a,b}(\mathbf{L}_{dc\,\mbox{-}Y})=\psi_Y/2$, so again $\mu_{a,b}(k\mathbf{L}_{cd\,\mbox{-}Y} - k\mathbf{L}_{dc\,\mbox{-}Y})= 0$.

It follows from (i)--(iii) that for each pair $a,b$ with $a\to b$ in $\wwt$ with weight $k\psi_Y$, we have $\mu_{a,b}(\mathfrak{L}_{Y,m}(\wwt))=k\psi_Y$; moreover, for each pair $a,b$ with neither $a\to b$ nor $b\to a$ in $\wwt$, we have $\mu_{a,b}(\mathfrak{L}_{Y,m}(\wwt))=0$. Thus, $\wwt =\mathcal{M}(\mathfrak{P}_{Y,m}(\wwt))$.\end{proof}

There are two key differences between our construction and those of McGarvey and Debord, namely that representation commutes with restriction to a set of candidates and transposition of candidates, as proved in the following subsections.

\begin{figure}[h]
\begin{center}
\begin{minipage}{2in}\begin{center} {\footnotesize\setlength{\tabcolsep}{.5pt}\begin{tabular}{cccccccccccccccccccccccc}
$1$ & $1$ & $1$ & $1$ & $1$ & $1$ 					& $1$ & $1$ & $1$ & $1$ & $1$ & $1$ 					& $1$ & $1$ & $1$ & $1$ & $1$ & $1$ 	& $1$ & $1$ & $1$ & $1$ & $1$ & $1$    \\\hline
$a$ & $a$ &  $\color{blue}{c}$ & $\color{blue}{c}$ & $\color{purple}{d}$ & $\color{purple}{d}$ 	& $\color{red}{a}$ & $\color{red}{a}$ & $b$ & $b$ & $b$ & $b$ 	& $c$ & $c$ &  $\color{medgreen}{b}$ & $\color{medgreen}{b}$ & $\color{orange}{d}$ & $\color{orange}{d}$ 	& $d$ & $d$ &  $\color{darkyellow}{b}$ & $\color{darkyellow}{b}$ & $d$ & $d$ \\
$b$ &  $b$ & $\color{blue}{a}$ & $\color{blue}{a}$ & $\color{purple}{a}$ & $\color{purple}{a}$	&$\color{red}{b}$ & $\color{red}{b}$ & $c$ & $c$ & $d$ & $d$	&$a$ & $a$ & $\color{medgreen}{c}$ & $\color{medgreen}{c}$ & $\color{orange}{c}$ & $\color{orange}{c}$	& $a$ & $a$ &  $\color{darkyellow}{d}$ & $\color{darkyellow}{d}$ & $c$ & $c$ \\
$c$ &  $d$ &  $b$ & $d$ & $b$ & $c$		 			& $c$ & $d$ & $a$ & $d$ & $a$ & $c$					&$b$ & $d$ & $a$ & $d$ & $a$ & $b$	& $b$ &  $c$ & $a$ & $c$ & $a$ & $b$ \\
$d$ &  $c$ &  $d$ & $b$ & $c$ & $b$					& $d$ & $c$ & $d$ & $a$ & $c$ & $a$					&$d$ & $b$ & $d$ & $a$ & $b$ & $a$	& $c$ &  $b$ &  $c$ & $a$ & $b$ & $a$  \\
\end{tabular}} \end{center} \end{minipage}\qquad\begin{minipage}{.35in} $\overset{(\cdot)|_{ \{a,b,c\}}}{\Longrightarrow}$ \end{minipage}\qquad\begin{minipage}{2in} \begin{center} 

{\footnotesize\setlength{\tabcolsep}{.5pt}\begin{tabular}{cccccccccccccccccccccccc}
$1$ & $1$ & $1$ & $1$ & $1$ & $1$ 					& $1$ & $1$ & $1$ & $1$ & $1$ & $1$ 					& $1$ & $1$ & $1$ & $1$ & $1$ & $1$ 								& $1$ & $1$ & $1$ & $1$ & $1$ & $1$    \\\hline
$a$ & $a$ &  $\color{blue}{c}$ & $\color{blue}{c}$ & $a$ & $a$ 	& $\color{red}{a}$ & $\color{red}{a}$ & $b$ & $b$ & $b$ & $b$ 	& $c$ & $c$ &  $\color{medgreen}{b}$ & $\color{medgreen}{b}$ & $c$ & $c$ 	& $a$ & $a$ &  $b$ & $b$ & $c$ & $c$ \\
$b$ &  $b$ & $\color{blue}{a}$ & $\color{blue}{a}$ & $b$ & $c$	&$\color{red}{b}$ & $\color{red}{b}$ & $c$ & $c$ & $a$ & $c$	&$a$ & $a$ & $\color{medgreen}{c}$ & $\color{medgreen}{c}$ & $a$ & $b$	& $b$ & $c$ &  $a$ & $c$ & $a$ & $b$ \\
$c$ &  $c$ &  $b$ & $b$ & $c$ & $b$		 			& $c$ & $c$ & $a$ & $a$ & $c$ & $a$					&$b$ & $b$ & $a$ & $a$ & $b$ & $a$								& $c$ &  $b$ & $c$ & $a$ & $b$ & $a$ \\
\end{tabular}}  \end{center}\end{minipage}

\end{center}

\vspace{.1in}
\begin{minipage}{2in}\begin{center}$\quad\Uparrow_{\mathfrak{P}_{\{a,b,c,d\},1}(\cdot)}$\end{center}\end{minipage}\qquad\qquad\qquad\; \begin{minipage}{2in}\begin{center}  $\Uparrow_{\mathfrak{P}_{\{a,b,c,d\},1}(\cdot)}$\end{center}\end{minipage}
\vspace{.1in}
\begin{center}
\begin{minipage}{2in}\begin{center}
\begin{tikzpicture}

\node[circle,draw, minimum width=0.25in] at (0,0) (a) {$a$}; 
\node[circle,draw,minimum width=0.25in] at (3,0) (c) {$c$}; 
\node[circle,draw,minimum width=0.25in] at (1.5,1.5) (b) {$b$}; 
\node[circle,draw,minimum width=0.25in] at (1.5,-1.5) (d) {$d$}; 

\path[->,draw,thick,medgreen] (b) to node [fill=white]  {$4$} (c);
\path[->,draw,thick,blue] (c) to node [fill=white,pos=.7]  {$4$} (a);
\path[->,draw,thick,red] (a) to node [fill=white]  {$4$} (b);
\path[->,draw,thick,darkyellow] (b) to node [fill=white,pos=.7]  {$4$} (d);
\path[->,draw,thick,purple] (d) to node [fill=white]  {$4$} (a);
\path[->,draw,thick,orange] (d) to node [fill=white]  {$4$} (c);

\end{tikzpicture}
\end{center}\end{minipage} \qquad \begin{minipage}{.35in} $\overset{(\cdot)|_{ \{a,b,c\}}}{\Longrightarrow}$\end{minipage}\qquad  \begin{minipage}{2in}\begin{center}\begin{tikzpicture}

\node[circle,draw, minimum width=0.25in] at (0,0) (a) {$a$}; 
\node[circle,draw,minimum width=0.25in] at (3,0) (c) {$c$}; 
\node[circle,draw,minimum width=0.25in] at (1.5,1.5) (b) {$b$}; 

\path[->,draw,thick,medgreen] (b) to node [fill=white]  {$4$}  (c);
\path[->,draw,thick,blue] (c) to node [fill=white]  {$4$} (a);
\path[->,draw,thick,red] (a) to node [fill=white]  {$4$}  (b);

\end{tikzpicture}\end{center}\end{minipage}
\end{center}

\vspace{.1in}
\begin{minipage}{2in}\begin{center}$\quad\;\Downarrow_{\mathsf{Deb}(\cdot)}$\end{center}\end{minipage}\qquad\qquad\qquad\; \begin{minipage}{2in}\begin{center}  $\Downarrow_{\mathsf{Deb}(\cdot)}$ \end{center}\end{minipage}
\vspace{.1in}

\begin{center}
\begin{minipage}{2in}\begin{center}\vspace{.75in} {\footnotesize\setlength{\tabcolsep}{2pt}\begin{tabular}{cccccccccccc}
$2$ & $2$ & $2$ & $2$ & $2$ & $2$ 					& $2$ & $2$ & $2$ & $2$ & $2$ & $2$ \\\hline 
 $\color{red}{a}$ & $d$ & $\color{medgreen}{b}$ &$d$ & $\color{blue}{c}$ & $d$ & $\color{purple}{d}$ & $c$  & \color{orange}{$d$} & $b$& \color{darkyellow}{$b$}& $c$\\
 $\color{red}{b}$& $c$ & $\color{medgreen}{c}$ &$a$& $\color{blue}{a}$ & $b$ & $\color{purple}{a}$& $b$ & \color{orange}{$c$} & $a$& \color{darkyellow}{$d$}& $a$\\
 $c$& $\color{red}{a}$ & $a$ & $\color{medgreen}{b}$& $b$ & $\color{blue}{c}$& $b$& $\color{purple}{d}$ & $a$ & \color{orange}{$d$} & $a$& \color{darkyellow}{$b$}\\
 $d$& $\color{red}{b}$ & $d$ & $\color{medgreen}{c}$& $d$ & $\color{blue}{a}$& $c$& $\color{purple}{a}$ & $b$ & \color{orange}{$c$}& $c$& \color{darkyellow}{$d$}
\end{tabular}}\end{center} \end{minipage}\qquad\begin{minipage}{.35in}\vspace{.75in} $\overset{(\cdot)|_{ \{a,b,c\}}}{\Longrightarrow}$ \end{minipage}\qquad\begin{minipage}{2in} \begin{center} 

{\footnotesize\setlength{\tabcolsep}{2pt}\begin{tabular}{cccccc}
$2$ & $2$ & $2$ & $2$ & $2$ & $2$   \\\hline
$\color{red}{a}$ & $c$ &  $\color{medgreen}{b}$ & $a$ & $\color{blue}{c}$ & $b$ \\
$\color{red}{b}$ &  $\color{red}{a}$ & $\color{medgreen}{c}$ & $\color{medgreen}{b}$ & $\color{blue}{a}$ & $\color{blue}{c}$ \\
$c$ &  $\color{red}{b}$ &  $a$ & $\color{medgreen}{c}$ & $b$ & $\color{blue}{a}$ \\
\end{tabular}}
\vspace{.05in}\\$\rotatebox{90}{$\neq$}$\vspace{.05in}\\

{\footnotesize\setlength{\tabcolsep}{2pt}\begin{tabular}{cccccccccccc}
$2$ & $2$ & $2$ & $2$ & $2$ & $2$ 					& $2$ & $2$ & $2$ & $2$ & $2$ & $2$ \\\hline 
 $\color{red}{a}$ & $c$ & $\color{medgreen}{b}$ &$a$ & $\color{blue}{c}$ & $b$ & {\cellcolor{gray!25}$\color{purple}{a}$} & {\cellcolor{gray!25}$c$}  & {\cellcolor{gray!25}\color{orange}{$c$}} & {\cellcolor{gray!25}$b$}& {\cellcolor{gray!25}\color{darkyellow}{$b$}}& {\cellcolor{gray!25}$c$}\\
 $\color{red}{b}$& $\color{red}{a}$ & $\color{medgreen}{c}$ &$\color{medgreen}{b}$& $\color{blue}{a}$ & $\color{blue}{c}$ & {\cellcolor{gray!25}$b$} & {\cellcolor{gray!25}$b$} & {\cellcolor{gray!25}\color{orange}{$a$}} & {\cellcolor{gray!25}$a$}& {\cellcolor{gray!25}\color{darkyellow}{$a$}}& {\cellcolor{gray!25}$a$}\\
 $c$& $\color{red}{b}$ & $a$ & $\color{medgreen}{c}$& $b$ & $\color{blue}{a}$& {\cellcolor{gray!25}$c$} & {\cellcolor{gray!25}$\color{purple}{a}$} & {\cellcolor{gray!25}$b$} & {\cellcolor{gray!25}\color{orange}{$c$}} & {\cellcolor{gray!25}$c$} & {\cellcolor{gray!25}\color{darkyellow}{$b$}}
\end{tabular}} \end{center}\end{minipage}

\end{center}
\caption{Representation commutes with restriction (above), but the McGarvey/Debord construction does not commute with restriction (below).}\label{RestrictionFig}
\end{figure}
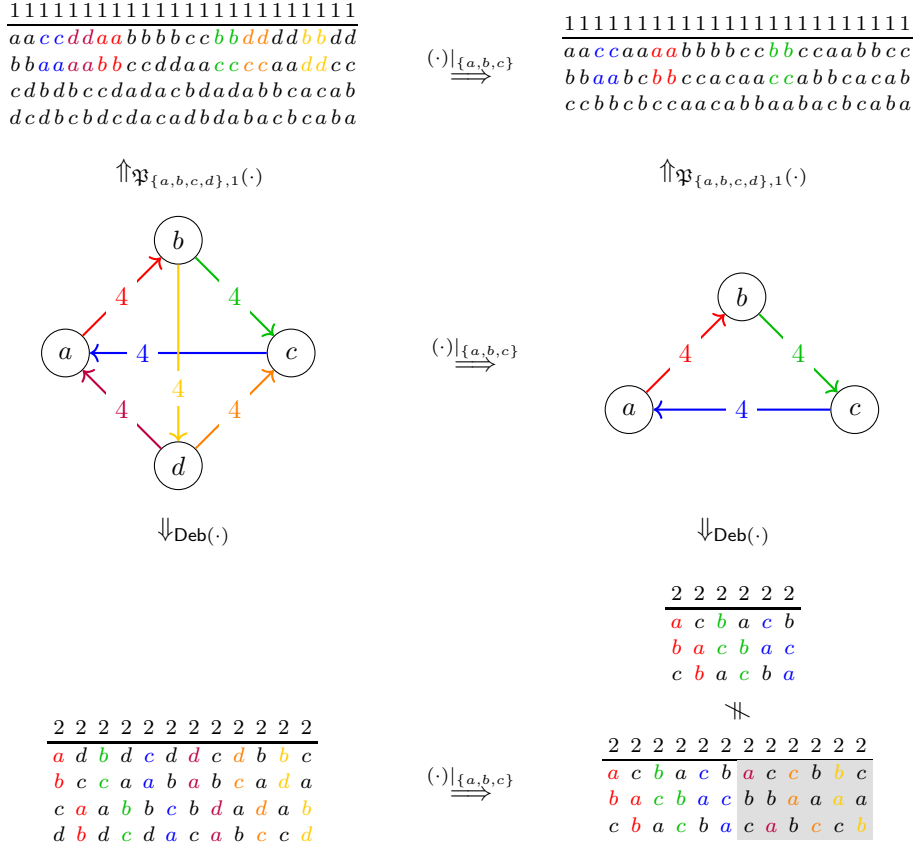

\subsection{Representation commutes with restriction} Given a binary relation $R$ on a set $X$ and $Y\subseteq X$, \textit{the restriction of $R$ to $Y$} is the binary relation $R|_{ Y}$ on $Y$ such that for all $a,b\in Y$, $(a,b)\in R|_{ Y}$ if and only if $(a,b)\in R$. 

Given a generalized anonymous profile $\mathbf{P}$ and $Y\subseteq X(\mathbf{P})$, the \textit{restriction of $\mathbf{P}$ to $Y$} is the profile $\mathbf{P}|_{ Y}$ with $X(\mathbf{P}|_{ Y})=Y$ such that for all $L\in\mathcal{L}(Y)$, \[\mathbf{P}|_{ Y}(L)=\sum_{L'\in\mathcal{L}(X(\mathbf{P})):\, L=L'|_{ Y}}\mathbf{P}(L').\] 

\noindent Given a weak tournament $\wt$ and $Z\subseteq X(\wt)$,  the \textit{restriction of $\wt$ to $Z$} is the directed graph  $\wt|_{ Z}$ whose set of vertices is $Y$ and whose edge relation is $\to\mid_{ Z}$.  The restriction of a weighted weak tournament is defined analogously, with the weight function restricted to $\to\mid_{ Z}$. Clearly restriction commutes with taking the majority or margin graph of a profile and with multiplication.

\begin{lemma}\label{RestrictCommute}Let $\mathbf{P}$ be a generalized anonymous profile and $\wwt$ a weighted weak tournament.
\begin{enumerate}
\item For any $\varnothing\neq Z\subseteq X(\mathbf{P})$,  
$M(\mathbf{P}|_{ Z})= M(\mathbf{P})|_{ Z}$ and $\mathcal{M}(\mathbf{P}|_{ Z})= \mathcal{M}(\mathbf{P})|_{ Z}$.
\item For any $\varnothing\neq Z\subseteq X(\wwt)$ and $n\in\mathbb{Z}^+$, $n(\wwt|_{ Z})=(n\wwt)|_{ Z}$.
\end{enumerate}\end{lemma}

Next we show that representation commutes with restriction, as in Figure \ref{RestrictionFig}.

\begin{proposition}\label{SubProp} For any finite $Y\subseteq\mathcal{X}$, $\wt\in\mathbb{W}\mathbb{T}_Y$, $m\in\mathbb{Z}^+$,  $\wwt\in \mathbb{W}\mathbb{T}^w_{Y,m}$: 
\begin{enumerate}
\item\label{SubProp1} if  ${\varnothing\neq Z\subseteq X(\wt)}$, then  $\mathfrak{P}_Y(\wt|_{ Z})=\mathfrak{P}_Y(\wt)|_{ Z}$;
\item\label{SubProp2} if $\varnothing\neq Z\subseteq X(\wwt)$, then $\mathfrak{P}_{Y,m}(\wwt|_{ Z})=\mathfrak{P}_{Y,m}(\wwt)|_{ Z}$.
\end{enumerate}
\end{proposition}

\begin{proof} We give only the proof for (\ref{SubProp1}), as the proof for (\ref{SubProp2}) is essentially the same but with more notation. The key observation is that if $\{a,b\}\not\subseteq Z$, then $(\mathbf{L}_{ab\,\mbox{-}Y})|_{ Z} = (\mathbf{L}_{ba\,\mbox{-}Y})|_{ Z}$. Then we have
{\renewcommand{\arraystretch}{1.7}
\begin{center}
\begin{longtable}{rcll} 
$\mathfrak{P}_Y(\wt)|_{ Z} $ &=& $(\mathfrak{L}_Y(\wt)|_{ X(\wt)})|_{ Z}$ by (\ref{PfromL})\\
&=& $\mathfrak{L}_{Y}(\wt)|_{ Z}$ since $Z\subseteq X(T)$\\
 &=& $\big(\mathbf{L}_Y + \underset{\mbox{ {\footnotesize $a\to b$ in $T$}}}{\sum} (\mathbf{L}_{ab\,\mbox{-}Y} - \mathbf{L}_{ba\,\mbox{-}Y})\big)|_{ Z}$ \\
&=& $\big(\mathbf{L}_Y + \underset{\{a,b\}\subseteq Z}{\underset{\mbox{ {\footnotesize $a\to b$ in $T$}}}{\sum}} (\mathbf{L}_{ab\,\mbox{-}Y} - \mathbf{L}_{ba\,\mbox{-}Y})+$ \\
&& \qquad\quad$\underset{\{a,b\}\not\subseteq Z}{\underset{\mbox{ {\footnotesize $a\to b$ in $T$}}}{\sum}} (\mathbf{L}_{ab\,\mbox{-}Y} - \mathbf{L}_{ba\,\mbox{-}Y}) \big)|_{ Z}$\\
&=& $\big(\mathbf{L}_Y + \underset{\{a,b\}\subseteq Z}{\underset{\mbox{ {\footnotesize $a\to b$ in $T$}}}{\sum}} (\mathbf{L}_{ab\,\mbox{-}Y} - \mathbf{L}_{ba\,\mbox{-}Y})\big)|_{ Z}+$\\
&& \qquad\quad$\underset{\{a,b\}\not\subseteq Z}{\underset{\mbox{ {\footnotesize $a\to b$ in $T$}}}{\sum}} ((\mathbf{L}_{ab\,\mbox{-}Y})|_{ Z} - (\mathbf{L}_{ba\,\mbox{-}Y})|_{ Z})$\\
&=& $\big(\mathbf{L}_Y + \underset{\mbox{ {\footnotesize $a\to b$ in $T|_{ Z}$}}}{\sum} (\mathbf{L}_{ab\,\mbox{-}Y} - \mathbf{L}_{ba\,\mbox{-}Y})\big)|_{ Z}$\\
&=& $\mathfrak{L}_{Y}(\wt|_{ Z})|_{ Z}$ \\
&=& $\mathfrak{L}_Y(\wt|_{ Z})|_{ X(\wt|_{ Z})}$ since $X(\wt|_{ Z})=Z$ \\
&=& $\mathfrak{P}_Y(\wt|_{ Z})$ by (\ref{PfromL}). \qedhere
\end{longtable}
\end{center}}
\end{proof}

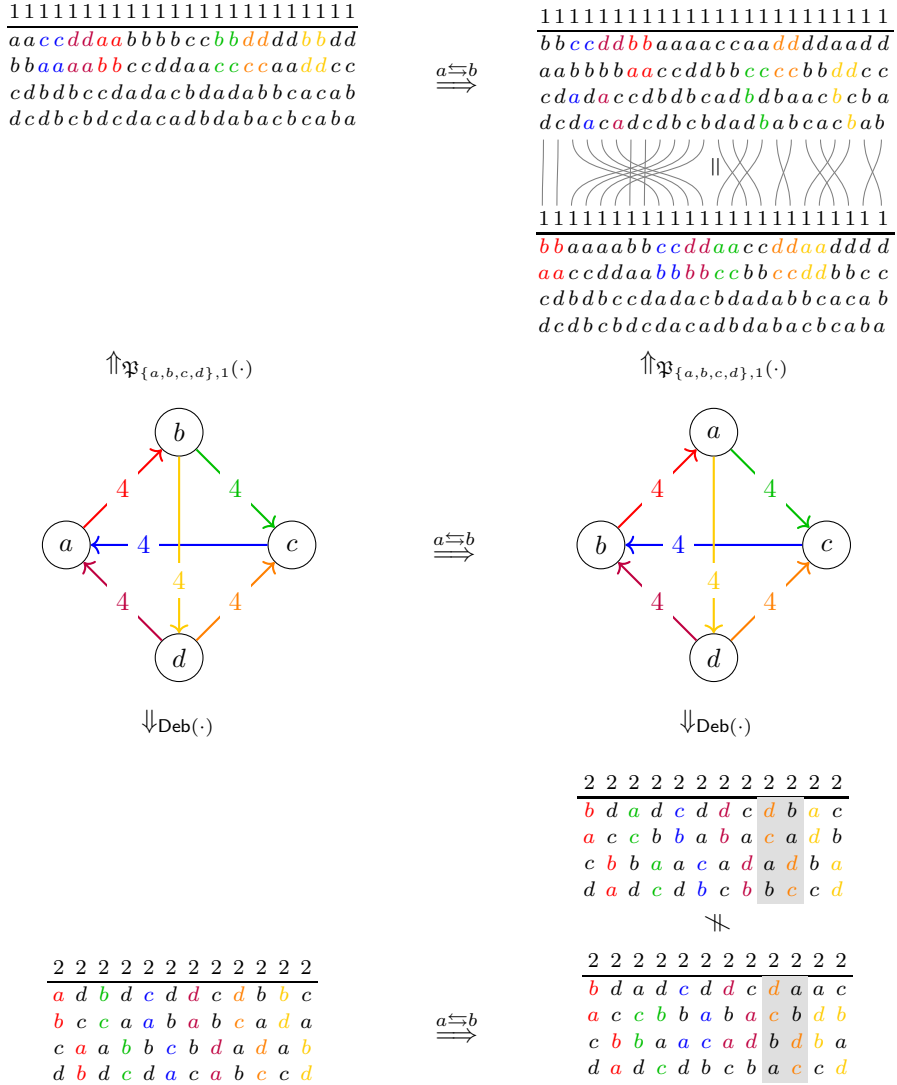
\begin{figure}[ht]

\begin{minipage}{2in}\begin{center}\vspace{-1.05in} {\footnotesize\setlength{\tabcolsep}{.5pt}\begin{tabular}{cccccccccccccccccccccccc}
$1$ & $1$ & $1$ & $1$ & $1$ & $1$ 					& $1$ & $1$ & $1$ & $1$ & $1$ & $1$ 					& $1$ & $1$ & $1$ & $1$ & $1$ & $1$ 	& $1$ & $1$ & $1$ & $1$ & $1$ & $1$    \\\hline
$a$ & $a$ &  $\color{blue}{c}$ & $\color{blue}{c}$ & $\color{purple}{d}$ & $\color{purple}{d}$ 	& $\color{red}{a}$ & $\color{red}{a}$ & $b$ & $b$ & $b$ & $b$ 	& $c$ & $c$ &  $\color{medgreen}{b}$ & $\color{medgreen}{b}$ & $\color{orange}{d}$ & $\color{orange}{d}$ 	& $d$ & $d$ &  $\color{darkyellow}{b}$ & $\color{darkyellow}{b}$ & $d$ & $d$ \\
$b$ &  $b$ & $\color{blue}{a}$ & $\color{blue}{a}$ & $\color{purple}{a}$ & $\color{purple}{a}$	&$\color{red}{b}$ & $\color{red}{b}$ & $c$ & $c$ & $d$ & $d$	&$a$ & $a$ & $\color{medgreen}{c}$ & $\color{medgreen}{c}$ & $\color{orange}{c}$ & $\color{orange}{c}$	& $a$ & $a$ &  $\color{darkyellow}{d}$ & $\color{darkyellow}{d}$ & $c$ & $c$ \\
$c$ &  $d$ &  $b$ & $d$ & $b$ & $c$		 			& $c$ & $d$ & $a$ & $d$ & $a$ & $c$					&$b$ & $d$ & $a$ & $d$ & $a$ & $b$	& $b$ &  $c$ & $a$ & $c$ & $a$ & $b$ \\
$d$ &  $c$ &  $d$ & $b$ & $c$ & $b$					& $d$ & $c$ & $d$ & $a$ & $c$ & $a$					&$d$ & $b$ & $d$ & $a$ & $b$ & $a$	& $c$ &  $b$ &  $c$ & $a$ & $b$ & $a$  \\
\end{tabular}}\end{center} \end{minipage}  \qquad \begin{minipage}{.2in}\vspace{-.95in}$\overset{a\leftrightarrows b}{\Longrightarrow}$ \end{minipage}\qquad \begin{minipage}{2in} \begin{center} {\footnotesize\setlength{\tabcolsep}{.5pt} \begin{tabular}{cccccccccccccccccccccccc}

$1$ & $1$ & $1$ & $1$ & $1$ & $1$ 											& $1$ & $1$ & $1$ & $1$ & $1$ & $1$ 					& $1$ & $1$ & $1$ & $1$ & $1$ & $1$ 														& $1$ & $1$ & $1$ & $1$ & $1$ & $1$    \\\hline
$b$ & $b$ &  $\color{blue}{c}$ & $\color{blue}{c}$ & $\color{purple}{d}$ & $\color{purple}{d}$ 	& $\color{red}{b}$ & $\color{red}{b}$ & $a$ & $a$ & $a$ & $a$ 	& $c$ & $c$ &  $a$ & $a$ & $\color{orange}{d}$ & $\color{orange}{d}$ 								& $d$ & $d$ &  $a$ & $a$ & $d$ & $d$ \\
$a$ &  $a$ & $b$ & $b$ & $b$ & $b$											&$\color{red}{a}$ & $\color{red}{a}$ & $c$ & $c$ & $d$ & $d$	&$b$ & $b$ & $\color{medgreen}{c}$ & $\color{medgreen}{c}$ & $\color{orange}{c}$ & $\color{orange}{c}$	& $b$ & $b$ &  $\color{darkyellow}{d}$ & $\color{darkyellow}{d}$ & $c$ & $c$ \\
$c$ &  $d$ &  $\color{blue}{a}$ & $d$ & $\color{purple}{a}$ & $c$		 				& $c$ & $d$ & $b$ & $d$ & $b$ & $c$					&$a$ & $d$ & $\color{medgreen}{b}$ & $d$ & $b$ & $a$											& $a$ &  $c$ & $\color{darkyellow}{b}$ & $c$ & $b$ & $a$ \\
$d$ &  $c$ &  $d$ & $\color{blue}{a}$ & $c$ & $\color{purple}{a}$						& $d$ & $c$ & $d$ & $b$ & $c$ & $b$					&$d$ & $a$ & $d$ & $\color{medgreen}{b}$ & $a$ & $b$											& $c$ &  $a$ &  $c$ & $\color{darkyellow}{b}$ & $a$ & $b$ \vspace{-.1in}   \\

 \tikzmark{u1}$\,$ &  \tikzmark{u2}$\,$ &  \tikzmark{u3}$\,$ &\tikzmark{u4}$\,$&\tikzmark{u5}$\,$&\tikzmark{u6}$\,$&\tikzmark{u7}$\,$&\tikzmark{u8}$\,$&\tikzmark{u9}$\,$&\tikzmark{u10}$\,$&\tikzmark{u11}$\,$&\tikzmark{u12}$\,$&\tikzmark{u13}$\,$&\tikzmark{u14}$\,$&\tikzmark{u15}$\,$&\tikzmark{u16}$\,$&\tikzmark{u17}$\,$&\tikzmark{u18}$\,$&\tikzmark{u19}$\,$&\tikzmark{u20}$\,$&\tikzmark{u21}$\,$&\tikzmark{u22}$\,$&\tikzmark{u23}$\,$&\tikzmark{u24}$\,$
\end{tabular}\,}\vspace{.05in}\\$\rotatebox{90}{=}$\vspace{.05in} \\ {\footnotesize\setlength{\tabcolsep}{.5pt} \begin{tabular}{cccccccccccccccccccccccc}
  \tikzmark{d1}$\,$ &  \tikzmark{d2}$\,$ &  \tikzmark{d3}$\,$& \tikzmark{d4}$\,$& \tikzmark{d5}$\,$& \tikzmark{d6}$\,$& \tikzmark{d7}$\,$& \tikzmark{d8}$\,$& \tikzmark{d9}$\,$&\tikzmark{d10}$\,$&\tikzmark{d11}$\,$&\tikzmark{d12}$\,$&\tikzmark{d13}$\,$&\tikzmark{d14}$\,$&\tikzmark{d15}$\,$&\tikzmark{d16}$\,$&\tikzmark{d17}$\,$&\tikzmark{d18}$\,$&\tikzmark{d19}$\,$&\tikzmark{d20}$\,$&\tikzmark{d21}$\,$&\tikzmark{d22}$\,$&\tikzmark{d23}$\,$&\tikzmark{d24}$\,$\vspace{-.05in}\\

$1$ & $1$ & $1$ & $1$ & $1$ & $1$ 						& $1$ & $1$ & $1$ & $1$ & $1$ & $1$ 			& $1$ & $1$ & $1$ & $1$ & $1$ & $1$ 				& $1$ & $1$ & $1$ & $1$ & $1$ & $1$    \\\hline
$\color{red}{b}$ & $\color{red}{b}$ &  $a$ & $a$ & $a$ & $a$ 		& $b$ & $b$ &  $\color{blue}{c}$ & $\color{blue}{c}$ & $\color{purple}{d}$ & $\color{purple}{d}$ 			& $\color{medgreen}{a}$ & $\color{medgreen}{a}$ &  $c$ & $c$ & $\color{orange}{d}$ & $\color{orange}{d}$ 				& $\color{darkyellow}{a}$ & $\color{darkyellow}{a}$ &  $d$ & $d$ & $d$ & $d$ \\
$\color{red}{a}$ &  $\color{red}{a}$ & $c$ & $c$ & $d$ & $d$		&$a$ & $a$ & $\color{blue}{b}$ & $\color{blue}{b}$ & $\color{purple}{b}$ & $\color{purple}{b}$			&$\color{medgreen}{c}$ & $\color{medgreen}{c}$ & $b$ & $b$ & $\color{orange}{c}$ & $\color{orange}{c}$				& $\color{darkyellow}{d}$ & $\color{darkyellow}{d}$ &  $b$ & $b$ & $c$ & $c$ \\
$c$ &  $d$ &  $b$ & $d$ & $b$ & $c$						 & $c$ & $d$ & $a$ & $d$ & $a$ & $c$			&$b$ & $d$ & $a$ & $d$ & $a$ & $b$				& $b$ &  $c$ & $a$ & $c$ & $a$ & $b$ \\
$d$ &  $c$ &  $d$ & $b$ & $c$ & $b$						& $d$ & $c$ & $d$ & $a$ & $c$ & $a$			&$d$ & $b$ & $d$ & $a$ & $b$ & $a$				& $c$ &  $b$ &  $c$ & $a$ & $b$ & $a$ \vspace{-.1in}  \\
 
\end{tabular}$\,$}  \end{center}\begin{tikzpicture}
    [
      remember picture,
      overlay,
      -latex,
      color=gray,
      yshift=1ex,
      shorten >=1pt,
      shorten <=1pt,
    ]
    \draw[-] ({pic cs:u1}) to [out=270, in=90] ({pic cs:d1});
    \draw[-]  ({pic cs:u2}) to [out=270, in=90] ({pic cs:d2});
     \draw[-]  ({pic cs:u3}) to [out=270, in=90] ({pic cs:d9});
          \draw[-]  ({pic cs:u4}) to [out=270, in=90] ({pic cs:d10});
               \draw[-]  ({pic cs:u5}) to [out=270, in=90] ({pic cs:d11});
          \draw[-]  ({pic cs:u6}) to [out=270, in=90] ({pic cs:d12});
                         \draw[-]  ({pic cs:u7}) to [out=270, in=90] ({pic cs:d7});
          \draw[-]  ({pic cs:u8}) to [out=270, in=90] ({pic cs:d8});
          
                         \draw[-]  ({pic cs:u9}) to [out=270, in=90] ({pic cs:d3});
          \draw[-]  ({pic cs:u10}) to [out=270, in=90] ({pic cs:d4});
                         \draw[-]  ({pic cs:u11}) to [out=270, in=90] ({pic cs:d5});
          \draw[-]  ({pic cs:u12}) to [out=270, in=90] ({pic cs:d6});
          
                                   \draw[-]  ({pic cs:u13}) to [out=270, in=90] ({pic cs:d15});
          \draw[-]  ({pic cs:u14}) to [out=270, in=90] ({pic cs:d16});
          
                                             \draw[-]  ({pic cs:u15}) to [out=270, in=90] ({pic cs:d13});
          \draw[-]  ({pic cs:u16}) to [out=270, in=90] ({pic cs:d14});
          
                                                       \draw[-]  ({pic cs:u17}) to [out=270, in=90] ({pic cs:d18});
          \draw[-]  ({pic cs:u18}) to [out=270, in=90] ({pic cs:d17});
          
                                                                 \draw[-]  ({pic cs:u19}) to [out=270, in=90] ({pic cs:d21});
          \draw[-]  ({pic cs:u20}) to [out=270, in=90] ({pic cs:d22});
          
                                                                           \draw[-]  ({pic cs:u21}) to [out=270, in=90] ({pic cs:d19});
          \draw[-]  ({pic cs:u22}) to [out=270, in=90] ({pic cs:d20});
          
                                                                                     \draw[-]  ({pic cs:u23}) to [out=270, in=90] ({pic cs:d24});
          \draw[-]  ({pic cs:u24}) to [out=270, in=90] ({pic cs:d23});
          
  \end{tikzpicture} \end{minipage}

\vspace{.1in}
\begin{minipage}{2.2in}\begin{center}$\Uparrow_{\mathfrak{P}_{\{a,b,c,d\},1}(\cdot)}\;\;$ \end{center}\end{minipage}\qquad\quad\; \begin{minipage}{2.2in}\begin{center}  $\;\;\;$$\Uparrow_{\mathfrak{P}_{\{a,b,c,d\},1}(\cdot)}$ \end{center}\end{minipage}
\vspace{.1in}

\begin{center}
\begin{minipage}{2in}
\begin{tikzpicture}

\node[circle,draw, minimum width=0.25in] at (0,0) (a) {$a$}; 
\node[circle,draw,minimum width=0.25in] at (3,0) (c) {$c$}; 
\node[circle,draw,minimum width=0.25in] at (1.5,1.5) (b) {$b$}; 
\node[circle,draw,minimum width=0.25in] at (1.5,-1.5) (d) {$d$}; 

\path[->,draw,thick,medgreen] (b) to node [fill=white]  {$4$} (c);
\path[->,draw,thick,blue] (c) to node [fill=white,pos=.7]   {$4$} (a);
\path[->,draw,thick,red] (a) to node [fill=white]  {$4$} (b);
\path[->,draw,thick,darkyellow] (b) to node [fill=white,pos=.7]  {$4$} (d);
\path[->,draw,thick,purple] (d) to node [fill=white]  {$4$} (a);
\path[->,draw,thick,orange] (d) to node [fill=white]  {$4$} (c);

\end{tikzpicture}\end{minipage} \begin{minipage}{.475in}$\overset{a\leftrightarrows b}{\Longrightarrow}$\end{minipage}\qquad \begin{minipage}{1.475in}\begin{tikzpicture}

\node[circle,draw, minimum width=0.25in] at (0,0) (a) {$b$}; 
\node[circle,draw,minimum width=0.25in] at (3,0) (c) {$c$}; 
\node[circle,draw,minimum width=0.25in] at (1.5,1.5) (b) {$a$}; 
\node[circle,draw,minimum width=0.25in] at (1.5,-1.5) (d) {$d$}; 

\path[->,draw,thick,medgreen] (b) to node [fill=white]  {$4$}  (c);
\path[->,draw,thick,blue] (c) to node [fill=white,pos=.7]  {$4$}  (a);
\path[->,draw,thick,red] (a) to node [fill=white]  {$4$}  (b);
\path[->,draw,thick,darkyellow] (b) to node [fill=white,pos=.7]  {$4$}  (d);
\path[->,draw,thick,purple] (d) to node [fill=white]  {$4$}  (a);
\path[->,draw,thick,orange] (d) to node [fill=white]  {$4$} (c);

\end{tikzpicture}\end{minipage}
\end{center}

\vspace{.1in}
\begin{minipage}{2in}\begin{center} $\;\;\;\Downarrow_{\mathsf{Deb}(\cdot)}$ \end{center}\end{minipage}\qquad\qquad\qquad\; \begin{minipage}{2in}\begin{center}  $\Downarrow_{\mathsf{Deb}(\cdot)}\;\;$ \end{center}\end{minipage}
\vspace{.1in}

\begin{center}
\begin{minipage}{2in}\begin{center}\vspace{1in}{\footnotesize {\setlength{\tabcolsep}{2pt}\begin{tabular}{cccccccccccc}
$2$ & $2$ & $2$ & $2$ & $2$ & $2$ 					& $2$ & $2$ & $2$ & $2$ & $2$ & $2$ \\\hline 
 $\color{red}{a}$ & $d$ & $\color{medgreen}{b}$ &$d$ & $\color{blue}{c}$ & $d$ & $\color{purple}{d}$ & $c$  & \color{orange}{$d$} & $b$& \color{darkyellow}{$b$}& $c$\\
 $\color{red}{b}$& $c$ & $\color{medgreen}{c}$ &$a$& $\color{blue}{a}$ & $b$ & $\color{purple}{a}$& $b$ & \color{orange}{$c$} & $a$& \color{darkyellow}{$d$}& $a$\\
 $c$& $\color{red}{a}$ & $a$ & $\color{medgreen}{b}$& $b$ & $\color{blue}{c}$& $b$& $\color{purple}{d}$ & $a$ & \color{orange}{$d$} & $a$& \color{darkyellow}{$b$}\\
 $d$& $\color{red}{b}$ & $d$ & $\color{medgreen}{c}$& $d$ & $\color{blue}{a}$& $c$& $\color{purple}{a}$ & $b$ & \color{orange}{$c$}& $c$& \color{darkyellow}{$d$}
\end{tabular}}}
\end{center} \end{minipage}  \qquad \begin{minipage}{.2in}\vspace{1.05in}$\overset{a\leftrightarrows b}{\Longrightarrow}$ \end{minipage}\qquad \begin{minipage}{2in} \begin{center} {\footnotesize {\setlength{\tabcolsep}{2pt}\begin{tabular}{cccccccccccc}
$2$ & $2$ & $2$ & $2$ & $2$ & $2$ 					& $2$ & $2$ & $2$ & $2$ & $2$ & $2$ \\\hline 
$\color{red}{b}$& $d$& $\color{medgreen}{a}$ & $d$& $\color{blue}{c}$ & $d$& \color{purple}{$d$} & $c$&  {\cellcolor{gray!25}$\color{orange}{d}$} & {\cellcolor{gray!25}$b$} & $\color{darkyellow}{a}$ & $c$\\
$\color{red}{a}$& $c$& $\color{medgreen}{c}$& $b$& $\color{blue}{b}$& $a$ & \color{purple}{$b$}& $a$& {\cellcolor{gray!25}$\color{orange}{c}$}& {\cellcolor{gray!25}$a$} & $\color{darkyellow}{d}$& $b$\\
$c$& $\color{red}{b}$& $b$& $\color{medgreen}{a}$& $a$& $\color{blue}{c}$& $a$& \color{purple}{$d$}& {\cellcolor{gray!25}$a$}& {\cellcolor{gray!25}$\color{orange}{d}$} & $b$& $\color{darkyellow}{a}$ \\
$d$& $\color{red}{a}$& $d$& $\color{medgreen}{c}$& $d$& $\color{blue}{b}$& $c$& \color{purple}{$b$} & {\cellcolor{gray!25}$b$}& {\cellcolor{gray!25}$\color{orange}{c}$} & $c$& $\color{darkyellow}{d}$
\end{tabular}}} \vspace{.03in}\\$\rotatebox{90}{$\neq$}$\vspace{.07in}\\{\footnotesize {\setlength{\tabcolsep}{2pt}\begin{tabular}{cccccccccccc}
$2$ & $2$ & $2$ & $2$ & $2$ & $2$ 					& $2$ & $2$ & $2$ & $2$ & $2$ & $2$ \\\hline 
 $\color{red}{b}$ & $d$ & $a$ &$d$ & $\color{blue}{c}$ & $d$ & $\color{purple}{d}$ & $c$  & {\cellcolor{gray!25}\color{orange}{$d$}} & {\cellcolor{gray!25}$a$} & $a$& $c$\\
 $\color{red}{a}$& $c$ & $\color{medgreen}{c}$ &$\color{medgreen}{b}$& $b$ & $\color{blue}{a}$ & $b$& $\color{purple}{a}$ & {\cellcolor{gray!25}\color{orange}{$c$}} & {\cellcolor{gray!25}$b$}& \color{darkyellow}{$d$}& \color{darkyellow}{$b$}\\
 $c$& $\color{red}{b}$ & $\color{medgreen}{b}$ & $a$& $\color{blue}{a}$ & $\color{blue}{c}$& $\color{purple}{a}$& $\color{purple}{d}$ & {\cellcolor{gray!25}$b$} & {\cellcolor{gray!25}\color{orange}{$d$}} & \color{darkyellow}{$b$} & $a$\\
 $d$& $\color{red}{a}$ & $d$ & $\color{medgreen}{c}$& $d$ & $b$& $c$& $b$ & {\cellcolor{gray!25}$a$} & {\cellcolor{gray!25}\color{orange}{$c$}}& $c$& \color{darkyellow}{$d$}
\end{tabular}}}  \end{center} \end{minipage}
\end{center}
\caption{Representation commutes with transposition (above), but the McGarvey/Debord construction does not commute with transposition (below).}\label{TranspositionFig}
\end{figure} \newpage

\subsection{Representation commutes with transposition} Given $a,b\in\mathcal{X}$, a linear order $L$, and $X\subseteq\mathcal{X}$, recall the definitions of $L_{a\leftrightarrows b}$ and $X_{a\leftrightarrows b}$ from Section \ref{IndividualToSocial}. Given a generalized anonymous profile $\mathbf{P}$ and $a,b\in \mathcal{X}$, we define $\mathbf{P}_{a\leftrightarrows b}$ such that for any linear order $L$, $\mathbf{P}_{a\leftrightarrows b}(L)=\mathbf{P}(L_{a\leftrightarrows b})$.

Given a weak tournament $\wt=(X(\wt),\to)$ and $a,b\in \mathcal{X}$, we define the weak tournament $\wt_{a\leftrightarrows b}=(X(\wt)_{a\leftrightarrows b},\to_{_{a\leftrightarrows b}})$ such that for all $x,y\in X(\wt)$, we have 
\[\mbox{$x \to_{_{a\leftrightarrows b}}y$ in $\wt_{a\leftrightarrows b}$  if $\pi_{_{a\leftrightarrows b}}(x)\to\pi_{_{a\leftrightarrows b}}(y)$ in $\wt$.}\] 
Similarly, for a weighted weak tournament $\wwt={(T,w)}$, we define $\wwt_{a\leftrightarrows b}={(T_{a\leftrightarrows b},w_{_{a\leftrightarrows b}})}$ where if $x \to_{_{a\leftrightarrows b}}y$ in $\wt_{a\leftrightarrows b}$, then \[w_{_{a\leftrightarrows b}}(x,y)=w(\pi_{_{a\leftrightarrows b}}(x),\pi_{_{a\leftrightarrows b}}(y)).\]

Later we will use the obvious fact that transposition commutes with taking the majority or margin graph of a profile, with restriction, and with multiplication.

\begin{lemma}\label{TransposeCommute} Let $\mathbf{P}$ be a generalized anonymous profile, $\mathcal{T}$ a weighted weak tournament, and $a,b\in\mathcal{X}$.
\begin{enumerate}
\item\label{TransposeCommute1} $M(\mathbf{P}_{a\leftrightarrows b})=M(\mathbf{P})_{a\leftrightarrows b}$ and  $\mathcal{M}(\mathbf{P}_{a\leftrightarrows b})=\mathcal{M}(\mathbf{P})_{a\leftrightarrows b}$.
\item\label{TransposeCommute2} For any $\varnothing\neq Z\subseteq X(\mathbf{P})$, then $(\mathbf{P}_{\mid Z})_{a\leftrightarrows b}=(\mathbf{P}_{a\leftrightarrows b})_{\mid \pi_{a\leftrightarrows b}[Z]}$.
\item\label{TransposeCommute3} For any  $n\in\mathbb{Z}^+$, $n(\wwt_{a\leftrightarrows b})=(n\wwt)_{a\leftrightarrows b}$.
\end{enumerate}
\end{lemma} 

As illustrated in Figure \ref{TranspositionFig}, unlike the constructions of McGarvey and Debord, representation commutes with transposition.

\begin{proposition}\label{TransposeProp} For any finite $Y\subseteq\mathcal{X}$, $\wt\in \mathbb{W}\mathbb{T}_Y$,  $m\in\mathbb{Z}^+$,  $\wwt\in \mathbb{W}\mathbb{T}^w_{Y,m}$, and ${a,b\in Y}$:
\begin{enumerate}
\item\label{TransposeProp1} $\mathfrak{P}_Y(\wt_{a\leftrightarrows b})= \mathfrak{P}_Y(\wt)_{a\leftrightarrows b}$;
\item\label{TransposeProp2} $\mathfrak{P}_{Y,m}(\wwt_{a\leftrightarrows b})= \mathfrak{P}_{Y,m}(\wwt)_{a\leftrightarrows b}$.
\end{enumerate}
\end{proposition}
\begin{proof}  We give only the proof for (\ref{SubProp1}), as the proof for (\ref{SubProp2}) is essentially the same but with more notation. Define $\rho: Y\to Y$ by $\rho(a)=b$, $\rho(b)=a$, and $\rho(c)=c$ for all $c\in Y\setminus \{a,b\}$. The key observations are that $(\mathbf{L}_Y)_{a\leftrightarrows b} = \mathbf{L}_Y$, since $a,b\in Y$, and for any $c,d\in X(\wt)$, $(\mathbf{L}_{cd\,\mbox{-}Y})_{a\leftrightarrows b}=\mathbf{L}_{\rho(c)\rho(d)\,\mbox{-}Y}$. Then we have
{\renewcommand{\arraystretch}{1.7}
\begin{center}
\begin{longtable}{rcll} 
$\mathfrak{P}_Y(\wt)_{a\leftrightarrows b} $ &=& $(\mathfrak{L}_Y(\wt)|_{ X(\wt)})_{a\leftrightarrows b}$ by (\ref{PfromL})\\
&=& $(\mathfrak{L}_Y(\wt)_{a\leftrightarrows b})|_{ X(\wt)_{a\leftrightarrows b}}$ \\
&=& $\big(\big(\mathbf{L}_Y + \underset{\mbox{ {\footnotesize $c\to d$ in $T$}}}{\sum} (\mathbf{L}_{cd\,\mbox{-}Y} - \mathbf{L}_{dc\,\mbox{-}Y})\big)_{a\leftrightarrows b}\big)|_{ X(\wt)_{a\leftrightarrows b}}$\\
&=& $\big((\mathbf{L}_Y)_{a\leftrightarrows b} + \underset{\mbox{ {\footnotesize $c\to d$ in $T$}}}{\sum} ((\mathbf{L}_{cd\,\mbox{-}Y})_{a\leftrightarrows b} - (\mathbf{L}_{dc\,\mbox{-}Y})_{a\leftrightarrows b})\big)|_{ X(\wt)_{a\leftrightarrows b}}$ \\
&=& $\big(\mathbf{L}_Y + \underset{\mbox{ {\footnotesize $c\to d$ in $T$}}}{\sum} (\mathbf{L}_{\rho(c)\rho(d)\,\mbox{-}Y} - \mathbf{L}_{\rho(d)\rho(c)\,\mbox{-}Y})\big)|_{ X(\wt)_{a\leftrightarrows b}}$ \\
&=& $\big(\mathbf{L}_Y + \underset{\mbox{ {\footnotesize $\rho(c)\to \rho(d)$ in $T_{a\leftrightarrows b}$}}}{\sum} (\mathbf{L}_{\rho(c)\rho(d)\,\mbox{-}Y} - \mathbf{L}_{\rho(d)\rho(c)\,\mbox{-}Y})\big)|_{ X(\wt)_{a\leftrightarrows b}}$ \\
&=& $\big(\mathbf{L}_Y + \underset{\mbox{ {\footnotesize $c'\to d'$ in $T_{a\leftrightarrows b}$}}}{\sum} (\mathbf{L}_{c'd'\,\mbox{-}Y} - \mathbf{L}_{d'c'\,\mbox{-}Y})\big)|_{ X(\wt)_{a\leftrightarrows b}}$ \\
&=& $\mathfrak{L}_Y(\wt_{a\leftrightarrows b})|_{ X(\wt)_{a\leftrightarrows b}}$ \\
&=& $\mathfrak{L}_Y(\wt_{a\leftrightarrows b})|_{ X(\wt_{a\leftrightarrows b})}$\\
&=& $\mathfrak{P}_Y(\wt_{a\leftrightarrows b} )$ by (\ref{PfromL}).\qedhere
\end{longtable}
\end{center}}
\end{proof}

\section{Majoritarian and pairwise projection}\label{ProjectionSection}

Using the representation of weak tournaments in Definition \ref{RepDef}, the next key step is to associate with each voting method $F$ and finite $Y\subseteq\mathcal{X}$ a voting method that we call its \textit{majoritarian projection} relative to $Y$. In fact, for convenience we will first do this for \textit{anonymous} voting methods defined on anonymous profiles and then for ordinary voting methods in our proof of the transfer lemma, Lemma \ref{transferlemone}.

\begin{definition} An \textit{anonymous voting method} is a function $F$ whose domain is a set of anonymous profiles such that for any  $\mathbf{P}\in\mathrm{dom}(F)$, $\varnothing\neq F(\mathbf{P})\subseteq X(\mathbf{P})$.
\end{definition}
\noindent Every axiom we have defined for voting methods applies in an obvious way to anonymous voting methods, so we will not repeat the definitions here.

\begin{definition}\label{majoritarianProjDef} Given an anonymous voting method $F$ and finite $Y\subseteq\mathcal{X}$, define the \textit{majoritarian projection of $F$ relative to $Y$}, $F_{maj,Y}$, as the anonymous voting method with 
\[\mathrm{dom}(F_{maj,Y})=\{\mathbf{P}\in\anonprof\mid X(\mathbf{P})\subseteq Y\mbox{ and }\mathfrak{P}_Y(M(\mathbf{P}))\in\mathrm{dom}(F)\} \] and \[F_{maj,Y}(\mathbf{P})=F(\mathfrak{P}_Y(M(\mathbf{P}))).\]
\end{definition}

We define an analogous notion of pairwise projection using Definition \ref{RepDefWeighted}.

\begin{definition}\label{PairProjDef} Given an anonymous  voting method $F$, finite $Y\subseteq\mathcal{X}$, and $m\in\mathbb{Z}^+$, define the \textit{pairwise projection of $F$ relative to $Y$ and $m$}, $F_{pair,Y,m}$, as the anonymous voting method with 
\begin{eqnarray*}\mathrm{dom}(F_{pair,Y,m})&=&\{\mathbf{P}\in\anonprof\mid  X(\mathbf{P})\subseteq Y,\mbox{ all weights in }\mathcal{M}(\mathbf{P})\mbox{ are $\leq m$, }\\
&&\qquad\qquad\qquad\quad\,\,\mbox{and }\mathfrak{P}_{Y,m}(\psi\mathcal{M}(\mathbf{P}))\in\mathrm{dom}(F)\} \end{eqnarray*} and \[F_{pair,Y,m}(\mathbf{P})=F(\mathfrak{P}_{Y,m}(\psi\mathcal{M}(\mathbf{P})))\] where $\psi=2\times (|Y|-2)!$.
\end{definition}
Note that multiplying $\mathcal{M}(\mathbf{P})$ by $\psi$ gives us a weighted tournament suitable for representation by Definition \ref{RepDefWeighted}. Though the resulting profile $\mathfrak{P}_{Y,m}(\psi\mathcal{M}(\mathbf{P}))$ does not represent $\mathcal{M}(\mathbf{P})$, but instead $\psi \mathcal{M}(\mathbf{P})$, this does not matter, since we only use $\mathfrak{P}_{Y,m}(\psi\mathcal{M}(\mathbf{P}))$ to select winners according to $F$.

It is easy to see that the majoritarian (resp.~pairwise) projection of $F$ is indeed  majoritarian (resp.~pairwise). Moreover, majoritarian and pairwise projection preserve neutrality thanks to the fact in Proposition \ref{TransposeProp} that representation commutes with transposition (in contrast to the McGarvey/Debord construction) plus Lemma \ref{TransposeCommute}. Recall that given $Y\subseteq \mathcal{X}$, the set $Y_{a\leftrightarrows b}$ is the image of $Y$ under the permutation of $\mathcal{X}$ that only swaps $a$ and $b$.

\begin{lemma}\label{Ismajoritarian} Given an anonymous voting method $F$, finite $Y\subseteq\mathcal{X}$, and $m\in\mathbb{Z}^+$:
\begin{enumerate}
\item\label{Ismajoritariana} $F_{maj,Y}$ is majoritarian;
\item\label{Ismajoritarianb} $F_{pair,Y,m}$ is pairwise;
\item\label{Ismajoritarianc} if $F$ is neutral, then $F_{maj,Y}$ and $F_{pair,Y,m}$ are neutral.
\end{enumerate}
\end{lemma}

\begin{proof} Parts \ref{Ismajoritariana} and \ref{Ismajoritarianb} are immediate from Definitions \ref{majoritarianProjDef} and \ref{PairProjDef}, respectively. For part \ref{Ismajoritarianc}, assuming $F$ is neutral, given any $\mathbf{P}\in\mathrm{dom}(F_{maj,Y})$ and $\mathbf{P}_{a\leftrightarrows b}\in\mathrm{dom}(F)$, we have $a,b\in Y$, and
\begin{center}
\begin{tabular}{rcll}
$F_{maj,Y}(\mathbf{P}_{a\leftrightarrows b})$ & = & $F(\mathfrak{P}_Y(M(\mathbf{P}_{a\leftrightarrows b})))$ & by Definition \ref{majoritarianProjDef} \\
&=&  $F(\mathfrak{P}_Y(M(\mathbf{P})_{a\leftrightarrows b}))$  & by Lemma \ref{TransposeCommute}.\ref{TransposeCommute1} \\ 
&=&  $F(\mathfrak{P}_Y(M(\mathbf{P}))_{a\leftrightarrows b})$  & by Proposition \ref{TransposeProp}.\ref{TransposeProp1} since $a,b\in Y$\\ 
&=&  $F(\mathfrak{P}_Y(M(\mathbf{P})))_{a\leftrightarrows b}$  & since $F$ is neutral \\ 
&=&  $F_{maj,Y}(\mathbf{P})_{a\leftrightarrows b}$  & by Definition \ref{majoritarianProjDef},\\ 
\end{tabular}
\end{center}
so $F_{maj,Y}$ is neutral. The proof for $F_{pair,Y,m}$ is analogous,  using Definition \ref{PairProjDef}, Proposition \ref{TransposeProp}.\ref{TransposeProp2}, and also Lemma \ref{TransposeCommute}.\ref{TransposeCommute3}. \end{proof}

\section{Preservation of axioms}\label{PreservationSection}

We now turn to the key question of which axioms on voting methods are preserved by majoritarian and pairwise projection, besides neutrality (Lemma \ref{Ismajoritarian}.\ref{Ismajoritarianc}). These preservation results will be key to proving the transfer lemmas from Section~\ref{ProofStrategy}.

\subsection{Axioms as sets} To prove general results about the preservation of axioms by majoritarian and pairwise projections, we will make a move that is unusual in social choice theory but standard in mathematical logic, by treating axioms on voting methods as mathematical objects. Thinking syntactically, we could introduce a formal language and define a \textit{universal axiom} as a formula of the form  \[\forall\mathbf{P}_1\dots\forall\mathbf{P}_n \varphi(\mathbf{P}_1,\dots,\mathbf{P}_n)\] where $\varphi$ contains no profile quantifiers (but possibly quantifiers over voters and candidates). Instead, we will view universal axioms semantically as follows.\footnote{For those familiar with mathematical logic (see, e.g., \cite{Chang1990}), consider a two-sorted model $\mathfrak{A}$ whose domains contains all profiles and all nonempty finite sets of candidates. Where $\mathrm{P}_i$ is a variable for a profile and $\mathrm{X}_i$ is a variable for a set of candidates, a formula $\varphi(\mathrm{P}_1,\dots,\mathrm{P}_n,\mathrm{X}_1,\dots,\mathrm{X}_n)$ defines the set of all tuples $\langle \mathbf{P}_1,\dots,\mathbf{P}_n,X_1,\dots, X_n\rangle$ such that 
\[ \mathfrak{A}\vDash \varphi(\mathrm{P}_1,\dots,\mathrm{P}_n,\mathrm{X}_1,\dots,\mathrm{X}_n) [\mathbf{P}_1,\dots,\mathbf{P}_n,X_1,\dots, X_n].\]
We repackage $\langle \mathbf{P}_1,\dots,\mathbf{P}_n,X_1,\dots, X_n\rangle$ as $\langle \langle \mathbf{P}_1, X_1\rangle,\dots,\langle \mathbf{P}_n,X_n\rangle  \rangle$ in Definition \ref{UniAxioms}.} Let $\mathscr{S}$ be the set of all pairs $\langle\mathbf{P},Y\rangle$ where $\mathbf{P}$ is an anonymous profile and $Y\subseteq X(\mathbf{P})$.

\begin{definition}\label{UniAxioms} An \textit{$n$-ary universal axiom} is a subset $Ax\subseteq \mathscr{S}^n$. An anonymous voting method $F$ \textit{satisfies} $Ax$ if $\{\langle \mathbf{P},F(\mathbf{P})\rangle \mid \mathbf{P}\in\mathrm{dom}(F)\}^n\subseteq Ax$. \end{definition}

We illustrate this concept by formalizing some of the axioms discussed in Section \ref{ImpossSection}.

\begin{example}\label{ExampleAxioms} $\,$
\begin{enumerate}
\item Given an anonymous profile $\mathbf{P}$, \textit{a Condorcet winner} is a candidate $a\in X(\mathbf{P})$ such that for all $b\in X(\mathbf{P})\setminus\{a\}$, the margin of $a$ over $b$ in $\mathbf{P}$ is positive. A voting method $F$ is \textit{Condorcet consistent} if for all $\mathbf{P}\in \mathrm{dom}(F)$, if $\mathbf{P}$ has a Condorcet winner, then $F(\mathbf{P})$ contains only that candidate. Condorcet consistency corresponds to the unary axiom 
\begin{eqnarray*}&&\mathsf{CC}=\{\langle \mathbf{P},X\rangle\mid \mbox{if $\mathbf{P}$ has a Condorcet winner,}\\
&&\qquad\qquad\qquad\;\;\;\mbox{then $X$ contains only that candidate}\}.\end{eqnarray*}

\item majoritarianism corresponds to the binary axiom 
\[\mathsf{majoritarian}=\{\langle \langle \mathbf{P},X\rangle, \langle \mathbf{P}',X'\rangle\rangle\mid \mbox{if $M(\mathbf{P})=M(\mathbf{P}')$, then $X=X'$}\}.\]
\item neutrality  corresponds to the binary axiom
\[\mathsf{N}=\{\langle \langle \mathbf{P},X\rangle, \langle \mathbf{P}',X'\rangle\rangle\mid \mbox{if $\mathbf{P}'=\mathbf{P}_{a\leftrightarrows b}$, then $X'=X_{a\leftrightarrows b}$}\}.\]
\item binary quasi-resoluteness  corresponds to the axiom
\[\qquad\mathsf{BQR}=\{\langle \mathbf{P},X\rangle\mid \mbox{if $M(\mathbf{P})$ is a tournament with 2 nodes, then $|X|=1$}\}.\]
\item  quasi-resoluteness corresponds to the axiom 
\[\mathsf{QR}=\{\langle \mathbf{P},X\rangle\mid \mbox{if $\mathcal{M}(\mathbf{P})$ is uniquely weighted, then $|X|=1$}\}.\]
\end{enumerate}
\end{example}

\subsection{Majoritarian and homogeneously pairwise axioms}

\noindent Inspired by the notions of majoritarian/pairwise voting methods, we will define the notions of majoritarian/pairwise axioms. In fact, we are interested in what we will call \textit{homogeneously} pairwise axioms, inspired by the notion of a voting method $F$ being homogeneous (recall Definition \ref{Invariance}).

\begin{definition}\label{majoritarianpairwiseAx} Let $Ax$ be an $n$-ary universal axiom. We say that $Ax$ is \textit{majoritarian} (resp.~\textit{homogeneously pairwise}) if for all anonymous profiles $\mathbf{P}_1,\dots,\mathbf{P}_n$ and $\mathbf{Q}_1,\dots,\mathbf{Q}_n$ such that for $1\leq i\leq n$, $M(\mathbf{P}_i)=M(\mathbf{Q}_i)$ (resp.~for which there is a $k\in \mathbb{Z}^+$ such that for $1\leq i\leq n$, $k\mathcal{M}(\mathbf{P}_i)=\mathcal{M}(\mathbf{Q}_i)$), we have that  for all $X_1\subseteq X(\mathbf{P}_1),\dots , X_n\subseteq X(\mathbf{P}_n)$, 
 \[\mbox{$\langle \langle \mathbf{P}_1,X_1\rangle,\dots, \langle \mathbf{P}_n,X_n\rangle\rangle \in Ax$ if and only if $\langle \langle \mathbf{Q}_1,X_1\rangle,\dots, \langle \mathbf{Q}_n,X_n\rangle\rangle\in Ax$.}\]\end{definition}
 
 \noindent Note that being majoritarian implies being homogeneously pairwise.
 
 The axioms from Example \ref{ExampleAxioms} can be classified using Definition \ref{majoritarianpairwiseAx} as follows.  
 \begin{example}\label{majoritarianvsNonmajoritarian} $\,$
 \begin{enumerate}
 \item Condorcet consistency is majoritarian because if $\langle \mathbf{P},X\rangle\in \mathsf{CC}$ and $M(\mathbf{P})=M(\mathbf{Q})$, then clearly $\langle \mathbf{Q},X\rangle\in \mathsf{CC}$.\footnote{Other examples of majoritarian axioms include the Condorcet loser criterion \cite{Nurmi1987} and Smith criterion \cite{Smith1973}.}
 \item The majoritarian axiom is obviously majoritarian.
  \item Neutrality is not homogeneously pairwise, because there are profiles $\mathbf{P}$, $\mathbf{P}'$, $\mathbf{Q}$, $\mathbf{Q}'$ such that $\mathbf{P}'=\mathbf{P}_{a\leftrightarrows b}$, $\mathcal{M}(\mathbf{P})=\mathcal{M}(\mathbf{Q})$, and $\mathcal{M}(\mathbf{P}')=\mathcal{M}(\mathbf{Q}')$, but $\mathbf{Q}'\neq\mathbf{Q}_{a\leftrightarrows b}$. Then for $X,X'\subseteq X(\mathbf{P})$ such that $X'\neq X_{a\leftrightarrows b}$, we have $\langle \langle \mathbf{P},X\rangle, \langle \mathbf{P}',X'\rangle\rangle\not\in\mathsf{N}$ but trivially $\langle \langle \mathbf{Q},X\rangle, \langle \mathbf{Q}',X'\rangle\rangle\in\mathsf{N}$.
 \item $\mathsf{BQR}$ is obviously majoritarian.
 \item $\mathsf{QR}$ is homogeneously pairwise (for $\mathcal{M}(\mathbf{P})$ is uniquely weighted iff $k\mathcal{M}(\mathbf{P})$ is), but it is not majoritarian, because there are profiles $\mathbf{P},\mathbf{P}'$ such that $M(\mathbf{P})=M(\mathbf{P}')$, $\mathcal{M}(\mathbf{P})$ is uniquely weighted, but $\mathcal{M}(\mathbf{P}')$ is not. Then for $X\subseteq X(\mathbf{P})$ with $|X|>1$, $\langle \mathbf{P},X\rangle\not\in \mathsf{QR}$ but trivially $\langle \mathbf{P}',X\rangle\in \mathsf{QR}$.
 \end{enumerate}
 \end{example}

We now verify that all majoritarian (resp.~homogeneously pairwise) universal axioms are preserved by the operation of majoritarian (resp.~pairwise) projection from Section \ref{ProjectionSection}.

\begin{lemma}\label{majoritarianpreserved} Let $F$ be an anonymous voting method. For any majoritarian \textnormal{(}resp.~homogeneously pairwise\textnormal{)} universal axiom $Ax$, if $F$ satisfies $Ax$, then $F_{maj,Y}$ \textnormal{(}resp.~$F_{pair,Y,m}$\textnormal{)} satisfies Ax.
\end{lemma}

\begin{proof} Suppose $Ax$ is a majoritarian $n$-ary universal axiom. To show that $F_{maj,Y}$ satisfies $Ax$, consider any $\mathbf{P}_1,\dots,\mathbf{P}_n\in \mathrm{dom}(F_{maj,Y})$. We must show that  \[\langle \langle \mathbf{P}_1,F_{maj,Y}(\mathbf{P}_1)\rangle,\dots, \langle \mathbf{P}_n,F_{maj,Y}(\mathbf{P}_n)\rangle\rangle\in Ax.\] By definition of $F_{maj,Y}$ (Definition \ref{majoritarianProjDef}), this is equivalent to \[\langle \langle \mathbf{P}_1,F(\mathfrak{P}_Y(M(\mathbf{P}_1)))\rangle,\dots, \langle \mathbf{P}_n,F(\mathfrak{P}_Y(M(\mathbf{P}_n)))\rangle\rangle\in Ax.\] By Lemma \ref{RepProp0}, for all $i$, $M(\mathbf{P}_i)=M(\mathfrak{P}_Y(M(\mathbf{P}_i)))$. Then since $Ax$ is majoritarian, setting $\mathbf{Q}_i=\mathfrak{P}_Y(M(\mathbf{P}_i))$ in Definition \ref{majoritarianpairwiseAx}, the previous  membership statement is equivalent to 
\[\langle \langle \mathfrak{P}_Y(M(\mathbf{P}_1)),F(\mathfrak{P}_Y(M(\mathbf{P}_1)))\rangle,\dots, \langle \mathfrak{P}_Y(M(\mathbf{P}_n)),F(\mathfrak{P}_Y(M(\mathbf{P}_n)))\rangle\rangle\in Ax,\] which holds by our assumption that $F$ satisfies $Ax$.

Now suppose $Ax$ is a homogeneously pairwise $n$-ary universal axiom. To show that $F_{pair,Y,m}$ satisfies $Ax$, consider any $\mathbf{P}_1,\dots,\mathbf{P}_n\in \mathrm{dom}(F_{pair,Y,m})$. We must show that  \[\langle \langle \mathbf{P}_1,F_{pair,Y,m}(\mathbf{P}_1)\rangle,\dots, \langle \mathbf{P}_n,F_{pair,Y,m}(\mathbf{P}_n)\rangle\rangle\in Ax.\] By definition of $F_{pair,Y,m}$ (Definition \ref{PairProjDef}), this is equivalent to \[\langle \langle \mathbf{P}_1,F(\mathfrak{P}_{Y,m}(\psi\mathcal{M}(\mathbf{P}_1)))\rangle,\dots, \langle \mathbf{P}_n,F(\mathfrak{P}_{Y,m}(\psi\mathcal{M}(\mathbf{P}_n)))\rangle\rangle\in Ax.\] By Lemma \ref{RepProp}, for all $i$, $\psi \mathcal{M}(\mathbf{P}_i)=\mathcal{M}(\mathfrak{P}_{Y,m}(\psi\mathcal{M}(\mathbf{P}_i)))$. Then since $Ax$ is homogeneously pairwise, setting $\mathbf{Q}_i= \mathfrak{P}_{Y,m}(\psi\mathcal{M}(\mathbf{P}_i))$ in Definition \ref{majoritarianpairwiseAx}, the previous  membership statement is equivalent to 
\begin{eqnarray*}&&\langle \langle \mathfrak{P}_{Y,m}(\psi\mathcal{M}(\mathbf{P}_1)),F(\mathfrak{P}_{Y,m}(\psi\mathcal{M}(\mathbf{P}_1)))\rangle,\dots, \\
&& \;\,\langle \mathfrak{P}_{Y,m}(\psi\mathcal{M}(\mathbf{P}_n)),F(\mathfrak{P}_{Y,m}(\psi\mathcal{M}(\mathbf{P}_n)))\rangle\rangle\in Ax,\end{eqnarray*} which holds by our assumption that $F$ satisfies $Ax$. \end{proof}

Lemma \ref{majoritarianpreserved} provides a method of extending impossibility theorems proved for majoritarian or pairwise methods to all voting methods. However, it does not go far enough for us, because the variable candidate axioms of binary $\gamma$ and $\alpha$-resoluteness are not homogeneously pairwise axioms (see Lemma \ref{PSWquasi}).

\subsection{Variable candidate axioms}

Next we extend Lemma \ref{majoritarianpreserved} to handle variable candidate axioms---axioms that concern adding or removing a candidate from each voter's ballot.

\begin{definition}\label{quasimajoritarian} Let $Ax$ be an $n$-ary universal axiom. We say that $Ax$ is a \textit{variable candidate axiom}  if for all profiles $\mathbf{P}_1,\dots,\mathbf{P}_n$ and $\mathbf{Q}_1,\dots,\mathbf{Q}_n$ such that 
\begin{enumerate}
\item\label{quasimajoritarianA}  for $1\leq i\leq n$,  $X(\mathbf{P}_i)=X(\mathbf{Q}_i)$, and 
\item\label{quasimajoritarianB} for $1\leq i,j\leq n$, if $\mathbf{P}_i=(\mathbf{P}_j)|_{ X(\mathbf{P}_i)}$, then $\mathbf{Q}_i=(\mathbf{Q}_j)|_{ X(\mathbf{Q}_i)}$, 
\end{enumerate}
 we have that for all ${X_1\subseteq X(\mathbf{P}_1)}$, \dots , $X_n\subseteq X(\mathbf{P}_n)$, 
 \begin{itemize}
 \item[(3)] if $\langle \langle \mathbf{Q}_1,X_1\rangle,\dots, \langle \mathbf{Q}_n,X_n\rangle\rangle\in Ax$, then $\langle \langle \mathbf{P}_1,X_1\rangle,\dots, \langle \mathbf{P}_n,X_n\rangle\rangle \in Ax$. 
 \end{itemize}
\end{definition}
The intuition is that if the $\mathbf{Q}_i$'s are related to each other in terms of restriction in all of the ways that the $\mathbf{P}_i$'s are, and  $X_1,\dots,X_n$ are acceptable choices for the $\mathbf{Q}_i$'s according to the relevant variable candidate axiom (e.g., binary $\gamma$ or $\alpha$-resoluteness), then they are also acceptable choices for the $\mathbf{P}_i$'s. 

\begin{example} Binary $\gamma$ and $\alpha$-resoluteness correspond to the following axioms:
\begin{eqnarray*}
&&\mathsf{BG}=\{\langle \langle \mathbf{P},X\rangle, \langle \mathbf{P}',X'\rangle,\langle \mathbf{P}'',X''\rangle\rangle\mid \mbox{if $\mathbf{P}'=\mathbf{P}_{-b}$ for some $b\in X(\mathbf{P})$,}  \\
&&\qquad\qquad\qquad\qquad\qquad\qquad\qquad\qquad \mathbf{P}''=\mathbf{P}|_{ \{a,b\}}\mbox{ for }a\in X(\mathbf{P}'), \\
&&\qquad\qquad\qquad\qquad\qquad\qquad\qquad\qquad a\in X'\mbox{, and }X''=\{a\}, \\
&&\qquad\qquad\qquad\qquad\qquad\qquad\qquad\qquad\mbox{then $a\in X$}  \};
\end{eqnarray*} 
\begin{eqnarray*}
&&\mathsf{AR}=\{\langle \langle \mathbf{P},X\rangle, \langle \mathbf{P}',X'\rangle,\langle \mathbf{P}'',X''\rangle\rangle\mid \mbox{if $\mathbf{P}'=\mathbf{P}_{-b}$ for some $b\in X(\mathbf{P})$,}  \\
&&\qquad\qquad\qquad\qquad\qquad\qquad\qquad\qquad \mathbf{P}''=\mathbf{P}|_{ \{a,b\}}\mbox{ for }a\in X(\mathbf{P}'), \\
&&\qquad\qquad\qquad\qquad\qquad\qquad\qquad\qquad a\in X'\mbox{, and }X''=\{a\}, \\
&&\qquad\qquad\qquad\qquad\qquad\qquad\qquad\qquad\mbox{then $|X|\leq |X'|$}  \}.
\end{eqnarray*} 
\end{example}

\begin{lemma}\label{PSWquasi} $\mathsf{BG}$ and $\mathsf{AR}$ are variable candidate axioms but neither majoritarian nor homogeneously pairwise axioms.
\end{lemma}

\begin{proof} We give only the proof for $\mathsf{BG}$, as the proof for $\mathsf{AR}$ is almost the same. Suppose (i) that $\mathbf{P},\mathbf{P}',\mathbf{P}''$ and $\mathbf{Q},\mathbf{Q}',\mathbf{Q}''$ are related as in Definition \ref{quasimajoritarian}.\ref{quasimajoritarianA}-\ref{quasimajoritarianB}. Further suppose (ii) that $\langle \langle \mathbf{Q},X\rangle, \langle\mathbf{Q}',X'\rangle,\langle\mathbf{Q}'',X''\rangle  \rangle\in\mathsf{BG}$. Finally, suppose (iii) that $\mathbf{P}'=\mathbf{P}_{-b}$ for some $b\in X(\mathbf{P})$, $\mathbf{P}''=\mathbf{P}|_{ \{a,b\}}$ for $a\in X(\mathbf{P}')$, $a\in X'$, and $X''=\{a\}$.  By (i) and (iii), we have that $\mathbf{Q}'=\mathbf{Q}_{-b}$ for $b\in X(\mathbf{Q})=X(\mathbf{P})$, $\mathbf{Q}''=\mathbf{Q}|_{ \{a,b\}}$ for $a\in X(\mathbf{Q}')=X(\mathbf{P}')$, $a\in X'$, and $X''=\{a\}$. Then by (ii), we have $a\in X$. Thus, we conclude that $\langle \langle \mathbf{P},X\rangle, \langle\mathbf{P}',X'\rangle,\langle\mathbf{P}'',X''\rangle  \rangle\in\mathsf{BG}$.

To see that $\mathsf{BG}$ is not homogeneously pairwise, it is easy to find  $\langle \langle \mathbf{P}_1,X_1\rangle, \langle \mathbf{P}_2,X_2\rangle, \langle \mathbf{P}_3,X_3\rangle\rangle$ and $\langle \langle \mathbf{Q}_1,X_1\rangle, \langle \mathbf{Q}_2,X_2\rangle, \langle \mathbf{Q}_3,X_3\rangle\rangle$ such that while $\mathcal{M}(\mathbf{P}_i)=\mathcal{M}(\mathbf{Q}_i)$ and $\mathbf{P}_2=(\mathbf{P}_1)_{-b}$, we have $\mathbf{Q}_2\neq (\mathbf{Q}_1)|_{ X(\mathbf{Q}_2)}$; then even if $\langle \langle \mathbf{P}_1,X_1\rangle, \langle \mathbf{P}_2,X_2\rangle, \langle \mathbf{P}_3,X_3\rangle\rangle \not\in \mathsf{BG}$, from $\mathbf{Q}_2\neq (\mathbf{Q}_1)|_{ X(\mathbf{Q}_2)}$ we trivially have $\langle \langle \mathbf{Q}_1,X_1\rangle, \langle \mathbf{Q}_2,X_2\rangle, \langle \mathbf{Q}_3,X_3\rangle\rangle \in \mathsf{BG}$.\end{proof}

Fortunately, not only majoritarian (resp.~homogeneously pairwise) axioms but also variable candidate axioms are preserved by our operation of majoritarian (resp.~pairwise) projection, thanks to the fact that our representation of weak (resp.~weighted weak) tournaments commutes with restriction, in contrast to McGarvey and Debord's constructions. \\

\begin{proposition}\label{quasimajoritarianTransfer} Let $F$ be an anonymous voting method satisfying a variable candidate universal axiom $Ax$. Then for any finite $Y\subseteq\mathcal{X}$, $F_{maj,Y}$ \textnormal{(}resp.~$F_{pair,Y,m}$\textnormal{)} satisfies $Ax$.
\end{proposition}

\begin{proof} Suppose $Ax$ is a majoritarian $n$-ary universal axiom. To show that $F_{maj,Y}$ satisfies $Ax$, consider any $\mathbf{P}_1,\dots,\mathbf{P}_n\in \mathrm{dom}(F_{maj,Y})$. We must show that  \[\langle \langle \mathbf{P}_1,F_{maj,Y}(\mathbf{P}_1)\rangle,\dots, \langle \mathbf{P}_n,F_{maj,Y}(\mathbf{P}_n)\rangle\rangle\in Ax.\] By definition of $F_{maj,Y}$, this is equivalent to 
\[\langle \langle \mathbf{P}_1,F(\mathfrak{P}_Y(M(\mathbf{P}_1)))\rangle,\dots, \langle \mathbf{P}_n,F(\mathfrak{P}_Y(M(\mathbf{P}_n)))\rangle\rangle\in Ax.\] 
By Proposition \ref{SubProp} and Lemma \ref{RestrictCommute}, we have  
\begin{eqnarray*}\mathfrak{P}_Y(M(\mathbf{P}_j))|_{ X(\mathfrak{P}_Y(M(\mathbf{P}_i)))} &=& \mathfrak{P}_Y(M(\mathbf{P}_j)|_{ X(\mathfrak{P}_Y(M(\mathbf{P}_i)))} ) \\
&=& \mathfrak{P}_Y(M((\mathbf{P}_j)|_{ X(\mathfrak{P}_Y(M(\mathbf{P}_i)))})) \\
&=& \mathfrak{P}_Y(M((\mathbf{P}_j)|_{ X(\mathbf{P}_i)})).\end{eqnarray*}
Thus, for $1\leq i,j\leq n$, we have that
\[\mbox{ if }\mathbf{P}_i=(\mathbf{P}_j)|_{ X(\mathbf{P}_i)}\mbox{, then }\mathfrak{P}_Y(M(\mathbf{P}_i))=\mathfrak{P}_Y(M(\mathbf{P}_j))|_{ X(\mathfrak{P}_Y(M(\mathbf{P}_i)))}.\]
Then since $Ax$ is a variable candidate axiom, setting $\mathbf{Q}_i= \mathfrak{P}_Y(M(\mathbf{P}_i))$ in Definition~\ref{quasimajoritarian}, the  membership statement displayed above is a consequence of
\[\langle \langle \mathfrak{P}_Y(M(\mathbf{P}_1)),F(\mathfrak{P}_Y(M(\mathbf{P}_1)))\rangle,\dots, \langle \mathfrak{P}_Y(M(\mathbf{P}_n)),F(\mathfrak{P}_Y(M(\mathbf{P}_n)))\rangle\rangle\in Ax,\] which holds by our assumption that $F$ satisfies $Ax$. The proof for the pairwise case is  analogous.\end{proof}

 Proposition \ref{quasimajoritarianTransfer} is the key to extending impossibility theorems involving variable candidate axioms proved for majoritarian or pairwise methods to all voting methods, as we show next.
 
 \subsection{Proof of the transfer lemmas} We are now in a position to prove the transfer lemmas introduced in Section \ref{ProofStrategy}.
 
\TransferLemOne*

\begin{proof} Fix a finite $Y\subseteq\mathcal{X}$, and let $F$ be a voting method satisfying the assumption. For any profile $\mathbf{P}$, let $a(\mathbf{P})$ be the anonymous profile such that $a(\mathbf{P})(L)=|\{i\in V(\mathbf{P})\mid \mathbf{P}_i=L\}|$ for $L\in \mathcal{L}(X(\mathbf{P}))$. Note that if $F$ is anonymous, then for any $\mathbf{P},\mathbf{P}'\in\mathrm{dom}(F)$, if $a(\mathbf{P})=a(\mathbf{P}')$, then $F(\mathbf{P})=F(\mathbf{P}')$. Also note that for any $a,b\in\mathcal{X}$, $a(\mathbf{P})_{a\leftrightarrows b} = a(\mathbf{P}_{a\leftrightarrows b} )$.

Fix a set $V_0\subseteq\mathcal{V}$ of voters with $|V|= |Y|!$, and consider the following restricted domain of profiles: \[D=\{\mathbf{P}\in\prof\mid X(\mathbf{P})\subseteq Y, V(\mathbf{P})=V_0\}.\] 
Now define the following domain of anonymous profiles:
\[D_A = \{\mathbf{P}\in\anonprof\mid X(\mathbf{P})\subseteq Y, \underset{L\in \mathcal{L}(X(\mathbf{P}))}{\sum} \mathbf{P}(L) = |Y|! \}.\]
For each anonymous profile $\mathbf{P}\in D_A$, fix a de-anonymized profile $d(\mathbf{P})\in D$ such that $a(d(\mathbf{P}))=\mathbf{P}$. Now define an anonymous voting method $G$ on $D_A$ by
\[G(\mathbf{P})= F(d(\mathbf{P})).\]
Then since $F$ satisfies anonymity, neutrality, binary quasi-resoluteness, binary $\gamma$, and $\alpha$-resoluteness, it is easy to see that the anonymous method $G$ satisfies neutrality, binary quasi-resoluteness, binary $\gamma$, and $\alpha$-resoluteness. For instance, for neutrality, since $\mathbf{P}=a(d(\mathbf{P}))$, we have  $\mathbf{P}_{a\leftrightarrows b} =  a(d(\mathbf{P}))_{a\leftrightarrows b}=a(d(\mathbf{P})_{a\leftrightarrows b})$, which with $\mathbf{P}_{a\leftrightarrows b}=a(d(\mathbf{P}_{a\leftrightarrows b}))$ implies $a(d(\mathbf{P}_{a\leftrightarrows b}))=a(d(\mathbf{P})_{a\leftrightarrows b})$. Then
\begin{eqnarray*}
G(\mathbf{P}_{a\leftrightarrows b}) &=&F(d(\mathbf{P}_{a\leftrightarrows b}))\\
&=& F(d(\mathbf{P})_{a\leftrightarrows b})  \mbox{ since }a(d(\mathbf{P}_{a\leftrightarrows b}))=a(d(\mathbf{P})_{a\leftrightarrows b})\mbox{ and }F\mbox{ is anonymous}\\
&=& F(d(\mathbf{P}))_{a\leftrightarrows b}\mbox{ by the neutrality of }F\\
&=& G(\mathbf{P})_{a\leftrightarrows b}.
\end{eqnarray*}
The proofs that $G$ satisfies the other axioms involve similar reasoning. Now recall from Definition \ref{majoritarianProjDef} that
\[\mathrm{dom}(G_{maj,Y})=\{\mathbf{P}\in\anonprof\mid X(\mathbf{P})\subseteq Y\mbox{ and }\mathfrak{P}_Y(M(\mathbf{P}))\in\mathrm{dom}(G)\}. \]
Since $\mathrm{dom}(G)=D_A$ and by Definition \ref{RepDef}, $\mathfrak{P}_Y(M(\mathbf{P}))\in D_A$ for all $\mathbf{P}$ with $X(\mathbf{P})\subseteq Y$, it follows that $\mathrm{dom}(G_{maj,Y})=\{\mathbf{P}\in\anonprof\mid X(\mathbf{P})\subseteq Y\}$. Moreover,  $G_{maj,Y}$ is an anonymous voting method that is majoritarian and neutral by Lemma \ref{Ismajoritarian}, satisfies binary quasi-resoluteness by Lemma \ref{majoritarianpreserved} (for majoritarian axioms), and satisfies binary $\gamma$ and $\alpha$-resoluteness by Lemma \ref{PSWquasi} and Proposition \ref{quasimajoritarianTransfer}. Finally, define a voting method $H$ such that for any profile $\mathbf{P}$, we have $H(\mathbf{P})=G_{maj,Y}(a(\mathbf{P}))$. Then it is easy to see that $H$ is a majoritarian  voting method with domain $\{\mathbf{P}\in\prof\mid X(\mathbf{P})\subseteq Y\}$ satisfying anonymity, neutrality, binary quasi-resoluteness, binary $\gamma$, and $\alpha$-resoluteness. \end{proof}

\TransferLemTwo*

\begin{proof} The proof is analogous to that of Lemma \ref{transferlemone} above, only now we fix a set $V_0\subseteq\mathcal{V}$ of voters with $|V|= m|Y|!$, since the profile $\mathfrak{P}_{Y,m}(\mathcal{M}(\mathbf{P}))$ from Definition \ref{RepDefWeighted} has $m|Y|!$ voters, and we move from the anonymous method $G$ satisfying neutrality, quasi-resoluteness, and binary $\gamma$ to $G_{pair,Y,m}$ being an anonymous voting method on the domain \[\{\mathbf{P}\in\anonprof\mid X(\mathbf{P})\subseteq Y,\,m\mbox{ is the maximum weight in }\mathcal{M}(\mathbf{P})\}\]  that is pairwise and neutral by Lemma \ref{Ismajoritarian}, quasi-resolute by Lemma \ref{majoritarianpreserved} (for pairwise axioms), and satisfies binary $\gamma$ by Lemma \ref{PSWquasi} and Proposition \ref{quasimajoritarianTransfer}. \end{proof}

This completes the proofs of Theorems \ref{Impossibility1}.\ref{Impossibility1a} and \ref{Impossibility4} as in Section \ref{ProofStrategy}.

\section{Conclusion}\label{Conclusion}

In this paper, we have proved impossibility theorems illustrating the tradeoff between binary expansion consistency and the power of voting methods to narrow down the set of winners. For further research, a natural refinement of these theorems would optimize not only the number of candidates but also the number of \textit{voters} needed to obtain the impossibility results, in the spirit of \cite{Brandt2017,Brandt2021}.\footnote{Inspection of our proofs shows that we use $6!=760$ voters to generate the profiles in our proof of Theorem \ref{Impossibility1}.\ref{Impossibility1a} and $12\times 4!=288$ voters to generate the profiles in our proof of Theorem~\ref{Impossibility4}.} As for the practical consequences of these theorems for the design of voting methods, a natural response is to seek variants of binary expansion consistency, still aimed at mitigating spoiler effects, that are consistent with the quasi-resoluteness property studied in this paper (see, e.g., \cite{HP2022}). 

Another avenue for further research concerns applying the proof strategy of this paper to other impossibility theorems in voting theory. Though the proof strategy may appear tailored to the specific axioms appearing in our impossibility theorems, in fact it applies more broadly. Our majoritarian and pairwise projections preserve neutrality (Lemma \ref{Ismajoritarian}.\ref{Ismajoritarianc}), all majoritarian axioms and  homogeneously pairwise axioms (Lemma \ref{majoritarianpreserved}), and all variable candidate universal axioms (Proposition~\ref{quasimajoritarianTransfer}). Furthermore, they preserve other well-known axioms such as monotonicity\ and Maskin monotonicity.\footnote{Monotonicity states that if $x\in F(\mathbf{P})$ and $\mathbf{P}'$ is obtained from $\mathbf{P}$ by moving $x$ up in one voter's ballot, keeping the relative order of all other candidates the same, then $x\in F(\mathbf{P}')$. Maskin monotonicity (for non-anonymous profiles) states that if $x\in F(\mathbf{P})$ and $\mathbf{P}'$ is a profile with the same voters such that if a voter ranks $x$ above $y$ in $\mathbf{P}$, then the voter still ranks $x$ above $y$ in $\mathbf{P}'$, then $x\in F(\mathbf{P}')$ (which can be suitably adapted to anonymous profiles). To see that these axioms are preserved, one can check that if profiles $\mathbf{P}$ and $\mathbf{P}'$ are related by a monotonic (resp.~Maskin monotonic) improvement of a candidate $x$, then $\mathfrak{P}_Y(M(\mathbf{P}))$ and $\mathfrak{P}_Y(M(\mathbf{P}'))$ are related in the same way, as are $\mathfrak{P}_{Y,m}(\mathcal{M}(\mathbf{P}))$ and $\mathfrak{P}_{Y,m}(\mathcal{M}(\mathbf{P}'))$, so the monotonicity (resp.~Maskin monotonicity) of  $F$ implies the monotonicity (resp.~Maskin monotonicity) of $F_{maj,Y}$ and $F_{pair,Y,m}$.} Thus, impossibility theorems involving monotonicity (resp.~Maskin monotonicity) together with members of the other classes of axioms above fall under the scope of our proof strategy. We leave it to future work to explore impossibility theorems of such forms.

\subsection*{Acknowledgements}

For helpful comments, we thank the anonymous referee and Jennifer Wilson. We are also grateful for the opportunity to present versions of this work at the Logic@IITK seminar in February 2021, the AMS Special Session on The Mathematics of Decisions, Elections and Games at the Joint Mathematics Meeting in April 2022, and the UMD Logic Seminar in October 2022.

\appendix

\section{SAT solving for majoritarian methods}\label{SATmajoritarian}

In this Appendix and the next, we explain how the SAT solving approach of Brandt and Geist \cite{Brandt2016} can be applied to prove our Propositions \ref{FromSatThm1} and \ref{FromQMFinal} for majoritarian and pairwise voting methods, respectively. 

Jupyter notebooks with our SAT formalization are available at \\ \href{https://github.com/szahedian/binary-gamma}{https://github.com/szahedian/binary-gamma}.

\subsection{Results}\label{SATResults}

Instead of working with profiles and majoritarian methods, our first SAT approach works with tournaments and tournament solutions.

\begin{definition} A \textit{weak tournament solution} (resp.~\textit{tournament solution}) is a function $F$ on a set of weak tournaments  (resp.~tournaments) such that for all $\wt\in\mathrm{dom}(F)$, we have $\varnothing\neq F(\wt)\subseteq X(\wt)$, where $X(\wt)$ is the set of nodes of $T$.
\end{definition}

The axioms of neutrality, binary quasi-resoluteness, binary $\gamma$, and $\alpha$-resoluteness from Section \ref{IndividualToSocial} apply in the obvious way to weak tournament solutions.

\begin{definition}\label{TournAx} Let $F$ be a weak tournament solution.
\begin{enumerate}
\item $F$ is \textit{neutral} if for all $T,T'\in\mathrm{dom}(F)$, $F(T')=h[F(T)]$ for any  isomorphism $h:T\to T'$.
\item $F$ satisfies \textit{binary quasi-resoluteness} if for any tournament $\wt\in \mathrm{dom}(F)$ with $|X(\wt)|=2$, we have $|F(\wt)|=1$.
\item $F$ satisfies \textit{binary $\gamma$} if for every $\wt\in \mathrm{dom}(F)$ and $a,b\in X(\wt)$, if $a\in F(\wt_{-b})$ and $F(\wt|_{ \{a,b\}})=\{a\}$, then $a\in F(\wt)$.
\item\label{TournAx2} $F$ satisfies \textit{$\alpha$-resoluteness} if for any $\wt\in \mathrm{dom}(F)$ and $a,b\in X(\wt)$, if $a\in F(\wt_{-b})$ and $ F(\wt|_{ \{a,b\}})=\{a\}$, then $|F(\wt)|\leq |F(\wt_{-b})|$.
\end{enumerate}
\end{definition}

We used SAT to find the dividing line in terms of number of candidates between the consistency and inconsistency of these axioms.

\begin{proposition}\label{SatThm0}  $\,$
\begin{enumerate}
\item\label{SatThm02} For any finite $Y\subseteq\mathcal{X}$ with $|Y|\geq 6$, there is no neutral tournament solution  satisfying binary quasi-resoluteness, binary $\gamma$, and $\alpha$-resoluteness on the domain $\{T\in \mathbb{T}\mid X(T)\subseteq Y\}$.
\item\label{SatThm01} There is a neutral weak tournament solution satisfying binary \\ quasi-resoluteness, binary $\gamma$, and $\alpha$-resoluteness on the \\ domain $\{\wt\in \mathbb{WT}\mid |X(\wt)|\leq 5\}$.
\end{enumerate}
\end{proposition}

This proposition is a consequence of Proposition \ref{CanSatThm} and Lemma \ref{CanProbLem} below.

\subsection{SAT encoding: canonical tournaments and solutions}\label{SAT1}

For computational feasibility, following Brandt and Geist \cite{Brandt2016} we work with tournament solutions on a smaller domain of \textit{canonical} tournaments up to a certain number of candidates.

\begin{definition} Given $\wt,\wt'\in \mathbb{WT}$, let $\wt\cong \wt'$ if $\wt$ and $\wt'$ are isomorphic. For each $\wt\in \mathbb{WT}$, define the equivalence class $[\wt]=\{\wt'\in \mathbb{WT}\mid \wt\cong \wt'\}$ and pick a \textit{canonical representative} $\wt_C\in [\wt]$. Let $\mathbb{WT}_{\!\!C} =\{\wt_C\mid T\in \mathbb{WT}\}$ and $\mathbb{T}_C =\{T_C\mid T\in \mathbb{T}\}$ be the sets of \textit{canonical weak tournaments} and \textit{canonical tournaments}, respectively. A \textit{canonical weak tournament solution} (resp.~\textit{canonical tournament solution}) is a weak tournament solution (resp.~tournament solution) whose domain is a subset of $\mathbb{WT}_{\!\!C}$ (resp.~$\mathbb{T}_C$).\end{definition}

Next we consider analogues for canonical weak tournament solutions of the axioms on weak tournament solutions from Definition \ref{TournAx}. Binary quasi-resoluteness is the most straightforward.

\begin{definition} A canonical weak tournament solution $F$ satisfies \textit{binary quasi-resoluteness} if for every $T\in \mathbb{T}_C$ with $|X(T)|=2$, we have $|F(T)|=1$.
\end{definition}

As for neutrality, there is a one-to-one correspondence between neutral weak tournament solutions and canonical weak tournament solutions satisfying what Brandt and Geist call the \textit{orbit condition}. Given $\wt\in\mathbb{WT}_C$ and $a,b\in X(\wt)$, let $a\sim_T b$ if there exists an automorphism $h$ of $\wt$ such that $h(a)=b$. The \textit{orbit of $a$ in $\wt$} is  $\{b\in \wt\mid a\sim_T b\}$. Let $\mathcal{O}_{\wt}$ be the set of all orbits of elements of $\wt$.

\begin{definition}\label{OrbitDef} Given a canonical weak tournament $\wt$ and $Y\subseteq X(\wt)$, we say that $(\wt,Y)$ satisfies the \textit{orbit condition} if for all $O\in \mathcal{O}_\wt$, we have $O\subseteq Y$ or $O\cap Y=\varnothing$. A canonical weak tournament solution $F$ satisfies the \textit{orbit condition} if for all $\wt\in \mathrm{dom}(F)$, $(\wt,F(\wt))$ satisfies the orbit condition. Let \[\mathscr{O}(\wt)=\{Y\subseteq X(\wt)\mid Y\neq\varnothing\mbox{ and } (\wt,Y)\mbox{ satisfies the orbit condition}\}.\]
\end{definition}

To state the one-to-one correspondence, we need the following preliminary lemma.

\begin{lemma}\label{IsoLem}  Let $F$ be a canonical weak tournament solution satisfying the orbit condition. Then for any $\wt\in\mathbb{WT}$ and isomorphisms $g:\wt_C\to \wt$ and $h:\wt_C\to \wt$, we have $g[F(\wt_C)]=h[F(\wt_C)]$.
\end{lemma}

\begin{proof} Suppose $a\in g[F(\wt_C)]$, so $g^{-1}(a)\in F(\wt_C)$. Since $h^{-1}\circ g$ is an automorphism of $\wt_C$ such that $h^{-1}(g(g^{-1}(a)))=h^{-1}(a)$, we have $g^{-1}(a)\sim_{T_C} h^{-1}(a)$. Then since $F$ satisfies the orbit condition, $g^{-1}(a)\in F(\wt_C)$ implies $h^{-1}(a)\in F(\wt_C)$, so $a\in h[F(\wt_C)]$. The other direction is analogous. \end{proof}

Lemma \ref{IsoLem} ensures that the function $F^*$ in the following is well defined.

\begin{lemma}\label{CanLem} $\,$
\begin{enumerate}
\item\label{CanLem1} Given a canonical weak tournament solution $F$ satisfying the orbit condition, the function $F^*$ on $\{\wt\in \mathbb{WT}\mid \wt_C\in \mathrm{dom}(F)\}$ defined by $F^*(\wt)= h_\wt[F(\wt_C)]$ for an isomorphism $h_\wt:\wt_C\to \wt$ is a neutral weak tournament solution. 
\item\label{CanLem2} Given a neutral weak tournament solution $F$, the function $F_*$ on $\{\wt_C\mid \wt\in\mathrm{dom}(F)\}$ defined by $F_*(\wt_C)=F(\wt_C)$ is a canonical weak tournament solution satisfying the orbit condition.
\item For any canonical weak tournament solution $F$ satisfying the orbit condition, $(F^*)_*=F$; and for any neutral weak tournament solution $F$, $(F_*)^*=F$.
\end{enumerate}
\end{lemma}

Next we state the analogues of binary $\gamma$ and $\alpha$-resoluteness for canonical weak tournament solutions. Given weak tournaments $\wt$ and $\wt'$, an \textit{embedding of $\wt'$ into $\wt$} is an injective function $e:X(\wt')\to X(\wt)$ such that for all $a,b\in X(\wt')$, we have $a\to b$ in $\wt'$ if and only if $e(a)\to e(b)$ in $\wt$, in which case we write $e: \wt'\hookrightarrow \wt$.

\begin{definition} Let $F$ be a canonical weak tournament solution.
\begin{enumerate}
\item $F$ satisfies \textit{canonical binary $\gamma$} if for any $\wt,\wt',\wt''\in \mathrm{dom}(F)$, if $|X(\wt)| = |X(\wt')| + 1$, $|X(\wt'')|=2$, and there are embeddings $e:\wt'\hookrightarrow \wt$ and $f:\wt''\hookrightarrow \wt$ such that $X(\wt)\setminus e[X(\wt')]\subseteq f[X(\wt'')]$, then for $a\in e[X(\wt')]\cap  f[X(\wt'')]$, if $e^{-1}(a)\in F(\wt')$ and $F(\wt'')=\{f^{-1}(a)\}$, then $a\in F(\wt)$.

\item $F$ satisfies \textit{canonical $\alpha$-resoluteness} if for any $\wt,\wt',\wt''\in \mathrm{dom}(F)$, if $|X(\wt)| = |X(\wt')| + 1$, $|X(\wt'')|=2$, and there are embeddings $e:\wt'\hookrightarrow \wt$ and $f:\wt''\hookrightarrow \wt$ such that $X(\wt)\setminus e[X(\wt')]\subseteq f[X(\wt'')]$, then for $a\in e[X(\wt')]\cap  f[X(\wt'')]$, if $e^{-1}(a)\in F(\wt')$ and $F(\wt'')=\{f^{-1}(a)\}$, then $|F(\wt)|\leq |F(\wt')|$.
\end{enumerate}
\end{definition}

The correspondence in Lemma \ref{CanLem} preserves the relevant axioms as follows.

\begin{lemma}\label{CanProbLem}$\,$
\begin{enumerate}
\item If $F$ is a canonical weak tournament solution satisfying the orbit condition and canonical binary $\gamma$ (resp.~canonical $\alpha$-resoluteness, binary quasi-resoluteness), then $F^*$ defined as in Lemma \ref{CanLem}.\ref{CanLem1} is a neutral weak tournament solution satisfying binary $\gamma$ \textnormal{(}resp.~$\alpha$-resoluteness, binary quasi-resoluteness\textnormal{)}.
\item If $F$ is a neutral weak tournament solution satisfying binary $\gamma$ \textit{(}resp.~$\alpha$-resoluteness, binary quasi-resoluteness\textnormal{)}, then $F_*$ defined as in Lemma \ref{CanLem}.\ref{CanLem2} is a canonical weak tournament solution satisfying the orbit condition and canonical binary $\gamma$ \textnormal{(}resp.~canonical $\alpha$-resoluteness, binary quasi-resoluteness\textnormal{)}.
\end{enumerate}
\end{lemma}

\subsection{SAT encoding: formulas}\label{SAT2} We assume familiarity with the basics of propositional logic  \cite[\S~1.1]{Enderton2001}, using  $\neg$ for `not', $\wedge$ for `and', and $\vee$ for `or'.  In this section, we explain how the axioms in Section \ref{SAT1} can be encoded as formulas of propositional logic for the purposes of SAT solving. Following \cite{Brandt2016}, our propositional variables are of the form $A_{\wt, Y}$, where $\wt$ is a canonical weak tournament, $\varnothing\neq Y \subseteq X(\wt)$, and $Y\in \mathscr{O}(\wt)$ (Definition \ref{OrbitDef}). The intended interpretation is that  $A_{\wt_C, Y}$ evaluates to true if and only if $Y$ is assigned as the set of winners of $\wt$. Next we introduce propositional formulas, written in  conjunctive normal form for processing by the SAT solver, which formalize the relevant axioms. 

Fix a finite set $\mathbb{D}$ of weak tournaments. The  following formula $\mathsf{func}_\mathbb{D}$ ensures that the assignment of sets of winners to canonical weak tournaments in $\mathbb{D}$ yields a function:
\[\bigwedge_{\wt\in \mathbb{D}}  \big(\bigvee_{Y\in \mathscr{O}(\wt)} A_{\wt, Y}\big) \wedge \bigwedge_{\wt\in \mathbb{D}}\,  \bigwedge_{Y,Z\in \mathscr{O}(\wt): Y \neq Z} (\lnot A_{\wt, Y} \lor \lnot A_{\wt, Z}).\] 
The first part of the formula (before the small `$\wedge$') says that for each weak tournament $T\in\mathbb{D}$, there is some subset $Y$ of $T$ such that $(T,Y)$ satisfies the orbit condition and $Y$ is assigned as the set of winners of $\wt$. Since $\mathbb{D}$ is finite, this is a finite formula. Similarly, the second part says that  for each weak tournament $T\in\mathbb{D}$ and distinct subsets $Y,Z$ of $T$ such that $(T,Y)$ and $(T,Z)$ both satisfy the orbit condition, either $Y$ is not assigned as the set of winners of $\wt$ or $Z$ is not assigned as the set of winners of $\wt$, i.e., it is not the case that both $Y$ and $Z$ are assigned as the set of winners. The formulas to follow can be read similarly.

To formalize binary quasi-resoluteness, where $\wt_2$ is the canonical tournament with two nodes, say $0$ and $1$, we define $\mathsf{bqr}_\mathbb{D}$ by:
 \[A_{\wt_2, \{0\}}\vee A_{\wt_2, \{1\}}.\]

To formalize binary $\gamma$ and $\alpha$-resoluteness, for any $T\in\mathbb{D}$, define:
\begin{eqnarray*}
E(\wt)&=& \{(\wt',e)\mid e: \wt' \hookrightarrow \wt\mbox{ and } |X(\wt)|=|X(\wt')|+1\} \\
F(\wt) &=& \{(\wt'',f)\mid f: \wt'' \hookrightarrow \wt\mbox{ and }|X(\wt'')|=2\}.
\end{eqnarray*}
Then let $\textsf{bg}_\mathbb{D}$ be the formula:
\[\bigwedge_{\wt\in\mathbb{D}} \,\bigwedge_{(\wt',e) \in E(\wt)}\, \bigwedge_{(\wt'',f) \in F(\wt):\, X(\wt)\setminus e[X(\wt')]\subseteq f[X(\wt'')]} \, \varphi\]
with $\varphi$ defined as follows, where $e[X(\wt')] \cap f[X(\wt'')]=\{a\}$:\footnote{Since $|X(\wt'')|=2$, we have $|f[X(\wt'')]| =2$. Then since $|X(\wt)|=|X(\wt')|+1$, we have $|X(\wt)\setminus e[X(\wt')]|=1$, which with $X(\wt)\setminus e[X(\wt')]\subseteq f[X(\wt'')]$ implies $|e[X(\wt')] \cap f[X(\wt'')]|=1$.}
\[\varphi = \underset{Y''\in \mathscr{O}(\wt''):\, Y''=\{f^{-1}(a)\}}{\bigwedge_{Y'\in \mathscr{O}(\wt'):\, e^{-1}(a)\in Y'}} (\neg A_{\wt', Y'}\vee \neg A_{\wt'', Y''}\vee   \bigvee_{Y\in \mathscr{O}(\wt):\, a\in Y }\, A_{\wt,Y}) .\]
The formula $\textsf{ar}_\mathbb{D}$ is defined in the same way but with the last disjunction replaced by 
\[ \bigvee_{Y\in \mathscr{O}(\wt):\, |Y|\leq |Y'|}\, A_{\wt,Y}.\]

Propositional valuations $v:\{A_{\wt, Y}\mid \wt\in\mathbb{D}, Y\in \mathscr{O}(\wt)\}\to \{0,1\}$ satisfying the formula $\mathsf{func}_\mathbb{D}\wedge \mathsf{bqr}_\mathbb{D}\wedge \mathsf{bg}_\mathbb{D}\wedge\mathsf{ar}_\mathbb{D}$ correspond to canonical weak tournament solutions on the domain $\mathbb{D}$ satisfying the relevant axioms.

\newpage

\begin{lemma}\label{ValCor}$\,$
\begin{enumerate}
\item If a valuation $v$ satisfies $\mathsf{func}_\mathbb{D}\wedge \mathsf{bqr}_\mathbb{D}\wedge \mathsf{bg}_\mathbb{D}\wedge\mathsf{ar}_\mathbb{D}$, then \[\{ (\wt,Y) \mid \wt\in\mathbb{D}, Y\in \mathscr{O}(\wt),  v(A_{\wt, Y})=1 \}\] is a canonical weak tournament solution on $\mathbb{D}$ satisfying the orbit condition, binary quasi-resoluteness, canonical binary $\gamma$, and canonical $\alpha$-resoluteness.
\item Given a canonical weak tournament solution $F$ on $\mathbb{D}$ satisfying the orbit condition, binary quasi-resoluteness, canonical binary $\gamma$, and canonical $\alpha$-resoluteness, the valuation $v$ defined on $\{A_{\wt, Y}\mid \wt\in\mathbb{D}, Y\in \mathscr{O}(\wt)\}$ by $v(A_{\wt, Y})=1$ if and only if $F(\wt)=Y$ satisfies $\mathsf{func}_\mathbb{D}\wedge \mathsf{bqr}_\mathbb{D}\wedge \mathsf{bg}_\mathbb{D}\wedge\mathsf{ar}_\mathbb{D}$.
\end{enumerate}
\end{lemma}

Using a SAT solver and Lemma \ref{ValCor}, we obtain the following. 

\begin{proposition}\label{CanSatThm}$\,$  
\begin{enumerate}
\item\label{SatThm02} For any finite $Y\subseteq\mathcal{X}$ with $|Y|\geq 6$, is no canonical tournament solution satisfying the orbit condition, binary quasi-resoluteness, canonical binary $\gamma$, and canonical $\alpha$-resoluteness on the domain $\{\wt\in \mathbb{T}_C\mid X(\wt)\subseteq Y\}$.
\item\label{SatThm01} There is a canonical weak tournament solution satisfying the orbit condition, binary quasi-resoluteness, canonical binary $\gamma$, and canonical $\alpha$-resoluteness on the domain $\{\wt\in \mathbb{WT}_{\!\!C}\mid |\wt|\leq 5\}$.
\end{enumerate}
\end{proposition}

Finally, combining Proposition \ref{CanSatThm} and Lemma \ref{CanProbLem} establishes the results for weak tournament solutions in Proposition \ref{SatThm0} from Section \ref{SATResults}.

\subsection{From tournament solutions to majoritarian voting methods}\label{TournamentTomajoritarian}

Having obtained Proposition \ref{SatThm0} for (weak) tournament solutions, it is straightforward using McGarvey's theorem (Theorem \ref{McGarveyThm}) to obtain an analogous statement for majoritarian voting methods. 

\fromsatthmone*

\begin{proof} Suppose for contradiction there is such a method $F$.  We define a canonical tournament solution $G$ on the domain $\{\wt\in \mathbb{T}_C\mid X(\wt)\subseteq Y\}$ as follows: for $T\in \mathrm{dom}(G)$, let $G(T)= F(\mathsf{Mc}(T))$. Since $F$ is majoritarian and satisfies the relevant axioms, it is not difficult to show that $G$ satisfies the tournament solution versions of the relevant axioms, contradicting Proposition \ref{SatThm0}.\ref{SatThm02}.
\end{proof}

\begin{proposition} There is an anonymous, neutral, and majoritarian voting method satisfying binary quasi-resoluteness, binary $\gamma$, and $\alpha$-resoluteness on the domain $\{\mathbf{P}\mid |X(\mathbf{P})|\leq 5\}$.
\end{proposition}
\begin{proof} Define the method $F$ as follows: where $G$ is the weak tournament solution from Proposition \ref{SatThm0}.\ref{SatThm01}, given a profile $\mathbf{P}$, let $F(\mathbf{P})= G(M(\mathbf{P}))$.\end{proof}

\subsection{Prospects for a human-readable proof}\label{HumanReadable1}

After obtaining Proposition \ref{CanSatThm}.\ref{SatThm02} using SAT solving, we investigated the prospects of a human-readable proof. First, a canonical tournament solution satisfying canonical $\alpha$-resoluteness either (i) picks only the majority winner in the two-candidate canonical tournament $T_2$ or (ii) picks only the minority winner in $T_2$. Thus, we can simply prove the impossibility assuming (i), as the proof in case (ii) then follows by reversing all edges in tournaments. With this simplification, the CNF formula used for Proposition \ref{CanSatThm}.\ref{SatThm02} contains 132,626 conjuncts.  To obtain an MUS (recall Section~\ref{ProofStrategy}), following Brandt and Geist \cite{Brandt2016}, we used the PicoMUS tool from PicoSAT \cite{Biere2008}, which produced an MUS with 266 conjuncts. By analyzing the MUS, we identified \textit{nine} canonical tournaments ranging from size two to six---down from 75 canonical tournaments total from size two to six, as shown in Table \ref{NumCanTourns}---such that no canonical tournament solution on the domain of those nine canonical tournaments satisfies the relevant axioms. Nine canonical tournaments initially sounds promising for a human readable proof, but the conjuncts in the MUS contain 20 $\textsf{bg}$ conjuncts and 30 $\textsf{ar}$ conjuncts, involving a total of 19 embeddings between canonical tournaments. Naively, this means a human-readable proof might contain 20 appeals to binary $\gamma$ and  30 appeals to $\alpha$-resoluteness, involving the consideration of 19 ways a smaller tournament can embed into a larger one. Of course, it is possible that several such steps could be handled at once. It is also possible that there is a smaller MUS more amenable to use for a human-readable proof.\footnote{Following Brandt and Geist \cite{Brandt2016}, we tried applying the CAMUS tool \cite{Liffiton2008} to find a smallest MUS, but it did not terminate in a reasonable amount of time.} 

\begin{table}
\begin{center}
\begin{tabular}{c|c|c}
candidates & canonical   & canonical   \\
 & tournaments from  & tournaments from   \\
 &initial set & MUS \\
 \hline 
2 & 1  & 1  \\
3 & 2 & 1   \\
4 & 4 & 2  \\
5 & 12 & 3  \\
6 &  56 & 2  \\
2 to 6 & 75 & 9   
\end{tabular}
\end{center}
\caption{Number of canonical tournaments (i) from size two to six (second column) and (ii) sufficient for an impossibility theorem obtained from an MUS (third column).}\label{NumCanTourns}
\end{table}

\section{SAT solving for pairwise methods}\label{SATQM}

\subsection{Results}

Using SAT we also obtained an impossibility theorem for the conjunction of binary $\gamma$, quasi-resoluteness, and pairwiseness, assuming anonymity and neutrality. However, the relevant SAT approach works not with profiles and pairwise voting methods but with weighted tournaments and weighted tournament solutions.

\begin{definition} A \textit{weighted tournament solution} is a function $F$ on a set of weighted tournaments such that for all $\wwt\in\mathrm{dom}(F)$, we have $\varnothing\neq F(\wwt)\subseteq X(\wwt)$, where $X(\wwt)$ is the set of nodes in the underlying tournament of $\mathcal{T}$.
\end{definition}

The binary $\gamma$ and quasi-resoluteness axioms from Section \ref{IndividualToSocial} apply in the obvious ways to weighted tournament solutions.

\begin{definition} Let $F$ be a weighted tournament solution.
\begin{enumerate}
\item $F$ satisfies \textit{binary $\gamma$} if for every $\wwt\in \mathrm{dom}(F)$ and $a,b\in X(\wwt)$, if $a\in F(\wwt_{-b})$ and $F(\wwt|_{ \{a,b\}})=\{a\}$, then $a\in F(\wwt)$.
\item $F$ satisfies \textit{quasi-resoluteness} if for every uniquely weighted $\wwt\in\mathrm{dom}(F)$, we have ${|F(\wwt)|=1}$.
\end{enumerate}
\end{definition}

The following result explains the phenomenon noted in Section \ref{IndividualToSocial} that no known (anonymous and neutral) pairwise voting method satisfies binary $\gamma$ and quasi-resoluteness. Note that since the SAT formalization requires a finite domain, we fix a finite set of weights for our weighted tournaments.

\begin{proposition}\label{OTImpossibility} There is no neutral weighted tournament solution satisfying quasi-resoluteness and binary $\gamma$ on the domain \begin{eqnarray*}&&\{\wwt\in\mathbb{T}^w\mid |X(\wwt)|\leq 4\mbox{, $\wwt$ uniquely weighted, }\\
&&\;\;\qquad\qquad\mbox{and all positive weights belong to } \{2,4,6,8,10,12\}\}.\end{eqnarray*}
\end{proposition}

\subsection{SAT encoding}

The SAT encoding strategy for weighted tournament solutions follows the strategy for tournament solutions in Sections \ref{SAT1} and \ref{SAT2}: for computational feasibility, we work with weighted tournament solutions defined on a smaller domain of \textit{canonical} weighted tournaments. In fact, we restrict our domain $\mathbb{D}$ to only those canonical weighted tournaments in which all weights are distinct. Moreover, we build quasi-resoluteness into our encoding, by taking only propositional variables of the form $A_{\wwt,Y}$ where $\wwt\in\mathbb{D}$  and $Y$ is a singleton. Note that for $\mathcal{T}\in\mathbb{D}$, there are no non-trivial automorphisms of $\mathcal{T}$, since all weights are distinct, so we no longer need to check the orbit condition. Functionality and binary $\gamma$ are formalized as formulas $\mathsf{func}_\mathbb{D}$ and $\mathsf{bg}_\mathbb{D}$ as in Section \ref{SAT2} but using canonical weighted tournaments $\wwt$ instead of canonical tournaments $T$, setting  $\mathscr{O}(T)=\{\{x\}\mid x\in X(\mathcal{T})\}$, and using weighted tournament embeddings\footnote{I.e., tournament embeddings $f:T\to T'$ for which $w(x,y)=w'(f(x),f(y))$, where $\mathcal{T}=(T,w)$ and $\mathcal{T}'=(T',w')$.} instead of tournament embeddings for binary $\gamma$. The unsatisfiability of $\mathsf{func}_\mathbb{D}\wedge\mathsf{bg}_\mathbb{D}$ according to the SAT solver yields Proposition \ref{OTImpossibility} by reasoning analogous to that in Section~\ref{SAT1}.

\subsection{From weighted tournament solutions to pairwise voting methods}\label{PairwiseImpossibility}

It is a short step to transfer the impossibility result in Proposition \ref{OTImpossibility} from weighted tournament solutions to pairwise voting methods in Proposition \ref{FromQMFinal}. The proof of Proposition \ref{FromQMFinal} is analogous to that of Proposition \ref{FromSatThm1}, except that we use Debord's representation of weighted tournaments instead of McGarvey's representation of tournaments. 

\fromqmfinal*

Recall that Proposition \ref{FromQMFinal} together with the transfer lemma, Lemma \ref{transferlemtwo}, yields the desired impossibility result in Theorem \ref{Impossibility4}.

\subsection{Prospects for a human-readable proof}\label{HumanReadable2}

After obtaining Proposition \ref{OTImpossibility} using the SAT solving strategy, we investigated the prospects of a human-readable proof.  There are 2,086 canonical uniquely weighted tournaments from size two to four with edge weights from $\{2,4,6,8,10,12\}$ (see the second column of Table \ref{NumCanTourns2}), yielding 8,172 propositional variables. Then our CNF formula used to prove Proposition \ref{OTImpossibility} contains 38,092 conjuncts. To obtain an MUS (recall Sections \ref{ProofStrategy} and \ref{HumanReadable1}), we again used PicoMUS, which produced an MUS with 9,294 conjuncts.\footnote{By contrast, in Brandt and Geist's \cite{Brandt2016} application of SAT to an impossibility theorem in social choice, PicoMUS returned an MUS with 55 conjuncts from an initial formula with over three million.} By analyzing the MUS, we identified 1,650 canonical weighted tournaments (see the third column of Table \ref{NumCanTourns2}) from the original 2,086 such that no canonical tournament solution on the domain of those 1,650 satisfies the relevant axioms. In fact, we were able to do better by a naive approach. In the same spirit as MUS extraction, we tried to find a smallest subset of the 2,086 canonical weighted tournaments for which SAT still produces an impossibility result, simply by iterating over these weighted tournaments (ordered by descending number of candidates and in an arbitrary order for weighted tournaments with the same number of candidates) and retaining a weighted tournament if and only if removing it leads to a satisfiability result.\footnote{We also tried this approach for Proposition \ref{CanSatThm}.\ref{SatThm02}, but in that case the naive approach yielded more canonical tournaments than the approach based on PicoMUS.} The smallest set found so far contains 716 canonical weighted tournaments (see the fourth column of Table \ref{NumCanTourns2}), which is unfortunately still too many for human inspection. Recall, by contrast, that in our proof assuming homogeneity and cancellativity for Proposition \ref{Impossibility4b}, we only had to consider three 4-candidate margin graphs, along with some of their subgraphs.

\begin{table}
\begin{center}
\begin{tabular}{c|c|c|c}
candidates & canonical weighted  & canonical weighted &  canonical weighted \\
 & tournaments from  & tournaments from & tournaments from   \\
 &initial set & MUS & naive approach\\
 \hline 
2 & 6 & 6 & 6 \\
3 & 160 & 160 & 140 \\
4 & 1,920 & 1,484 & 570\\
2 to 4 & 2,086 & 1,650 & 716
\end{tabular}
\end{center}
\caption{Number of canonical uniquely weighted tournaments (i) from size two to four with edges weights from $\{2,4,6,8,10,12\}$ (second column), (ii) sufficient for an impossibility theorem obtained from an MUS (third column), and (iii) sufficient for an impossibility theorem obtained from the naive approach described in the main text (fourth column).}\label{NumCanTourns2}
\end{table}

\section{Lean verification of SAT encoding}\label{Verification}

To increase confidence in the SAT-based Proposition \ref{OTImpossibility}, we produced a formally verified SAT encoding in the Lean Theorem Prover (Lean 3), building on recent work of Codel \cite{Codel2022}. The Lean verification is available at \href{https://github.com/chasenorman/verified-encodings-social-choice}{https://github.com/chasenorman/verified-encodings-social-choice}. In this appendix, we explain what exactly we verified. A similar approach could be used for Proposition \ref{CanSatThm}.\ref{SatThm02}, but we concentrated our efforts on Proposition \ref{OTImpossibility}.

Let $\mathbb{G}$ be the set of uniquely weighted tournaments $T=(X(T),w)$ such that $X(T)=\{0,\dots,n\}$ for $n\in \{0,\dots,3\}$ and for each $a,b\in X(T)$ with $a\neq b$, we have $|w(a,b)|\in \{2,4,6,8,10,12\}$. Given $T\in \mathbb{G}$ and $a,b\in X(T)$, let $T_{a,b}$ be the weighted tournament with $X(T_{a,b})=\{0,1\}$ and $w(0,1)= |w(a,b)|$.

To prove Proposition  \ref{OTImpossibility}, it suffices to show that for some set $\mathcal{E}$ of weighted tournament embeddings between elements of $\mathbb{G}$, there is no resolute weighted tournament solution satisfying the following condition of  \textit{binary $\gamma$ with respect to $\mathcal{E}$}.

\begin{definition} Let  $F$ be a function on $\mathbb{G}$ such that for each $T\in\mathbb{G}$, $F(T)\in X(T)$. Let $\mathcal{E}$ be a family of maps $e:T\hookrightarrow T'$ where $T,T'\in \mathbb{G}$. Then $F$ satisfies \textit{binary $\gamma$ with respect to $\mathcal{E}$} if for any $\wt,\wt'\in \mathbb{G}$ and map $e\in \mathcal{E}$ from $\wt'$ to $\wt$, if
\begin{enumerate} 
\item\label{AddOne} $|X(\wt)| = |X(\wt')| + 1$,
\item\label{BeforeB}  $a\in e[X(\wt')]$, 
\item\label{Newcomer} $b\in X(\wt)\setminus e[X(\wt')]$, 
\item\label{WinsWithoutB} $F(\wt')=e^{-1}(a)$,  and
\item\label{WinsHeadToHead} $F(T_{a,b})=0$ if $w(a,b)>0$ and $F(T_{a,b})=1$  otherwise, 
\end{enumerate}
 then $F(\wt)=a$.
\end{definition}
Modulo the switching between images and pre-images of the embedding, the items of the definition capture the assumptions of binary $\gamma$, namely that we are considering removing from a situation $T$ one candidate to obtain $T'$, as in (\ref{AddOne}), where $a$ is one of the candidates not removed, as in (\ref{BeforeB}), whereas $b$ is the candidate removed, as in (\ref{Newcomer}), such that $a$ wins without $b$, as in (\ref{WinsWithoutB}), and $a$ wins against $b$ head-to-head, as in (\ref{WinsHeadToHead}). Note that to capture the condition that $a$ wins head-to-head against $b$, we check in (\ref{WinsHeadToHead}) whether the counterpart of $a$ in $T_{a,b}$ wins.

The repository linked above contains Lean code that for a given $\mathbb{D}\subseteq\mathbb{G}$ and set $\mathcal{E}$ of maps generates a propositional formula $\mathsf{func}_\mathbb{D}\wedge \mathsf{bg}_{\mathcal{E}}$ in CNF. We then give a formal proof in Lean of the following result, showing that  $\mathsf{func}_\mathbb{D}\wedge \mathsf{bg}_{\mathcal{E}}$ correctly encodes the existence of a function  on $\mathbb{D}$ satisfying binary $\gamma$ with respect to $\mathcal{E}$.

\begin{lemma}\label{KeyLem} Let $F$ be a function on $\mathbb{G}$ such that for each $T\in \mathbb{G}$, $F(T)\in X(T)$. Let $\mathbb{D}\subseteq \mathbb{G}$, and let $\mathcal{E}$ be a family of maps $e:T\hookrightarrow T'$ where $T,T'\in \mathbb{G}$. If $F$ satisfies  binary $\gamma$ with respect to $\mathcal{E}$, then the valuation $v$ defined on \[\{A_{T, c}\mid T\in\mathbb{D}, c\in X(T)\}\] by 
\[\mbox{$v(A_{T, c})=1$ if and only if $F(T)=c$}\] 
satisfies $\mathsf{func}_\mathbb{D}\wedge \mathsf{bg}_{\mathcal{E}}$.\end{lemma}

We then defined in Lean a set $\mathbb{D}$   and a set $\mathcal{E}$ of maps for which we formally verified the following:
\begin{itemize}
\item[(i)] $\mathbb{D}$ is indeed a subset of $\mathbb{G}$;
\item[(ii)] each map in $\mathcal{E}$ is indeed a weighted tournament embedding between elements of $\mathbb{G}$.
\end{itemize}
In fact, $\mathbb{D}$ contains the 716 canonical weighted tournaments obtained as described in Section~\ref{HumanReadable2}, and $\mathcal{E}$ was obtained using Python code to find a minimal set of embeddings sufficient for the impossibility theorem. However, the provenance of $\mathbb{D}$ and $\mathcal{E}$ are irrelevant, thanks to the Lean verification of (i) and (ii). There is no need to trust our Python code.

The final step was to test the satisfiability of the formula $\mathsf{func}_\mathbb{D}\wedge \mathsf{bg}_{\mathcal{E}}$ produced by Lean using a SAT solver. As Lean does not include a SAT solver (as of this writing), we had to leave the Lean environment for this last step. Standard SAT solvers returned unsatisfiability for the formula $\mathsf{func}_\mathbb{D}\wedge \mathsf{bg}_{\mathcal{E}}$. Thus,  in order to trust our SAT-based proof of Proposition \ref{OTImpossibility}, one need only trust that it is not the case that multiple standard SAT solvers contain bugs producing an incorrect unsatisfiability result for our formula.

\bibliographystyle{amsplain}
\bibliography{extending}

\providecommand{\bysame}{\leavevmode\hbox to3em{\hrulefill}\thinspace}
\providecommand{\MR}{\relax\ifhmode\unskip\space\fi MR }
\providecommand{\MRhref}[2]{%
  \href{http://www.ams.org/mathscinet-getitem?mr=#1}{#2}
}
\providecommand{\href}[2]{#2}
\begin{thebibliography}{10}

\bibitem{Arrow1963}
Kenneth~J. Arrow, \emph{Social choice and individual values}, 2nd ed., John
  Wiley \& Sons, Inc., New York, 1963.

\bibitem{Baldwin1926}
J.~M. Baldwin, \emph{A technique of the {N}anson preferential majority system
  of election}, Transactions and Proceedings of the Royal Society of Victoria
  \textbf{39} (1926), 45--52.

\bibitem{Banks1985}
J.~S. Banks, \emph{Sophisticated voting outcomes and agenda control}, Social
  Choice and Welfare \textbf{1} (1985), no.~4, 295--306.

\bibitem{Biere2008}
Armin Biere, \emph{{PicoSAT} essentials}, Journal on Satisfiability, Boolean
  Modeling and Computation \textbf{4} (2008), no.~2-4, 75--79.

\bibitem{Bordes1983}
Georges Bordes, \emph{On the possibility of reasonable consistent majoritarian
  choice: {S}ome positive results}, Journal of Economic Theory \textbf{31}
  (1983), no.~1, 122--132.

\bibitem{Brandl2018}
Florian Brandl, Felix Brandt, Manuel Eberl, and Christian Geist, \emph{Proving
  the incompatibility of efficiency and strategyproofness via {SMT} solving},
  Journal of the ACM \textbf{65} (2018), no.~2, 6:1--6:28.

\bibitem{Brandl2019}
Florian Brandl, Felix Brandt, Christian Geist, and Johannes Hofbauer,
  \emph{Strategic abstention based on preference extensions: {P}ositive results
  and computer-generated impossibilities}, Journal of Artificial Intelligence
  Research \textbf{66} (2019), 1031--1056.

\bibitem{Brandt2022}
Felix Brandt and Chris Dong, \emph{Local rationalizability and choice
  consistency}, \href{https://arxiv.org/abs/2204.05062}{arXiv:2204.05062
  [econ.TH]}, 2022.

\bibitem{Brandt2007}
Felix Brandt and Felix Fischer, \emph{Pagerank as a weak tournament solution},
  Internet and Network Economics: Third International Workshop, WINE 2007
  (Berlin) (Xiaotie Deng and Fan~Chung Graham, eds.), Springer-Verlag, 2007,
  pp.~300--305.

\bibitem{Brandt2016}
Felix Brandt and Christian Geist, \emph{Finding strategyproof social choice
  functions via {SAT} solving}, Journal of Artificial Intelligence Research
  \textbf{55} (2016), 565--602.

\bibitem{Brandt2017}
Felix Brandt, Christian Geist, and Dominik Peters, \emph{Optimal bounds for the
  no-show paradox via {SAT} solving}, Mathematical Social Sciences \textbf{90}
  (2017), 18--27.

\bibitem{Brandt2014}
Felix Brandt, Christian Geist, and Hans~Georg Seedig, \emph{Identifying
  k-majority digraphs via {SAT} solving}, Proceedings of the 1st AAMAS Workshop
  on Exploring Beyond the Worst Case in Computational Social Choice, 2014.

\bibitem{Brandt2011}
Felix Brandt and Paul Harrenstein, \emph{Set-rationalizable choice and
  self-stability}, Journal of Economic Theory \textbf{146} (2011), no.~4,
  1721--1731.

\bibitem{Brandt2021}
Felix Brandt, Marie Matth\"{a}us, and Christian Saile, \emph{Minimal voting
  paradoxes}, Working paper, 2021.

\bibitem{Brandt2018}
Felix Brandt, Christian Saile, and Christian Stricker, \emph{Voting with ties:
  Strong impossibilities via {SAT} solving}, Proceedings of the 17th
  International Conference on Autonomous Agents and Multiagent Systems (AAMAS
  2018) (M.~Dastani, G.~Sukthankar, E.~Andr\'{e}, and S.~Koenig, eds.),
  International Foundation for Autonomous Agents and Multiagent Systems, 2018,
  pp.~1285--1293.

\bibitem{Chang1990}
C.~C. Chang and H.~Jerome Keisler, \emph{Model theory}, North-Holland,
  Amsterdam, 1990.

\bibitem{Codel2022}
Cayden Codel, \emph{Verifying {SAT} encodings in {L}ean}, 2022.

\bibitem{Coombs1964}
Clyde~Hamilton Coombs, \emph{A theory of data}, John Wiley and Sons, New York,
  1964.

\bibitem{Copeland1951}
A.~H. Copeland, \emph{A `reasonable' social welfare function}, Notes from a
  seminar on applications of mathematics to the social sciences, University of
  Michigan, 1951.

\bibitem{Debord1987}
Bernard Debord, \emph{Caract\'{e}risation des matrices des pr\'{e}f\'{e}rences
  nettes et m\'{e}thodes d'agr\'{e}gation associ\'{e}es}, Math\'{e}matiques et
  sciences humaines \textbf{97} (1987), 5--17.

\bibitem{Duggan2013}
John Duggan, \emph{Uncovered sets}, Social Choice and Welfare \textbf{41}
  (2013), no.~3, 489--535.

\bibitem{Duggan2019}
\bysame, \emph{Weak rationalizability and arrovian impossibility theorems for
  responsive social choice}, Public Choice \textbf{179} (2019), 7--40.

\bibitem{Dutta1999}
Bhaskar Dutta and Jean-Francois Laslier, \emph{Comparison functions and choice
  correspondences}, Social Choice and Welfare \textbf{16} (1999), 513--532.

\bibitem{Enderton2001}
Herbert~B. Enderton, \emph{A mathematical introduction to logic}, Harcourt
  Academic Press, 2001.

\bibitem{Geist2011}
Christian Geist and Ulle Endriss, \emph{Automated search for impossibility
  theorems in social choice theory: {R}anking sets of objects}, Journal of
  Artificial Intelligence Research \textbf{40} (2011), 143--174.

\bibitem{Geist2017}
Christian Geist and Dominik Peters, \emph{Computer-aided methods for social
  choice theory}, Trends in Computational Social Choice (Ulle Endriss, ed.), AI
  Access, 2017, pp.~249--267.

\bibitem{Grofman2004}
Bernard Grofman and Scott~L. Feld, \emph{If you like the alternative vote
  (a.k.a. the instant runoff), then you ought to know about the coombs rule},
  Electoral Studies \textbf{23} (2004), no.~4, 641--659.

\bibitem{Herron2007}
Michael~C. Herron and Jeffrey~B. Lewis, \emph{Did {R}alph {N}ader spoil a
  {G}ore presidency? a ballot-level study of {G}reen and {R}eform {P}arty
  voters in the 2000 presidential election}, Quarterly Journal of Political
  Science \textbf{2} (2007), no.~3, 205--226.

\bibitem{Hoag1926}
Clarence Hoag and George Hallett, \emph{Proportional representation},
  Macmillan, New York, 1926.

\bibitem{HP2020b}
Wesley~H. Holliday and Eric Pacuit, \emph{Split {C}ycle: {A} new {C}ondorcet
  consistent voting method independent of clones and immune to spoilers},
  \href{https://arxiv.org/abs/2004.02350}{arXiv:2004.02350}, 2020.

\bibitem{HP2021}
\bysame, \emph{Axioms for defeat in democratic elections}, Journal of
  Theoretical Politics \textbf{33} (2021), no.~4, 475--524,
  \href{https://arxiv.org/abs/2008.08451}{arXiv:2008.08451}.

\bibitem{HP2022}
\bysame, \emph{Stable voting},
  \href{https://arxiv.org/abs/2108.00542}{arXiv:2108.00542}, 2022.

\bibitem{Horan2022}
Sean Horan and Yves Sprumont, \emph{Two-stage majoritarian choice}, Theoretical
  Economics \textbf{17} (2922), no.~2, 521--537.

\bibitem{Kluiving2020}
Boas Kluiving, Adriaan de~Vries, Pepijn Vrijbergen, Arthur Boixel, and Ulle
  Endriss, \emph{Analysing irresolute multiwinner voting rules with approval
  ballots via sat solving}, Proceedings of the 24th European Conference on
  Artificial Intelligence (ECAI-2020), August 2020, pp.~131--138.

\bibitem{Kramer1977}
Gerald~H. Kramer, \emph{A dynamical model of political equilibrium}, Journal of
  Economic Theory \textbf{16} (1977), no.~2, 310--334.

\bibitem{Liffiton2008}
Mark~H. Liffiton and Karem~A. Sakallah, \emph{Algorithms for computing minimal
  unsatisfiable subsets of constraints}, Journal of Automated Reasoning
  \textbf{40} (2008), 1--33.

\bibitem{Magee2003}
Christopher S.~P. Magee, \emph{Third-party candidates and the 2000 presidential
  election}, Social Science Quarterly \textbf{84} (2003), no.~3, 29--35.

\bibitem{McGarvey1953}
David~C. McGarvey, \emph{A theorem on the construction of voting paradoxes},
  Econometrica \textbf{21} (1953), no.~4, 608--610.

\bibitem{Miller1980}
Nicholas~R. Miller, \emph{A new solution set for tournaments and majority
  voting: {F}urther graph-theoretical approaches to the theory of voting},
  American Journal of Political Science \textbf{24} (1980), no.~1, 68--96.

\bibitem{Nanson1882}
E.~J. Nanson, \emph{Methods of election}, Transactions and Proceedings of the
  Royal Society of Victoria \textbf{19} (1882), 197--240.

\bibitem{Niou1987}
Emerson M.~S. Niou, \emph{A note on {N}anson's rule}, Public Choice \textbf{54}
  (1987), no.~2, 191--193.

\bibitem{Nurmi1987}
Hannu Nurmi, \emph{Comparing voting systems}, D. Reidel, Dordrecht, 1987.

\bibitem{Egecioglu2009}
\"{O}mer E\v{g}ecio\v{g}lu, \emph{Uniform generation of anonymous and neutral
  preference profiles for social choice rules}, Monte Carlo Methods and
  Applications \textbf{15} (2009), no.~3, 241--255.

\bibitem{Egecioglu2013}
\"{O}mer E\v{g}ecio\v{g}lu and Ay\c{c}a~E. Giritligil, \emph{The impartial,
  anonymous, and neutral culture model: a probability model for sampling public
  preference structures}, Journal of Mathematical Sociology \textbf{37} (2013),
  203--222.

\bibitem{Fernandez2018}
Ra\'{u}l P\'{e}rez-Fern\'{a}ndez and Bernard {De Baets}, \emph{The
  supercovering relation, the pairwise winner, and more missing links between
  {B}orda and {C}ondorcet}, Social Choice and Welfare \textbf{50} (2018),
  329--352.

\bibitem{Saari2008}
Donald~G. Saari, \emph{Disposing dictators, demystifying voting paradoxes},
  Cambridge University Press, Cambridge, 2008.

\bibitem{Schulze2011}
Markus Schulze, \emph{A new monotonic, clone-independent, reversal symmetric,
  and condorcet-consistent single-winner election method}, Social Choice and
  Welfare \textbf{36} (2011), 267--303.

\bibitem{Schulze2018}
\bysame, \emph{The {S}chulze method of voting},
  \href{https://arxiv.org/abs/1804.02973}{arXiv:1804.02973}, 2018.

\bibitem{Schwartz1986}
Thomas Schwartz, \emph{The logic of collective choice}, Columbia University
  Press, New York, 1986.

\bibitem{Sen1971}
Amartya Sen, \emph{Choice functions and revealed preference}, The Review of
  Economic Studies \textbf{38} (1971), no.~3, 307--317.

\bibitem{Simpson1969}
Paul~B. Simpson, \emph{On defining areas of voter choice: {P}rofessor {T}ullock
  on stable voting}, The Quarterly Journal of Economics \textbf{83} (1969),
  no.~3, 478--490.

\bibitem{Smith1973}
John~H. Smith, \emph{Aggregation of preferences with variable electorate},
  Econometrica \textbf{41} (1973), no.~6, 1027--1041.

\bibitem{Stearns1959}
Richard Stearns, \emph{The voting problem}, The American Mathematical Monthly
  \textbf{66} (1959), no.~9, 761--763.

\bibitem{Suzumura1983}
Kotaro Suzumura, \emph{Rational choice, collective decisions, and social
  welfare}, Cambridge University Press, Cambridge, 1983.

\bibitem{Tideman1987}
T.~Nicolaus Tideman, \emph{Independence of clones as a criterion for voting
  rules}, Social Choice and Welfare \textbf{4} (1987), 185--206.

\bibitem{Tyson2008}
Christopher~J. Tyson, \emph{Cognitive constraints, contraction consistency, and
  the satisficing criterion}, Journal of Economic Theory \textbf{138} (2008),
  51--70.

\bibitem{Wilson1972}
Robert Wilson, \emph{Social choice theory without the {P}areto principle},
  Journal of Economic Theory \textbf{5} (1972), 478--486.

\bibitem{ZavistTideman1989}
T.~M. Zavist and T.~Nicolaus Tideman, \emph{Complete independence of clones in
  the ranked pairs rule}, Social Choice and Welfare \textbf{6} (1989),
  167--173.

\end{thebibliography}

\end{document}